\documentclass[11pt, a4paper]{article}
\usepackage[utf8]{inputenc}
\usepackage[T1]{fontenc}
\usepackage{fullpage}
\usepackage{hyperref}
\usepackage{amsthm}
\usepackage{amssymb}
\usepackage{amsmath}
\usepackage{enumerate}
\usepackage{graphicx} 
\usepackage{subcaption}
\usepackage{xcolor}
\usepackage{authblk}
\usepackage[authoryear,round]{natbib}
\usepackage{multirow}
\usepackage[font=small,labelfont=it,textfont=it]{caption}
\usepackage{booktabs}
\usepackage{float}

\newtheorem{defi}{Definition}
\newtheorem{thm}{Theorem}
\newtheorem{prop}{Proposition}
\newtheorem{rk}{Remark}

\newtheorem{ex}{Example}
\newtheorem{nota}{Notation}

\DeclareMathOperator{\Span}{span}
\providecommand{\keywords}[1]
{
  \small	
  \textbf{Keywords:} #1
}
\numberwithin{equation}{section}
\numberwithin{defi}{section}
\numberwithin{thm}{section}
\numberwithin{prop}{section}
\numberwithin{rk}{section}
\numberwithin{lem}{section}
\numberwithin{coro}{section}
\numberwithin{ex}{section}
\numberwithin{nota}{section}

\title{Signature-based validation of real-world economic scenarios}
\author[1,2]{Hervé Andrès \thanks{Corresponding author: herve.andres@milliman.com }}
\author[1]{Alexandre Boumezoued}
\author[2]{Benjamin Jourdain}
\affil[1]{Milliman R\&D, Paris, France}
\affil[2]{CERMICS, Ecole des Ponts, INRIA, Marne-la-Vallée, France.}
\date{\today}
\begin{document}

\makeatletter  
\renewcommand{\@seccntformat}[1]{\csname the#1\endcsname .\quad}
\makeatother
\maketitle
\begin{abstract}
Motivated by insurance applications, we propose a new approach for the validation of real-world economic scenarios. This approach is based on the statistical test developed by \citet{chevyrev2018signature} and relies on the notions of signature and maximum mean distance. This test allows to check whether two samples of stochastic processes paths come from the same distribution. Our contribution is to apply this test to a variety of stochastic processes exhibiting different pathwise properties (Hölder regularity, autocorrelation, regime switches) and which are relevant for the modelling of stock prices and stock volatility as well as of inflation in view of actuarial applications. 

\end{abstract}
\hspace{10pt}

\keywords{real-world economic scenarios, economic scenarios validation, insurance, signature, maximum mean distance.}

\section{Introduction}
Real-world economic scenarios provide stochastic forecasts of economic variables like interest rates, equity returns, inflation, etc. and are widely used in the insurance sector for a variety of applications including asset and liability management (ALM) studies, strategic asset allocation, computing the Solvency Capital Requirement (SCR) within an Internal Model or pricing assets or liabilities including a risk premium. Unlike risk-neutral economic scenarios that are used for market consistent pricing, real-world economic scenarios aim at being realistic in view of the historical data and/or expert expectations about future outcomes. In the literature, many real-world models have been studied for applications in insurance. Those applications relate to (i) valuation of insurance products, (ii) hedging strategies for annuity portfolios and (iii) risk calculation for economic capital assessment. On item (i), we can mention the work of \citet{boudreault2009} who study the impact on Conditional Tail Expectation provision of GARCH and regime-switching models calibrated on historical data, and the work of \citet{graf2014} who perform simulations under the real-world probability measure to estimate the risk-return profile of various old-age provision products. On item (ii), \citet{zhu2018} measure the hedging error of several dynamic hedging strategies along real-world scenarios for cash balance pension plans while \citet{lin2020} calculate the value of a large variable annuity portfolio and its hedge using nested simulations (real-world scenarios for the outer simulations and risk-neutral scenarios for the inner simulations). Finally, on item (iii), \citet{hardy2006} compare several real-world models for the equity return process in terms of fitting quality and resulting capital requirements and discuss the problem of the validation of real-world scenarios. Similarly, \citet{otero2012estimating} measure the impact on the Solvency II capital requirements (SCR) when using a regime-switching model in comparison to lognormal, GARCH and E-GARCH models. \citet{floryszczak2019} introduce a simple model for equity returns allowing to avoid over-assessment of the SCR specifically after market disruptions. On the other hand, \citet{asadi2020} propose a more complex model for stocks based on ARMA and GARCH processes that results in a higher SCR than in the Solvency II standard model. This literature shows the importance of real-world economic scenarios in various applications in insurance. We observe that the question of the consistency of the generated real-world scenarios is barely discussed or only from a specific angle such as the model likelihood or the ability of the model to reproduce the 1 in 200 worst shock observed on the market.\\

In the insurance industry, the assessment of the realism of real-world economic scenarios is often referred to as scenario validation. It allows to verify a posteriori the consistency of a given set of real-world economic scenarios with historical data and/or expert views. As such, it also guides which models can better be used to generate real-world economic scenarios. In the risk-neutral framework, the validation step consists for example in verifying the martingale property of the discounted values along each scenario. In the real-world framework, the most widespread practice is to perform a so-called point-in-time validation. It consists in analyzing the distribution of some variables derived from the generated scenarios (for example annual log-returns for equity stocks or relative variation for an inflation index) at some specific horizons like one year which is the horizon considered in the Solvency II directive. Generally, this analysis only focuses on the first moments of the one-year distribution as real-world models are often calibrated by a moment-matching approach. The main drawback of this approach is that it only allows to capture properties of the simulated scenarios at some point in time. In particular, the consistency of the paths between $t=0$ and $t=1$ year is not studied so that properties like clustering, smoothness, high-order autocorrelation, etc. are not captured. Capturing these properties has its importance as their presence or absence in the economic scenarios can have an impact in real practice, such as for example on a strategic asset allocation having a monthly rebalancing frequency or on the SCR calculation when a daily hedging strategy is involved, since the yearly loss distribution will be path-dependent. In this paper, we propose to address this drawback by comparing the distribution of the stochastic process underlying the simulated paths to the distribution of the historical paths. This can be done using a distance between probability measures, called the Maximum Mean Distance (MMD), and a mathematical object, called the signature, allowing to encode a continuous path in an efficient and parsimonious way. Based on these tools, \citet{chevyrev2018signature} designed a statistical test allowing to accept or reject the hypothesis that the distributions of two samples of paths are equal. This test has already been used by \citet{buehler2020generating} to test whether financial paths generated by a Conditional Variational Auto Encoder (CVAE) are close to the historical paths being used to train the CVAE. An alternative way to compare the distributions of two sample of paths is to flatten each sequence of observations into a long vector of length $d\times L$, where $L$ is the length of the sequence of observations and $d$ is the dimension of each observation, and to apply a multi-variate statistical test. However, Chevyrev and Oberhauser have shown that their signature-based test performs overall better (both in terms of statistical power and in terms of computational cost) than standard multi-variate tests on a collection of multidimensional time series data sets. Moreover, this alternative approach requires that each sequence of observations is of the same length which is not a prerequisite in the case of the signature-based test. \\

Our contribution is to study more deeply this statistical test from a numerical point of view on a variety of stochastic models and to show its practical interest for the validation of real-world stochastic scenarios when this validation is specified as an hypothesis testing problem. First, we present a numerical analysis with synthetic data in order to measure the statistical power of the test and second, we work with historical data to study the ability of the test to discriminate between several models in practice. Moreover, two constraints are considered in the numerical experiments. The first one is to impose that the distributions of the annual increments are the same in the two compared samples, which is natural since, without this constraint, point-in-time validation methods could distinguish the two samples. Secondly, in order to mimic the operational process of real-world scenarios validation in insurance, we consider samples of different sizes: the first sample consisting of synthetic or real historical paths is of small size (typically below 50) while the second sample consisting of the simulated scenarios is of greater size (typically around 1000). Our aim is to demonstrate the high statistical power of the test under these constraints. Numerical results are presented for three financial risk drivers, namely of a stock price, stock volatility as well as inflation. As for the stock price, we first generate two samples of paths under the widespread Black-Scholes dynamics, each sample corresponding to a specific choice of the (deterministic) volatility structure. We also simulate two samples of paths under both the classic and rough Heston models. As for the stock volatility, the two samples are generated using fractional Brownian motions with different Hurst parameters. Note that the exponential of a fractional Brownian motion, as model for volatility, has been proposed by \citet{gatheral2018volatility} who showed that such a model is consistent with historical data. For the inflation, one sample is generated using a regime-switching $AR(1)$ process and the other sample is generated using a random walk with i.i.d. Gamma noises. Finally, we compare two samples of two-dimensional processes with either independent or correlated coordinates when the first coordinate evolves according to the price in the rough Heston model and the second coordinate evolves according to an $AR(1)$ regime-switching process. Besides these numerical results on simulated paths, we also provide numerical results on real historical data. More specifically, we test historical paths of S\&P 500 realized volatility (used as a proxy of spot volatility) against sample paths from a standard Ornstein-Uhlenbeck model on the one hand and against sample paths from a fractional Ornstein-Uhlenbeck model on the other hand. We show that the test allows to reject the former model while the latter is not rejected. Similarly, it allows to reject a random walk model with i.i.d. Gamma noises when applied to US inflation data while a regime-switching $AR(1)$ process is not rejected. A summary of the studied risk factors and associated models is provided in Table \ref{tab:summary_models}. \\

\begin{rk}
   In the literature, the closest work to ours is the one of \cite{bilokon2021market}. Given a set of paths, they propose to compute the signature of these paths in order to apply the clustering algorithm of \cite{azran2006new} and to ultimately identify market regimes. The similarity function underlying the clustering algorithm relies in particular on the Maximum Mean Distance. Using Black-Scholes sample paths corresponding to four different configurations of the drift and volatility parameters, they show the ability of their methodology to correctly cluster the paths and to identify the number of different configurations (i.e. four). Let us point out that our contributions are different from theirs in several ways. First, their objective is to cluster a set of paths in several groups while our objective is to statistically test whether two sets of paths come from the same probability distribution. Second, our numerical results go beyond the Black-Scholes model and explore more sophisticated models. Moreover, we also provide numerical results on historical data. Finally, our numerical experiments are conducted in a setting closer to our practical applications where the marginal one-year distributions of the two samples are the same or very close while the frequency of observation of the paths is lower in our case (we consider monthly observations over one year while they consider 100 observations). 
\end{rk}

\begin{table}
   \centering
   \caption{Summary of the studied risk factors and associated models in each framework (synthetic data and real historical data)}
   \label{tab:summary_models}
   \resizebox{\textwidth}{!}{%
   \begin{tabular}{@{}ccc@{}}
   \toprule
                                       & Risk factor      & Models                                                 \\ \midrule
   \multirow{3}{*}{Synthetic data}    & Stock price      & Black-Scholes dynamics and rough Heston model                              \\ \cmidrule(l){3-3} 
                                       & Stock volatility & Fractional Brownian motion                             \\ \cmidrule(l){3-3} 
                                       & Inflation        & Regime-switching AR(1) process and   Gamma random walk \\ \midrule
   \multirow{2}{*}{Historical   data} & Stock volatility & Ordinary and fractional   Ornstein-Uhlenbeck processes \\
                                       & Inflation        & Regime-switching AR(1) process and   Gamma random walk \\ \bottomrule
   \end{tabular}%
   }
\end{table}

The objective of the present article is also to provide a concise introduction to the signature theory that does not require any prerequisite for insurance practitioners. Introduced for the first time by \citet{chen1957} in the late 50s and then rediscovered in the 90s in the context of rough path theory \citep{lyons1998}, the signature is a mapping that allows to characterize deterministic paths up to some equivalence relation (see Theorem \ref{thm:sig_injectivity}). \citet{chevyrev2018signature} have extended this result to stochastic processes as they have shown that the expected signature of a stochastic process characterizes its law. The idea to use the signature to address problems in finance is not new although it is quite recent. To our knowledge, \citet{gyurko2010} are the first in this area. They present a general framework for deriving high order, stable and tractable pathwise approximations of stochastic differential equations relying on the signature and apply their results to the simulation of the Cox-Ingersoll-Ross process. Then, \citet{gyurko2013extracting} introduced the signature as a way to obtain a faithful transform of financial data streams that is used as a feature of a classification method. \citet{levin2013learning} use the signature to study the problem of regression where the input and the output variables are paths and illustrate their results by considering the prediction task for AR and ARCH time series models. In the same vein, \cite{cohen2023nowcasting} address the nowcasting problem - the problem of computing early estimates of economic indicators that are published with a delay such as the GDP - by applying a regression on the signature of frequently published indicators. \citet{ni2021sig} develop a Generative Adversarial Network (GAN) based on the signature allowing to generate time series that capture the temporal dependence in the training and validation data set both for synthetic and real data. In his PhD thesis, \citet{perez2020signatures} shows several applications of the signature in finance including the pricing and hedging of exotic derivatives (see \citeauthor{lyons2020}, \citeyear{lyons2020}) or optimal execution problems (see \citeauthor{kalsi2020}, \citeyear{kalsi2020} and \citeauthor{cartea2022}, \citeyear{cartea2022}). \citet{cuchiero2022signature} extend the work of Perez Arribas and develop a new class of asset price models based on the signature of semimartingales allowing to approximate arbitrarily well classical models such as the SABR and the Heston models. Using this new modelling framework, \cite{cuchiero2023joint} propose a method to solve the joint S\&P 500/VIX calibration problem without adding jumps or rough volatility. \cite{akyildirim2022applications} introduce a signature-based machine learning algorithm to detect rare or unexpected items in a given data set of time series type. \cite{cuchiero2023signature} employ the signature in the context of stochastic portfolio theory to introduce a novel class of portfolios making in particular several optimization tasks highly tractable. Finally, \cite{bayer2023} propose a new method relying on the signature for solving optimal stopping problems.\\

The present article is organized as follows: as a preliminary, we introduce in Section \ref{sec:theoretical_foundations} the Maximum Mean Distance and the signature before describing the statistical test proposed by \citet{chevyrev2018signature}. This test is based on these two notions and allows to assess whether two stochastic processes have the same law using finite numbers of their sample paths. Then in Section \ref{sec:num_results}, we study this test from a numerical point of view. We start by studying its power using synthetic data in settings that are realistic in view of insurance applications and then, we apply it to real historical data. We also discuss several challenges related to the numerical implementation of this approach, and highlight its domain of validity in terms of the distance between models and the volume of data at hand.

\section{From the MMD and the signature to a two-sample test for stochastic processes}\label{sec:theoretical_foundations}
In this section, we start by introducing the Maximum Mean Distance (MMD), which allows to measure how similar two probability measures are. Secondly, Reproducing Kernel Hilbert Spaces (RKHS) are presented as they are key to obtain a simple formula for the MMD. Then, we briefly introduce the signature and we show how it allows to construct a RKHS that we can use to make the MMD a metric able to discriminate two probability measures defined on the bounded variation paths quotiented by some equivalence relation. Finally, the statistical test underlying the signature-based validation is introduced. In what follows, $\mathcal{X}$ is a metric space. 

\subsection{The Maximum Mean Distance}

\begin{defi}[Maximum Mean Distance]\label{def:mmd}
Let $\mathcal{G}$ be a class of functions $f:\mathcal{X}\rightarrow \mathbb{R}$ and $\mu$, $\nu$ two Borel probability measures defined on $\mathcal{X}$. The Maximum Mean Distance (MMD) is defined as:
\begin{equation}
MMD_{\mathcal{G}}(\mu,\nu) = \sup_{f\in \mathcal{G}}\left|\int_{\mathcal{X}}f(x)\mu(dx)-\int_{\mathcal{X}}f(x)\nu(dx) \right|.
\end{equation}
\end{defi}

Depending on $\mathcal{G}$, the MMD is not necessarily a metric (actually, it is a pseudo-metric, that is a metric without the property that two points with zero distance are identical), i.e. we could have $MMD_{\mathcal{G}}(\mu,\nu) = 0$ for some $\mu \ne \nu$ if the class of functions $\mathcal{G}$ is not rich enough. A sufficiently rich class of functions that makes $MMD_{\mathcal{G}}$ a metric is for example the space of bounded continuous functions on $\mathcal{X}$ equipped with a metric $d$ (Lemma 9.3.2 of \citeauthor{dudley2002real}, \citeyear{dudley2002real}). A sufficient condition on $\mathcal{G}$ for $MMD_\mathcal{G}$ to be a metric is given in Theorem 5 of \citet{gretton2012kernel}. \\

As presented in Definition \ref{def:mmd}, the MMD appears more as a theoretical tool than a practical one since computing this distance seems impossible in practice due to the supremum over a class of functions. However, if this class of function is the unit ball in a reproducing kernel Hilbert space (RKHS), the MMD is much simpler to estimate. Before setting out this result precisely, let us make a quick reminder about Mercer kernels and RKHSs. 

\begin{defi}[Mercer kernel]
A mapping $K:\mathcal{X}\times \mathcal{X}\rightarrow \mathbb{R}$ is called a Mercer kernel if it is continuous, symmetric and positive semi-definite i.e. for all finite sets $\{x_1, \dots, x_k\} \subset \mathcal{X}$ and for all $(\alpha_1,\dots,\alpha_k)\in\mathbb{R}^k$, the kernel $K$ satisfies:
\begin{equation}
\sum_{i=1}^k\sum_{j=1}^k \alpha_i\alpha_j K(x_i,x_j) \ge 0.
\end{equation}
\end{defi}
\begin{rk}
In the kernel learning litterature, it is common to use the terminology "positive definite" instead of "semi-positive definite" but we prefer the latter one in order to be consistent with the linear algebra standard terminology.
\end{rk}

We set:
\begin{equation}
\mathcal{H}_0 = \Span\{K_x:=K(x,\cdot) \mid x\in \mathcal{X}\}.
\end{equation}
With these notations, we have the following theorem due to Moore-Aronszajn (see Theorem 2 of Chapter III in \citeauthor{cucker2002}, \citeyear{cucker2002}):
\begin{thm}[Moore-Aronszajn]
Let $K$ be a Mercer kernel. Then, there exists a unique Hilbert space $\mathcal{H}\subset \mathbb{R}^{\mathcal{X}}$ with scalar product $\langle \cdot, \cdot \rangle_{\mathcal{H}}$ satisfying the following conditions:
\begin{enumerate}[(i)]
\item $\mathcal{H}_0$ is dense in $\mathcal{H}$
\item For all $f\in \mathcal{H}$, $f(x) = \langle K_x, f \rangle_{\mathcal{H}}$ (reproducing property). 
\end{enumerate} 
\end{thm}

\begin{rk}
The obtained Hilbert space $\mathcal{H}$ is said to be a reproducing kernel Hilbert space (RKHS) whose reproducing kernel is $K$. 
\end{rk}

We may now state the main theorem about the MMD. 
\begin{thm}\label{thm:mmd_rkhs}
Let $(\mathcal{H},K)$ be a reproducing kernel Hilbert space and let $\mathcal{G}:=\{f\in  \mathcal{H} \mid \|f\|_{\mathcal{H}} \le 1\}$. If $\int_{\mathcal{X}}\sqrt{K(x,x)}\mu(dx) < \infty$ and $\int_{\mathcal{X}}\sqrt{K(x,x)}\nu(dx) < \infty$, then for $X$, $X'$ independent random variables distributed according to $\mu$ and $Y$, $Y'$ independent random variables distributed according to $\nu$ and such that $X$ and $Y$ are independent, $K(X,X')$, $K(Y,Y')$ and $K(X,Y)$ are integrable and:
\begin{equation}
MMD^2_{\mathcal{G}}(\mu,\nu) = \mathbb{E}[K(X,X')]+\mathbb{E}[K(Y,Y')]-2\mathbb{E}[K(X,Y)].
\end{equation}
\end{thm}

The proof of this theorem can be found in Lemma 6 of \citet{gretton2012kernel}. A natural question at this stage is that of the choice of the Mercer kernel in order to obtain a metric on the space of probability measures defined on the space of continuous mappings from $[0,T]$ to a finite dimensional vector space or at least on a non-trivial subspace of this space. \citet{chevyrev2018signature} constructed such a Mercer kernel using the signature, which we will now define. \\

\subsection{The signature}\label{sec:intro_signature}

This subsection aims at providing a short overview of the signature, more details are given in the Appendix. We call "path" any continuous mapping from some time interval $[0,T]$ to a finite dimensional vector space $E$ which we equip with a norm $\|\cdot \|_E$. We denote by $\otimes$ the tensor product defined on $E\times E$ and by $E^{\otimes n}$ the tensor space obtained by taking the tensor product of $E$ with itself $n$ times:
\begin{equation}
E^{\otimes n} = \underbrace{E\otimes \dots \otimes E}_{n \text{ times}}.
\end{equation} 
The space in which the signature takes its values is called the space of formal series of tensors. It can be defined as:
\begin{equation}
T(E) = \left\{(\mathbf{t}^n)_{n\ge 0} \mid \forall n\ge 0, \mathbf{t}^n \in E^{\otimes n} \right\}
\end{equation}
with the convention $E^{\otimes 0} = \mathbb{R}$.\\

In a nutshell, the signature of a path $X$ is the collection of all iterated integrals of $X$ against itself. In order to be able to define these iterated integrals of $X$, one needs to make some assumptions about the regularity of $X$. The simplest framework is to assume that $X$ is of bounded variation.

\begin{defi}[Bounded variation path]\label{def:bounded_variation}
We say that a path $X:[0,T]\rightarrow E$ is of bounded variation on $[0,T]$ if its total variation
\begin{equation}
\|X\|_{1,[0,T]}:=\sup_{(t_0,\dots,t_r )\in \mathcal{D} } \sum_{i=0}^{r-1} \| X_{t_{i+1}}- X_{t_i} \|_E 
\end{equation}
is finite with $\mathcal{D}=\{(t_0,\dots,t_r ) \mid r\in \mathbb{N}^*, t_0 = 0 < t_1 < \dots < t_r = T   \}$ the set of all subdivisions of $[0,T]$.   
\end{defi}
\begin{nota}
We denote by $\mathcal{C}^1([0,T],E)$ the set of bounded variation paths from $[0,T]$ to $E$. 
\end{nota}
\begin{rk}
Intuitively, a bounded variation path on $[0,T]$ is a path whose graph vertical arc length is finite. In fact, if $X$ is a real-valued continuously differentiable path on $[0,T]$, then $\|X\|_{1,[0,T]} = \int_0^T |X'_t| dt $.
\end{rk}

We can now define the signature of a bounded variation path. 

\begin{defi}[Signature]\label{def:signature}
Let $X:[0,T]\rightarrow E$ be a bounded variation path. The signature of $X$ on $[0,T]$ is defined as $S_{[0,T]}(X) = (\mathbf{X}^n)_{n\ge 0}$ where by convention $\mathbf{X}^0=1$ and 
\begin{equation}
\mathbf{X}^n = \int_{0\le u_1 < u_2 < \dots < u_n \le T} dX_{u_1}\otimes \dots \otimes dX_{u_n} \in E^{\otimes n}.
\end{equation}
where the integrals must be understood in the sense of Riemann-Stieljes. We call $\mathbf{X}^n$ the term of order $n$ of the signature and $S^N_{[0,T]}(X) = (\mathbf{X}^n)_{0\le n \le N}$ the truncated signature at order $N$. Note that when the time interval is clear from the context, we will omit the subscript of $S$. 
\end{defi}
\begin{ex}\label{ex:sig_1d}
If $X$ is a one-dimensional bounded variation path, then its signature over $[0,T]$ is very simple as it reduces to the powers of the increment $X_T-X_0$, i.e. for any $n\ge 0$:
\begin{equation}
\mathbf{X}^n = \frac{1}{n!}(X_T-X_0)^n.
\end{equation}
\end{ex}

The above definition could be extended to less regular paths, namely to paths of finite $p$-variation with $p<2$. In this case, the integrals can be defined in the sense of \citet{young1936}. However, if $p\ge 2$, it is no longer possible to define the iterated integrals. Still, it is possible to give a sense to the signature but the definition is much more involved and relies on the rough path theory so we refer the interested reader to \citet{lyons2007differential}. \\

In this work, we will focus on bounded variation paths where the signature takes its values in the space of finite formal series (as a consequence of Proposition 2.2 of \citeauthor{lyons2007differential}, \citeyear{lyons2007differential}):
\begin{equation}\label{eq:finite_formal_series}
T^*(E) := \left\{\mathbf{t}\in T(E) \mid \|\mathbf{t}\| := \sqrt{\sum_{n\ge 0} \|\mathbf{t}^n\|_{E^{\otimes n}}^2} <\infty \right\}
\end{equation}
where $\|\cdot\|_{E^{\otimes n}}$ is the norm induced by the scalar product $\langle \cdot, \cdot \rangle_{E^{\otimes n}}$ defined on $E^{\otimes n}$ by:
\begin{equation}
\langle x, y \rangle_{E^{\otimes n}} = \sqrt{\sum_{I=(i_1,\dots,i_n)\in \{1,\dots,d\}^n} x_Iy_I} \quad \text{for } x,y\in E^{\otimes n}
\end{equation}
with $d$ the dimension of $E$ and $x_I$ (resp. $y_I$) the coefficient at position $I$ of $x$ (resp. $y$).\\

The signature is a powerful tool allowing to encode a path in a hierarchical and efficient manner. In fact, two bounded variation paths having the same signature are equal up to an equivalence relation (the so-called tree-like equivalence, denoted by $\sim_t$, and defined in Appendix \ref{sec:sig_properties}). In other words, the signature is one-to-one on the space $\mathcal{P}^1([0,T],E):=\mathcal{C}^1([0,T],E)/\sim_t$ of bounded variation paths quotiented by the tree-like equivalence relation. This is presented in a more comprehensive manner in Appendix \ref{sec:sig_properties}. Now, we would like to characterize the law of stochastic processes with bounded variation sample paths using the expected signature, that is the expectation of the signature taken component-wise. In a way, the expected signature is to stochastic processes what the sequence of moments is to random vectors. Thus, in the same way that the sequence of moments characterizes the law of random vectors only if the moments do not grow too fast, we need that the terms of the expected signature do not grow too fast in order to be able to characterize the law of stochastic processes. In order to avoid to have to restrict the study to laws with compact support (as assumed by \citeauthor{fawcett2002problems}, \citeyear{fawcett2002problems}), \citet{chevyrev2018signature} propose to apply a normalization mapping to the signature ensuring that the norm of the normalized signature is bounded. This property allows them to prove the characterization of the law of a stochastic process by its expected normalized signature (Theorem \ref{thm:sig_law_sto}). One of the consequences of this result is the following theorem, which makes the connection between the MMD and the signature and represents the main theoretical result underlying the signature-based validation. This result is a particular case of Theorem 30 from \citet{chevyrev2018signature}.

\begin{thm}\label{thm:application_rkhs_sig}
Let $E$ be a Hilbert space and $\langle\cdot, \cdot\rangle$ the scalar product on $T^*(E)$ defined by $\langle \mathbf{x}, \mathbf{y} \rangle= \sum_{n\ge 0} \langle \mathbf{x}^n, \mathbf{y}^n \rangle_{E^{\otimes n}}$
for all $\mathbf{x}$ and $\mathbf{y}$ in $T^*(E)$. Then the signature kernel defined on $\mathcal{P}^1([0,T],E)$ by $K^{sig}(x,y) = \left\langle\Phi(x),\Phi(y)\right\rangle$ where $\Phi$ is the normalized signature (see Theorem \ref{thm:sig_law_sto}), is a Mercer kernel and we denote by $\mathcal{H}^{sig}$ the associated RKHS. Moreover, $MMD_{\mathcal{G}}$ where $\mathcal{G}$ is the unit ball of $\mathcal{H}^{sig}$ is a metric on the space $\mathcal{M}$ defined as:
\begin{equation}
\mathcal{M} = \left\{\mu \text{ Borel probability measure defined on } \mathcal{P}^1([0,T],E) \middle\vert \int_{\mathcal{P}^1}\sqrt{K^{sig}(x,x)}\mu(dx) < \infty \right\}
\end{equation}
and we have:
\begin{equation}\label{eq:mmd_sig_rkhs}
MMD_{\mathcal{G}}(\mu,\nu) = \mathbb{E}[K^{sig}(X,X')]+\mathbb{E}[K^{sig}(Y,Y')]-2\mathbb{E}[K^{sig}(X,Y)]
\end{equation}
where $X$,$X'$ are independent random variables distributed according to $\mu$ and $Y$,$Y'$ are independent random variables distributed according to $\nu$ such that $Y$ is independent from $X$. 
\end{thm}
Based on this theorem, Chevyrev and Oberhauser propose a two-sample statistical test that allows to test whether two samples of paths come from the same distribution and that we now introduce.

\subsection{A two-sample test for stochastic processes}
\label{sec:two_sample_test}
Assume that we are given a sample $(X_1,\dots,X_m)$ consisting of $m$ independent realizations of a stochastic process of unknown law $\mathbb{P}_X$ and an independent sample $(Y_1,\dots,Y_n)$ consisting of $n$ independent realizations of a stochastic process of unknown law $\mathbb{P}_Y$. We assume that both processes are in $\mathcal{X} = \mathcal{P}^1([0,T],E)$ almost surely. A natural question is whether $\mathbb{P}_X = \mathbb{P}_Y$. Let us consider the following null and alternative hypotheses:
\begin{equation}
H_0:\mathbb{P}_X=\mathbb{P}_Y,\ H_1:\mathbb{P}_X\ne \mathbb{P}_Y.
\end{equation}
According to Theorem \ref{thm:application_rkhs_sig}, we have $MMD_{\mathcal{G}}(\mathbb{P}_X,\mathbb{P}_Y) \neq 0$ under $H_1$ while $MMD_{\mathcal{G}}(\mathbb{P}_X,\mathbb{P}_Y) = 0$ under $H_0$ when $\mathcal{G}$ is the unit ball of the reproducing Hilbert space associated to the signature kernel (note that we use the notation $K$ instead of $K^{sig}$ for the signature kernel in this section as there is no ambiguity). Moreover,
\begin{equation}
MMD^2_{\mathcal{G}}(\mathbb{P}_X,\mathbb{P}_Y) = \mathbb{E}[K(X,X')]+\mathbb{E}[K(Y,Y')]-2\mathbb{E}[K(X,Y)]
\end{equation}
where $X$,$X'$ are two random variables of law $\mathbb{P}_X$ and $Y$, $Y'$ are two random variables of law $\mathbb{P}_Y$ with $X,Y$ independent. This suggests to consider the following test statistic:
\begin{equation}\label{eq:mmd_estimator}
\begin{aligned}
MMD^2_{m,n}(X_1,\dots,X_m,Y_1,\dots,Y_n):= &\frac{1}{m(m-1)}\sum_{1\le i\ne j \le m} K(X_i,X_j)+\frac{1}{n(n-1)}\sum_{1\le i\ne j \le n} K(Y_i,Y_j)\\
&-\frac{2}{mn}\sum_{\substack{1\le i \le m \\ 1\le j \le n}} K(X_i,Y_j)
\end{aligned}
\end{equation}
as it is an unbiased estimator of $MMD^2_{\mathcal{G}}(\mathbb{P}_X,\mathbb{P}_Y)$. In the sequel, we omit the dependency on $X_1,\dots,X_m$ and $Y_1,\dots,Y_n$ for notational simplicity. This estimator is even strongly consistent as stated by the following theorem which is an application of the strong law of large numbers for two-sample $U$-statistics \citep{kumar1977}. 
\begin{thm}\label{thm:lfgn_u_stats}
Assuming that
\begin{enumerate}[(i)]
\item $\mathbb{E}\left[ \sqrt{K(X,X)}\right] < \infty$ with $X$ distributed according to $\mathbb{P}_X$ and $\mathbb{E}\left[ \sqrt{K(Y,Y)}\right] < \infty$ with $Y$ distributed according to $\mathbb{P}_Y$
\item $\mathbb{E}[|h|\log^+|h|] < \infty$ where $\log^+$ is the positive part of the logarithm and 
\begin{equation}
h = K(X,X')+K(Y,Y')-\frac{1}{2}\left( K(X,Y)+K(X,Y')+K(X',Y)+K(X',Y')\right)
\end{equation}
with $X,X'$ distributed according to $\mathbb{P}_X$, $Y,Y'$ distributed according to $\mathbb{P}_Y$ and $X,X',Y,Y'$ independent
\end{enumerate}
then:
\begin{equation}
MMD^2_{m,n} \underset{m,n\to +\infty}{\overset{a.s.}{\rightarrow}} MMD^2_{\mathcal{G}}(\mathbb{P}_X,\mathbb{P}_Y).
\end{equation}
\end{thm}
Under $H_1$, $MMD^2_{\mathcal{G}}(\mathbb{P}_X,\mathbb{P}_Y)>0$ so that $N\times MMD^2_{m,n} \overset{a.s.}{\underset{m,n \to +\infty}{\rightarrow}} +\infty$ for $N=m+n$. Thus, we reject the null hypothesis at level $\alpha$ if $MMD^2_{m,n}$ is greater than some threshold $c_{\alpha}$. In order to determine this threshold, we rely on the asymptotic distribution of $MMD^2_{m,n}$ under $H_0$ which is due to \citeauthor{gretton2012kernel} (\citeyear{gretton2012kernel}, Theorem 12).\\

\begin{thm}\label{thm:H0_tcl}
Let us define the kernel $\tilde{K}$ by:
\begin{equation}
\tilde{K}(x,y) = K(x,y)-\mathbb{E}[K(x,X)] -\mathbb{E}[K(X,y)] -\mathbb{E}[K(X,X')] 
\end{equation}
where $X$ and $X'$ are i.i.d. samples drawn from $\mathbb{P}_X$. Assume that:
\begin{enumerate}[(i)]
\item $\mathbb{E}[\tilde{K}(X,X')^2] <+\infty$
\item $m/N \rightarrow \rho \in (0,1)$ as $N=m+n \rightarrow +\infty$. 
\end{enumerate}
Under these assumptions, we have:
\begin{enumerate}
\item under $H_0$:
\begin{equation}
N\times MMD^2_{m,n} \underset{N \to +\infty}{\overset{\mathcal{L}}{\rightarrow}}\frac{1}{\rho(1-\rho)}\sum_{\ell=1}^{+\infty}\lambda_{\ell} \left(G_{\ell}^2 -1 \right)
\end{equation}
where $(G_{\ell})_{\ell \ge 1}$ is an infinite sequence of independent standard normal random variables and the $\lambda_{\ell}$'s are the eigenvalues of the operator $S_{\tilde{K}}$ defined as:
\begin{equation}
\begin{array}{rrcl}
S_{\tilde{K}}:& L_2(\mathbb{P}_X) &\rightarrow &L_2(\mathbb{P}_X) \\
&g &\mapsto &\mathbb{E}[\tilde{K}(\cdot,X)g(X)]
\end{array}
\end{equation} 
with $L_2(\mathbb{P}_X):=\{g: \mathcal{X} \rightarrow \mathbb{R} \mid \mathbb{E}[g(X)^2] <\infty \}$. 
\item under $H_1$:
\begin{equation}
N\times MMD^2_{m,n} \underset{N\to +\infty}{\rightarrow} +\infty
\end{equation}
\end{enumerate}
\end{thm}

This theorem indicates that if one wants to have a test with level $\alpha$, one should take the $1-\alpha$ quantile of the above asymptotic distribution as rejection threshold. In order to approximate this quantile, we use an approach proposed by \citet{gretton2009fast} that aims at estimating the eigenvalues in Theorem \ref{thm:H0_tcl} using the empirical Gram matrix spectrum. This approach relies on the following theorem (Theorem 1 of \citeauthor{gretton2009fast}, \citeyear{gretton2009fast}).
\begin{thm}\label{thm:eigenv_threshold}
Let $(\lambda_{\ell})_{\ell\ge 1}$ be the eigenvalues defined in Theorem \ref{thm:H0_tcl} and $(G_{\ell})_{\ell\ge 1}$ be a sequence of i.i.d. standard normal variables. For $N=m+n$, we define the centred Gram matrix $\hat{A}= HAH $ where $A=(K(Z_i,Z_j))_{1\le i,j \le N}$ ($Z_i = X_i$ if $i\le m$ and $Z_i = Y_{i-m}$ if $i > m$) and $H = I_{N} -\frac{1}{N}\mathbf{1}\mathbf{1}^T$. If $\sum_{l=1}^{+\infty}\sqrt{\lambda_l} < \infty$ and $m/N \rightarrow \rho \in (0,1)$ as $N \rightarrow +\infty$, then under $H_0$:
\begin{equation}
\frac{1}{\hat{\rho}(1-\hat{\rho})}\sum_{\ell=1}^{+\infty} \frac{\nu_{\ell}}{N}\left(G_{\ell}^2-1\right)\underset{N\to +\infty}{\overset{\mathcal{L}}{\rightarrow}}\frac{1}{\rho(1-\rho)}\sum_{\ell=1}^{+\infty} \lambda_{\ell}\left(G_{\ell}^2-1\right)
\end{equation}
where $\hat{\rho}=m/N$ and the $\nu_{\ell}$'s are the eigenvalues of $\hat{A}$.
\end{thm}
Therefore, we can approximate the asymptotic distribution in Theorem \ref{thm:H0_tcl} by:
\begin{equation}
\frac{1}{\hat{\rho}(1-\hat{\rho})}\sum_{\ell=1}^{R} \frac{\nu_{\ell}}{N} \left(G_l^2-1\right)
\label{eq:H0_MMD_distribution}
\end{equation}
with $R$ the truncation order and $\nu_1 > \nu_2 > \dots > \nu_R$ are the $R$ first eigenvalues of $\hat{A}$ in decreasing order. A rejection threshold is then obtained by simulating several realizations of the above random variable and then computing their empirical quantile at level $1-\alpha$. 

\section{Implementation and numerical results}\label{sec:num_results}
The objective of this section is to show the practical interest of the two-sample test described in the previous section for the validation of real-world economic scenarios. In the sequel, we refer to the two-sample test as the signature-based validation test. As a preliminary, we discuss the challenges implied by the practical implementation of the signature-based validation test.

\subsection{Practical implementation of the signature-based validation test}
\subsubsection{Signature of a finite number of observations}
In practice, the historical paths of the risk drivers under interest are not observed continuously and the data are only available on a discrete time grid with steps that are generally not very small (one day for the realized volatility data and one month for the inflation considered in Section \ref{sec:hist_data}). One has to embed these observations into a continuous path in order to be able to compute the signature and \textit{a fortiori} the MMD with the model that we want to test. The two most popular embeddings in the literature are the linear and the rectilinear embeddings. The former one consists in a plain linear interpolation of the observations, while the latter consists in an interpolation using only parallel shifts with respect to the $x$ and $y$-axis. In the sequel, we will only use the linear embedding as this choice does not seem to have a material impact on the information contained in the obtained signature as shown by the comparative study led by Fermanian (section 4.2 of  \citeauthor{fermanian2021embedding}, \citeyear{fermanian2021embedding}). Even if several of the models considered below have paths with unbounded variation, we discretize them with a rather crude step (consistent with insurance practice) and then we compute the signatures of the continuous paths with bounded variation deduced by the linear interpolation embedding. \\

\subsubsection{Extracting information of a one-dimensional process}\label{sec:transf}

Remember that if $X$ is a one-dimensional bounded variation path, then its signature over $[0,T]$ is equal to the powers of the increment $X_T-X_0$. As a consequence, finer information than the global increment about the evolution of $X$ on $[0,T]$ is lost. In practice, $X$ is often of dimension 1. Indeed, the validation of real-world economic scenarios is generally performed separately for each risk driver and not all together as a first step and only then the validation is performed on a more global level (e.g. through the analysis of the copula among risk drivers). For example, when validating risk drivers such as equity or inflation, $X$ represents an equity or inflation index which is one-dimensional. Moreover, even for multidimensional risk drivers such as interest rates, practitioners tend to work on the marginal distributions (e.g. they focus separately on each maturity for interest rates) so that the validation is again one-dimensional. In order to nonetheless be able to capture finer information about the evolution of $X$ on $[0,T]$, one can apply a transformation to $X$ to recover a multi-dimensional path. The two most widely used transformations are the time transformation and the lead-lag transformation. The time transformation consists in considering the two-dimensional path $\hat{X}_t:t\mapsto (t,X_t)$ instead of $t\mapsto X_t$. The lead-lag transformation has been introduced by \citet{gyurko2013extracting} in order to capture the quadratic variation of a path in the signature. Let $X$ be a real-valued stochastic process and $0=t_0< t_{1/2} < t_1 < \dots < t_{N-1/2} < t_N = T$ be a partition of $[0,T]$. The lead-lag transformation of $X$ on the partition $(t_{i/2})_{i=0,\dots,2N}$ is the two-dimensional path $t\mapsto (X_t^{lead},X_t^{lag})$ defined on $[0,T]$ where:
\begin{enumerate}
\item the lead process $t\mapsto X_t^{lead}$ is the linear interpolation of the points $(X^{lead}_{t_{i/2}})_{i=0,\dots,2N}$ with:
\begin{equation}
X^{lead}_{t_{i/2}} = 
\left\{
\begin{array}{ll}
X_{t_j} & \text{if } i = 2j \\
X_{t_{j+1}} & \text{if } i = 2j+1 
\end{array}
\right.
\end{equation}
\item the lag process $t\mapsto X_t^{lag}$ is the linear interpolation of the points $(X^{lag}_{t_{i/2}})_{i=0,\dots,2N}$ with:
\begin{equation}
X^{lag}_{t_{i/2}} = 
\left\{
\begin{array}{ll}
X_{t_j} & \text{if } i=2j \\
X_{t_j} & \text{if } i=2j+1
\end{array}
\right.
\end{equation}
\end{enumerate}
Illustrations of the lead and lag paths as well as of the lead-lag transformation are provided in Figure \ref{fig:lead-lag}. Note that by summing the area of each rectangle formed by the lead-lag transformation between two orange dots in Figure \ref{fig:lead_lag_2d}, one recovers the sum of the squares of the increments of $X$ which is exactly the quadratic variation of $X$ over the partition $(t_i)_{i=0,\dots,N}$. The link between the signature and the quadratic variation is more formally stated in Proposition \ref{prop:levy_area_ll}.
\begin{rk}
The choice of the dates $(t_{i+1/2})_{i=0,\dots,N-1}$ such that $t_i < t_{i+1/2} < t_{i+1}$ can be arbitrary since the signature is invariant by time reparametrization (see Proposition \ref{prop:time_reparam}). 
\end{rk}
\begin{figure}
\centering
\begin{subfigure}{.5\textwidth}
  \centering
  \includegraphics[width=\linewidth]{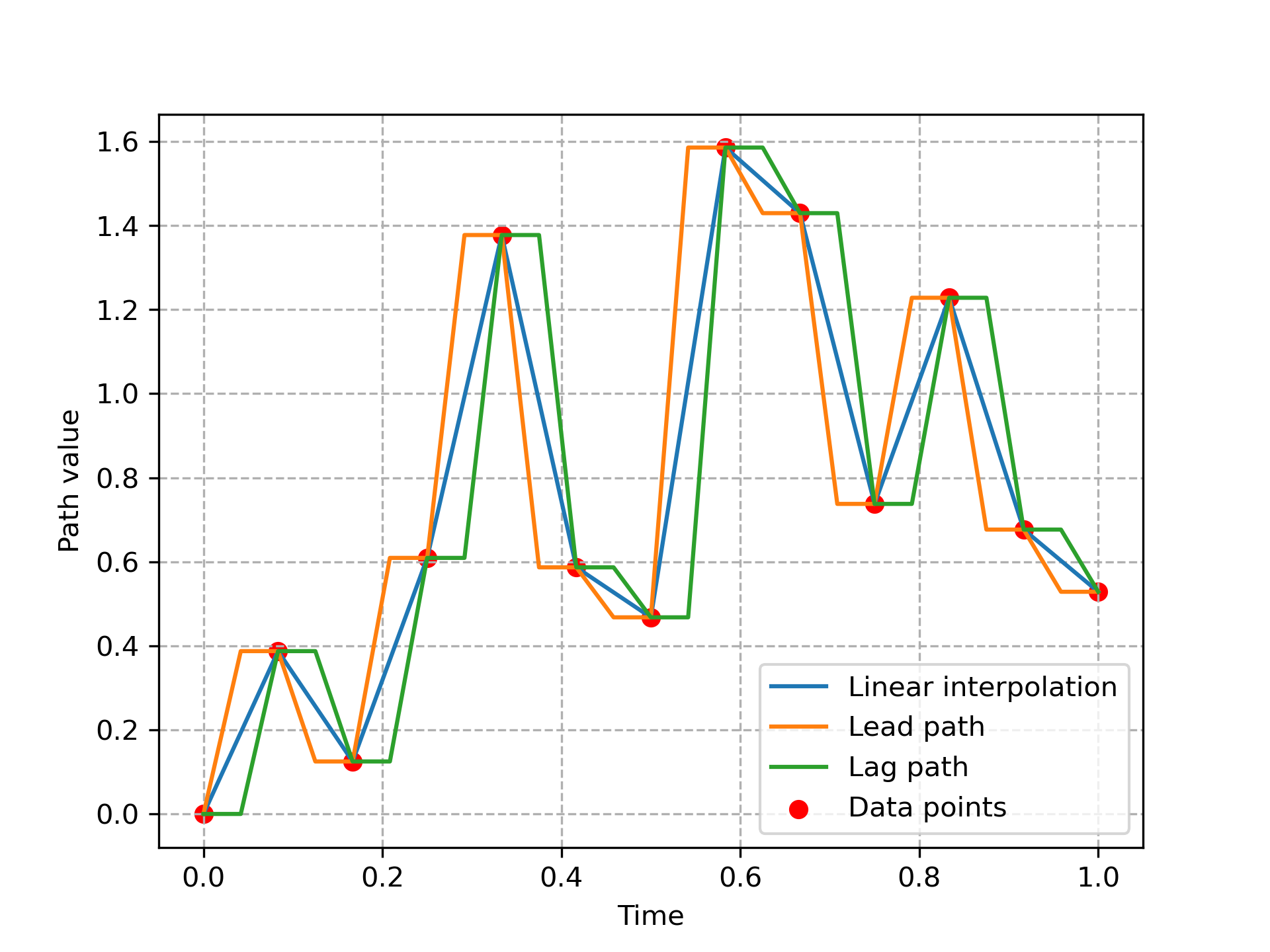}
  \caption{Lead and lag paths}
\end{subfigure}%
\begin{subfigure}{.5\textwidth}
  \centering
  \includegraphics[width=\linewidth]{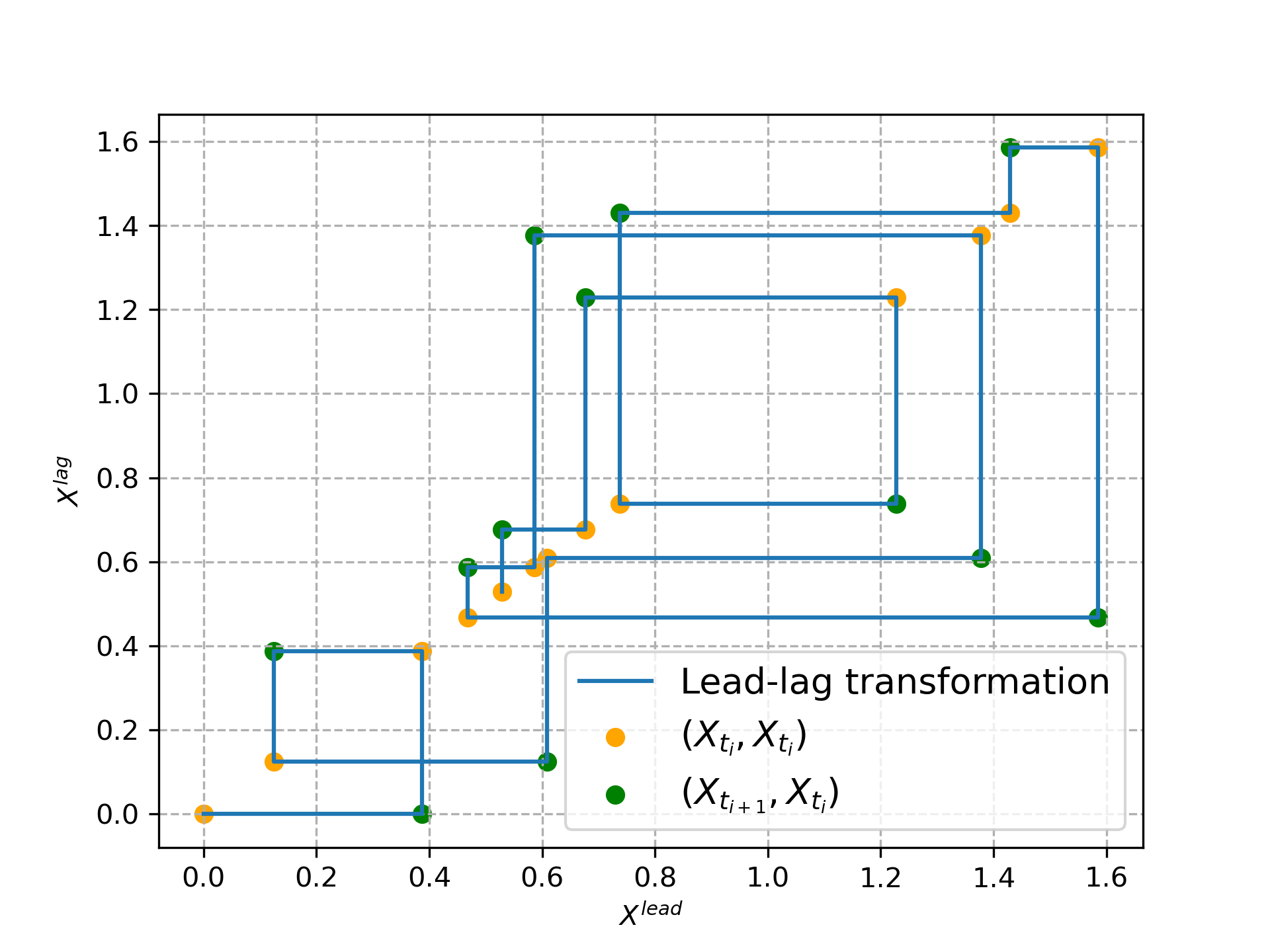}
  \caption{Lead-lag transformation}
  \label{fig:lead_lag_2d}
\end{subfigure}
\caption{Lead-lag illustrations}
\label{fig:lead-lag}
\end{figure}
 
A third transformation can be constructed from the time and the lead-lag transformations. Indeed, given a finite set of observations $(X_{t_i})_{i=0,\dots,N}$, one can consider the three-dimensional path $t\mapsto (t,X_t^{lead},X_t^{lag})$. We call this transformation the time lead-lag transformation. Finally, the cumulative lead-lag transformation is the two-dimensional path $t\mapsto (\tilde{X}_t^{lead},\tilde{X}_t^{lag})$ where $\tilde{X}_t^{lead}$ (resp. $\tilde{X}_t^{lag}$) is the lead (resp. lag) transformation of the points $(\tilde{X}_{t_{i}})_{i=0,\dots,N+1}$ with:
\begin{equation}\label{eq:cum_lead_lag}
\tilde{X}_{t_{i}} = \left\{
\begin{array}{ll}
0 & \text{for $i=0$}\\
\sum_{k=0}^{i-1} X_{t_{k}} & \text{for $i=1,\dots,N+1$}.
\end{array}
\right.
\end{equation}
This transformation has been introduced by \citet{chevyrev2016primer} because its signature is related to the statistical moments of the initial path $X$. More details on this point are provided in Remark \ref{rk:cum_lead_lag}. 

\subsubsection{Numerical computation of the signature and the MMD}
The numerical computation of the signature is performed using the \texttt{iisignature} Python package (version 0.24) of \cite{reizenstein2020}. The \texttt{signatory} Python package of \cite{kidger2020} could also be used for faster computations. Because the signature is an infinite object, we compute in practice only the truncated signature up to some specified order $R$. The influence of the truncation order on the statistical power of the test will be discussed in Section \ref{sec:sim_data_results}. Note that we will focus on truncation orders below 8 as there is not much information beyond this order given that we work with a limited number of observations of each path which implies that the approximation of high order iterated integrals will rely on very few points. 

\subsection{Analysis of the statistical power on synthetic data}\label{sec:sim_data_results}
In this subsection, we apply the signature-based validation on simulated data, i.e. the two samples of stochastic processes are numerically simulated. Keeping in mind insurance applications, the two-sample test is structured as follows:
\begin{itemize}
\item Each path is obtained by a linear interpolation from a set of 13 equally-spaced observations of the stochastic process under study. The first observation (i.e. the initial value of each path) is the same across all paths. These 13 observations represent monthly observations over a period of one year. In insurance practice, computational time constraints around the asset and liability models generally limit the simulation frequency to a monthly time step. The period of one year is justified by the fact that one needs to split the historical path under study into several shorter paths to get a test sample of size greater than 1. Because the number of historical data points is limited (about 30 years of data for the major equity indices), a split frequency of one year appears reasonable given the monthly observation frequency. 
\item The two samples are assumed to be of different sizes (i.e. $m\neq n$ with the notations of section \ref{sec:two_sample_test}). Several sizes $m$ of the first sample will be tested while the size $n$ of the second sample is always set to 1000. The first sample representing historical paths, we will mainly consider small values of $m$ as $m$ will in practice be equal to the number of years of available data (considering a split of the historical path in one-year length paths as discussed above). For the second sample which consists in simulated paths (for example by an Economic Scenario Generator), we take 1000 simulated paths as it corresponds to a lower bound of the number of scenarios typically used by insurers. Numerical tests (not presented here) have also been performed with a sample size of 5000 instead of 1000 but the results were essentially the same.
\item As we aim to explore the capability of the two-sample test to capture properties of the paths that cannot be captured by looking at their marginal distribution at some dates, we impose that the distributions of the increment over $[0,T]$ of the two compared stochastic processes are the same. In other words, we only compare stochastic processes $(X_t)_{0\le t \le T}$ and $(Y_t)_{0\le t\le T}$ satisfying $X_T-X_0\overset{\mathcal{L}}{=}Y_T-Y_0$ with $T=1$ year. This constraint is motivated by the fact that two models that do not have the same marginal one-year distribution are already discriminated by the current point-in-time validation methods. Moreover, it is a common practice in the insurance industry to calibrate the real-world models by minimizing the distance between model and historical moments so that the model marginal distribution is often close to the historical marginal distribution. Because of this constraint, we will remove the first order term of the signature in our estimation of the MMD because it is equal to the global increment $X_T-X_0$ and it does not provide useful statistical information. 
\end{itemize}

In order to measure the ability of the signature-based validation to distinguish two different samples of paths, we compute the statistical power of the underlying test, which is the probability to correctly reject the null hypothesis under $H_1$, by simulating 1000 times two samples of sizes $m$ and $n$ respectively and counting the number of times that the null hypothesis (the stochastic processes underlying the two samples are the same) is rejected. The rejection threshold is obtained using the empirical Gram matrix spectrum as described in Section \ref{sec:two_sample_test}. First, we generate a sample of size $m$ and a sample of size $n$ under $H_0$ to compute the eigenvalues of the matrix $\hat{A}$ in Theorem \ref{thm:eigenv_threshold}. Then, we keep the 20 first eigenvalues in decreasing order and we perform 10000 simulations of the random variable in Equation (\ref{eq:H0_MMD_distribution}) whose distribution approximates the  MMD asymptotic distribution under $H_0$. The rejection threshold is obtained as the empirical quantile of level $99$\% of these samples. For each experiment presented in the sequel, we also simulate 1000 times two samples of sizes $m$ and $n$ under $H_0$ and we count the number of times that the null hypothesis is rejected with this rejection threshold, which gives us the type I error. This step allows us to verify that the computed rejection threshold provides indeed a test of level $99$\% in all experiments. As we obtain a type I error around 1\% in all numerical experiments, we conclude to the accuracy of the computed rejection threshold.  \\

We will now present numerical results for three stochastic processes: the fractional Brownian motion, the Black-Scholes dynamics and the rough Heston model. Two time series models will also be considered: a regime-switching $AR(1)$ process and a random walk with i.i.d. Gamma noises. Finally, we study a two-dimensional process where one component evolves according to the rough Heston dynamics and the other one evolves according to a regime-switching $AR(1)$ process. The choice of this catalogue of models is both motivated by:
\begin{enumerate}
  \item the fact that the models are either widely used or realistic for the simulation of specific risk drivers (equity volatility, equity prices and inflation) and 
  \item the fact that the models allow to illustrate different pathwise properties such as Hölder regularity, volatility, autocorrelation or regime switches. 
\end{enumerate}  

\subsubsection{The fractional Brownian motion}\label{sec:fBm}
The fractional Brownian motion (fBm) is a generalization of the standard Brownian motion that, outside this standard case, is neither a semimartingale nor a Markov process and whose sample paths can be more or less regular than those of the standard Brownian motion. More precisely, it is the unique centred Gaussian process $(B_t^H)_{t\ge 0}$ whose covariance function is given by:
\begin{equation}
\mathbb{E}[B_s^HB_t^H]=\frac{1}{2}\left(s^{2H}+t^{2H}-(s-t)^{2H} \right)\quad \forall s,t\ge 0
\end{equation}
where $H\in (0,1)$ is called the Hurst parameter. Taking $H=1/2$, we recover the standard Brownian motion. The fBm exhibits two interesting pathwise properties.
\begin{enumerate}
\item The fBm sample paths are $H-\epsilon$ Hölder for all $\epsilon >0$. Thus, when $H<1/2$, the fractional Brownian motion sample paths are rougher than those of the standard Brownian motion and when $H>1/2$, they are smoother. 
\item The increments are correlated: if $s<t<u<v$, then $\mathbb{E}[(B_t^H-B_s^H)(B_v^H-B_u^H)]$ is positive if $H>1/2$ and negative if $H<1/2$ since $x\mapsto x^{2H}$ is convex if $H>1/2$ and concave otherwise. 
\end{enumerate}
One of the main motivations for studying this process is the work of \citet{gatheral2018volatility} which shows that the historical volatility of many financial indices essentially behaves as a fBm with a Hurst parameter around 10\%. \\

In the following numerical experiments, we will compare samples from a fBm with Hurst parameter $H$ and samples from a fBm with a different Hurst parameter $H'$. One can easily check that $B_1^H$ has the same distribution than $B_1^{H'}$ since $B_1^H$ and $B_1^{H'}$ are both standard normal variables. Thus, the constraint that both samples have the same one-year marginal distribution (see the introduction of Section \ref{sec:sim_data_results}) is satisfied. Note that a variance rescaling should be performed if one considers a horizon that is different from 1 year. We start with a comparison of fBm paths having a Hurst parameter $H=0.1$ with fBm paths having a Hurst parameter $H=0.2$ using the lead-lag transformation. In Figure \ref{fig:0.1vs0.2}, we plot the statistical power as a function of the truncation order $R$ for different values of the first sample size $m$ (we recall that the size of the second sample is fixed to 1000). 
\begin{figure}[ht]
\centering
\includegraphics[scale=0.6]{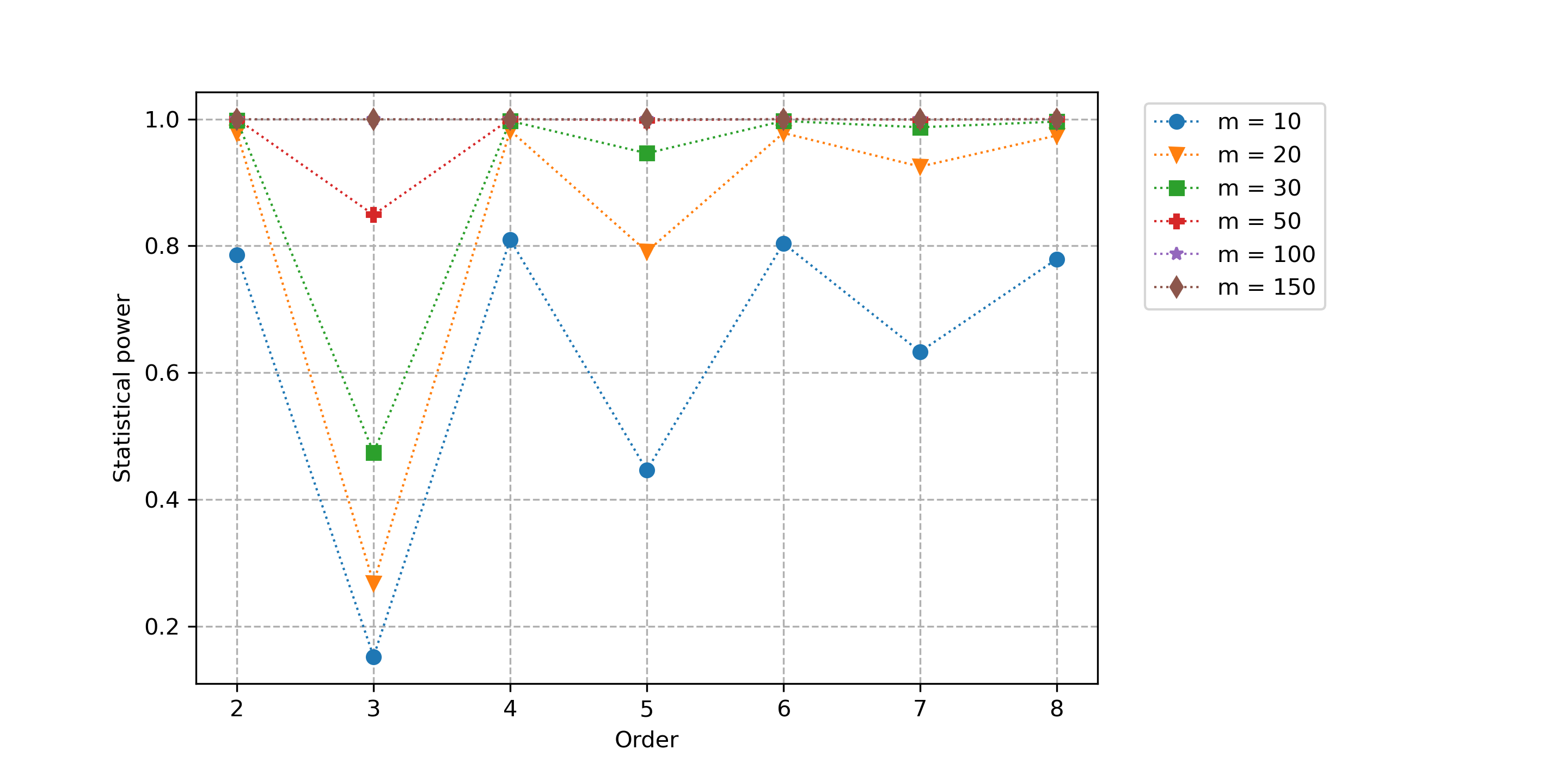}
\caption{Statistical power (as a function of the truncation order and the size of the first sample) of the signature-based validation test when comparing fBm paths with $H=0.1$ against fBm paths with $H'=0.2$. Note that the lead-lag transformation is applied.}
\label{fig:0.1vs0.2}
\end{figure}
We observe that even with small sample sizes, we already obtain a power close to 1 at order 2. Note that the power does not increase with the order but decreases at odd orders when the sample size is smaller than 50. This can be explained by the fact that the odd-order terms of the signature of the lead-lagged fBm are linear combinations of monomials in $B_{t_1}^H,\dots,B^H_{t_{N}}$ that are of odd degree. Since $(B_t^H)$ is a centred Gaussian process, the expectation of these terms are zero no matter the value of $H$. As a consequence, the contribution of odd-order terms of the signature to the MMD is the same under $H_0$ and under $H_1$. If we conduct the same experiment for $H=0.1$ versus $H'=0.5$ (corresponding to the standard Brownian motion), we obtain cumulated powers greater than 99\% for all tested orders and sample sizes (even $m=10$) even if the power of the odd orders is small (below 45\%). This is very promising as it shows that the signature-based validation allows distinguishing very accurately rough fBm paths (with a Hurst parameter in the range of those estimated by \citeauthor{gatheral2018volatility}, \citeyear{gatheral2018volatility}) from standard Brownian motion paths even with small sample sizes. \\

Note that in these numerical experiments, we have not used any tensor normalization while it is a key ingredient in Theorem \ref{thm:application_rkhs_sig}. This is motivated by the fact that the power is much worse when we use Chevyrev and Oberhauser's normalization (see Example 4 in \citeauthor{chevyrev2018signature}, \citeyear{chevyrev2018signature}) as one can see on Figure \ref{fig:c&0_normalization}. These lower powers can be understood as a consequence of the fact that the normalization is specific to each path. So the normalization can bring the distribution of $S(\hat{B}^{H'})$ closer to the one of $S(\hat{B}^H)$ than without normalization so that it is harder to distinguish them at fixed sample size. Moreover, if the normalization constant $\lambda$ is smaller than 1 (which we observe numerically), the high-order terms of the signature become close to zero and their contribution to the MMD is not material. For $H=0.1$ versus $H'=0.5$, we observed that the powers remain very close to 100\%. \\

\begin{figure}[ht]
  \centering
  \begin{subfigure}{.5\textwidth}
    \centering
    \includegraphics[width=0.9\linewidth]{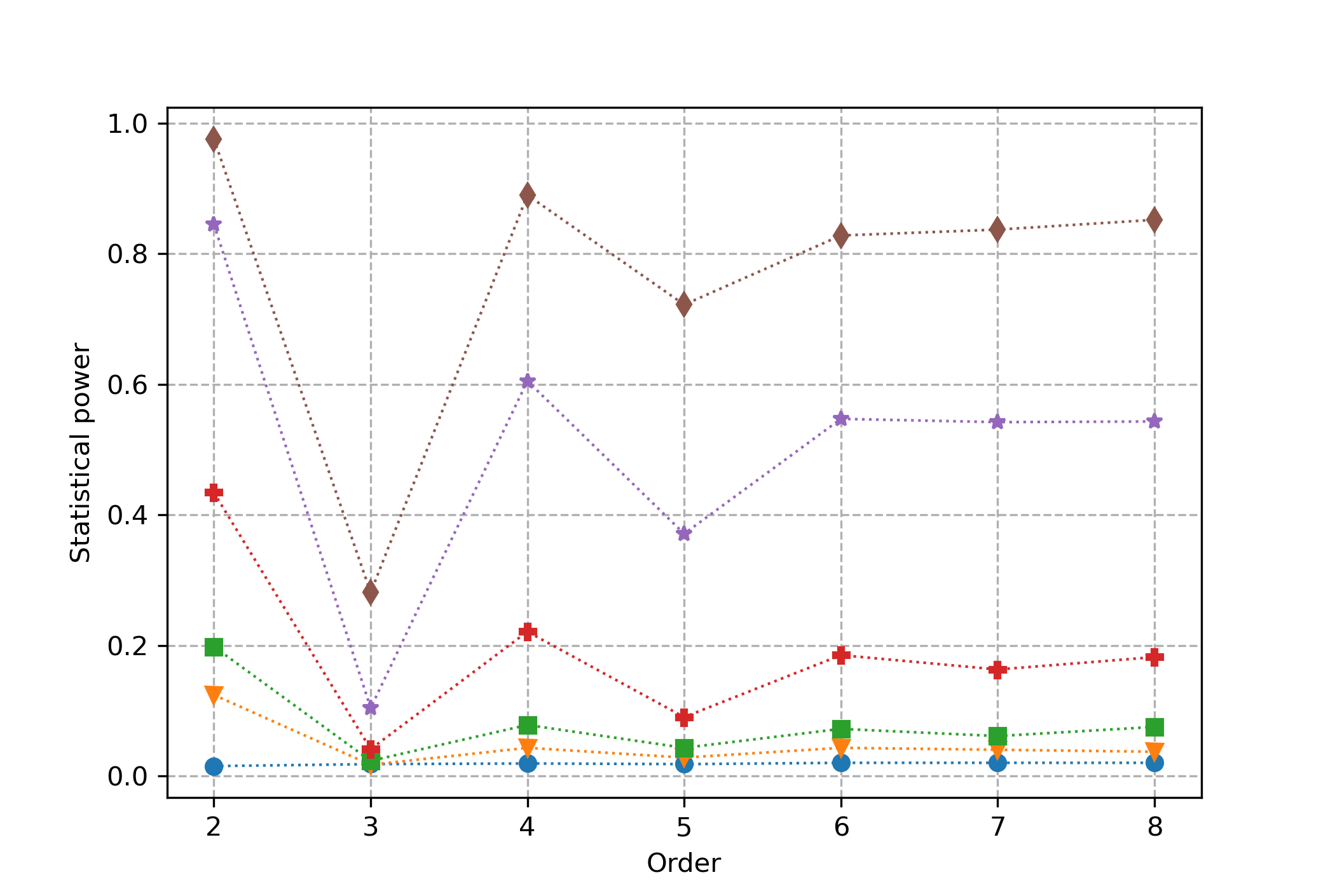}
    \caption{With Chevyrev and Oberhauser's normalization.}
    \label{fig:c&0_normalization}
  \end{subfigure}%
  \begin{subfigure}{.5\textwidth}
    \centering
    \includegraphics[width=10cm]{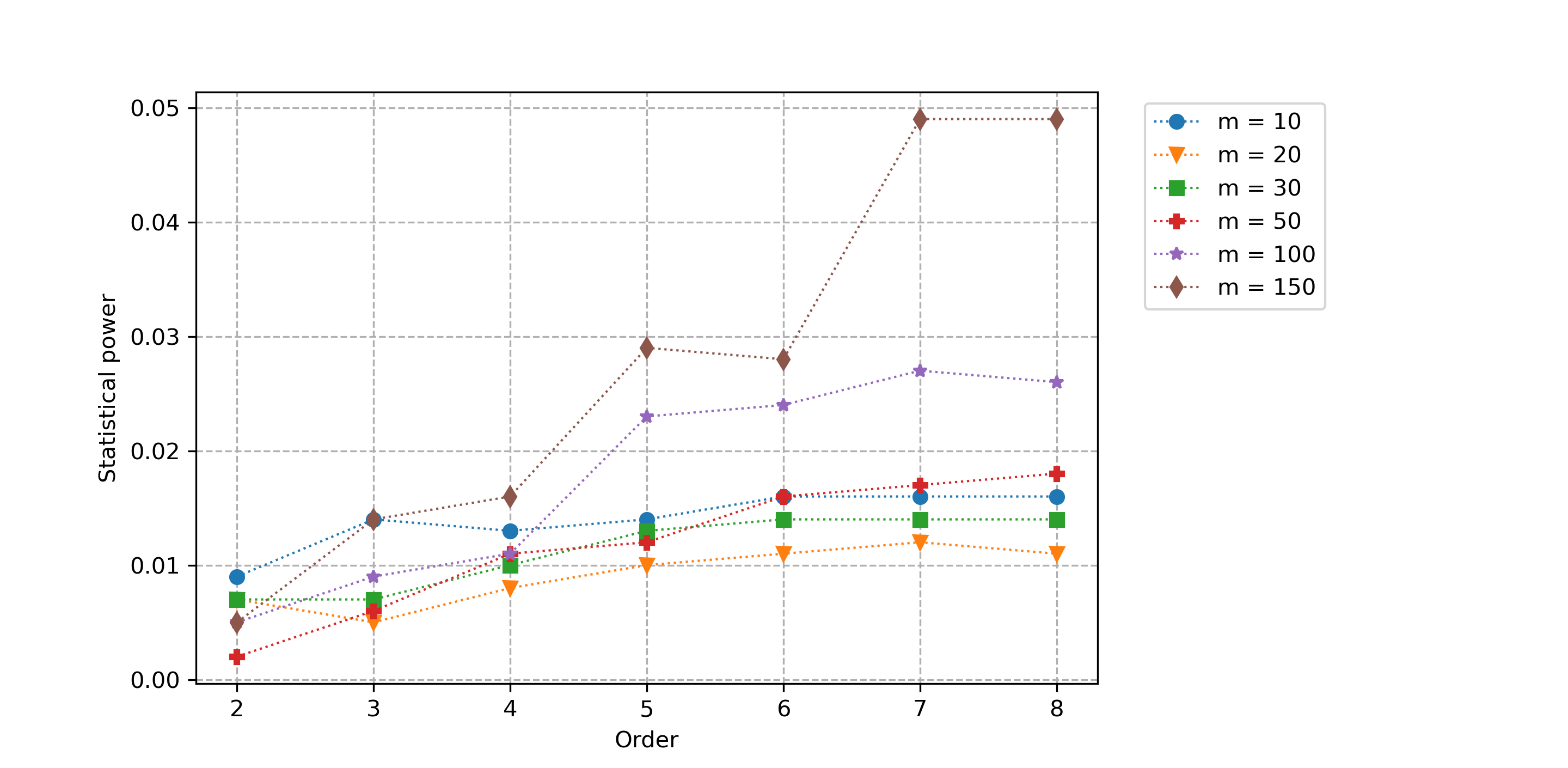}
    \caption{With the time transformation.}
    \label{fig:time}
  \end{subfigure}
  \caption{Statistical power (as a function of the truncation order and the size of the first sample) of the signature-based validation test when comparing fBm paths with $H=0.1$ against fBm paths with $H'=0.2$. }
  \end{figure}

Also note that the lead-lag transformation is key for this model as replacing it by the time transformation (see Section \ref{sec:transf}) results in much lower statistical powers, see Figure \ref{fig:time}. This observation is consistent with the previous study from \citet{fermanian2021embedding} which concluded that the lead-lag transformation is the best choice in a learning context. \\

\begin{figure}[ht]
  \centering
  \includegraphics[scale=0.6]{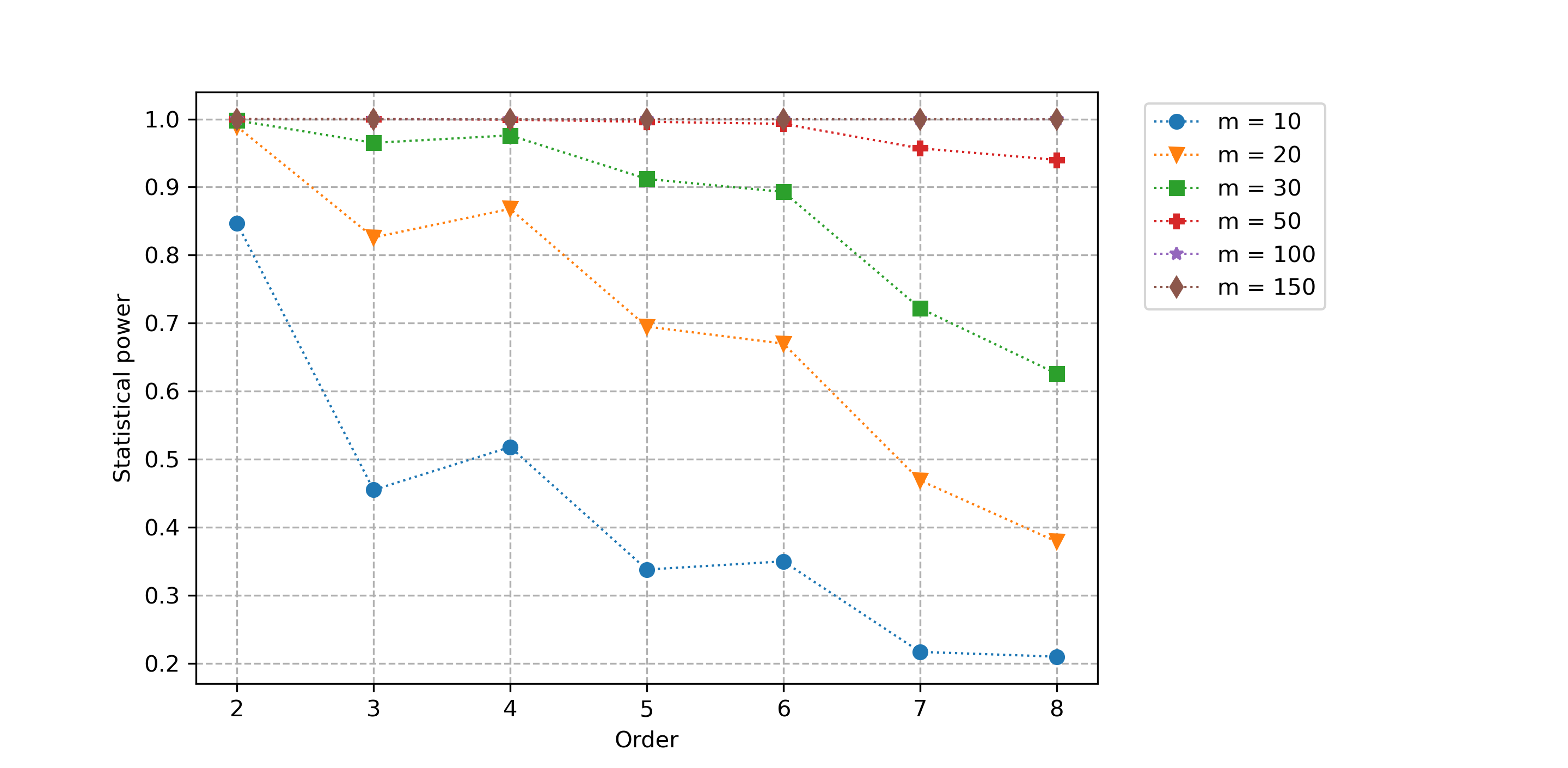}
  \caption{Statistical power (as a function of the truncation order and the size of the first sample) of the log-signature-based validation test when comparing fBm paths with $H=0.1$ against fBm paths with $H'=0.2$. Note that the lead-lag transformation is applied. } 
  \label{fig:log_signature}
  \end{figure}

Before moving to the Black-Scholes dynamics, we present results of the test when the signature is replaced by the log-signature. The log-signature is a more parsimonious - though equivalent - representation of paths than the signature as it contains more zeros. A formal definition of the log-signature and more insights can be found in the Appendix. Although no information is lost by the log-signature, it is not clear whether the MMD is still a metric when the signature is replaced by the log-signature in the kernel (see Remark \ref{rk:log-signature_caveat}). Numerically, the log-signature shows satisfying powers for $H=0.1$ versus $H'=0.2$, especially when the truncation order is 2 (see Figure \ref{fig:log_signature}). One can remark in particular that the power decreases with the order. This observation likely results from the $1/n$ factor appearing in the log-signature formula (Equation (\ref{eq:log})) which makes high-order terms of the log-signature small so that the even-order terms no longer compensate the odd-order terms.

\subsubsection{The Black-Scholes dynamics}\label{sec:bs_dyn}
In the well-known Black-Scholes model \citep{black1973pricing}, the evolution of the stock price $(S_t)_{t\ge 0}$ is modelled using the following dynamics: 
\begin{equation}
dS_t = \mu S_t dt + \gamma(t)S_t dW_t, \quad S_0=s_0
\end{equation}
where $\mu\in \mathbb{R}$, $\gamma$ is a deterministic function of time and $(W_t)_{t\ge 0}$ is a standard Brownian motion. Because of its simplicity, this model is still widespread in the insurance industry, in particular for the modelling of equity and real estate indices. \\

Like for the fractional Brownian motion, we want to compare two parametrizations of this model that share the same one-year marginal distribution. For this purpose, we consider a Black-Scholes dynamics (BSd) with drift $\mu$ and constant volatility $\sigma$ and a BSd with the same drift $\mu$ but with a deterministic volatility $\gamma(t)$ satisfying $\int_0^1 \gamma^2(s)ds = \sigma^2$ which guarantees that the one-year marginal distribution constraint is met. For the sake of simplicity, we take a piecewise constant volatility with $\gamma(t) = \gamma_1$ if $t\in [0,1/2)$ and $\gamma(t)=\gamma_2$ if $t\in [1/2,1]$. In this setting, the objective of the test is no longer to distinguish two stochastic processes with different regularity but two stochastic processes with different volatility which is \textit{a priori} more difficult since the volatility is not directly observable in practice. When $\mu=0$, the first two terms of the signature of the lead-lag transformation have the same asymptotic distribution in the two parametrizations as the time step converges to 0. This is explained in Example \ref{ex:time_reparametrization}. We conjecture that this result extends to the full signature so that the two models cannot be distinguished using the signature. \\

We start by comparing BSd paths with $\mu = 0.05$ and $\sigma = 0.2$ and BSd paths with $\mu=0.05$, $\gamma_1=0$ and $\gamma_2 = \sqrt{2}\sigma$ using the lead-lag transformation. We consider a zero volatility on half of the time interval in order to obtain very different paths in the two samples. Despite this extreme parametrization, we obtain very low powers (barely above 2\% even for a sample size of $m=500$). It seems that the constraint of same quadratic variation in both samples makes the signatures from sample 1 too close from those of sample 2. \\

In order to improve the power of the test, we can consider another data transformation which allows capturing information about the initial one-dimensional path in a different manner. We observed that the time lead-lag transformation (see Section \ref{sec:transf}) allowed to better distinguish the signatures from the two parametrizations given above. However, the order of magnitude of the differentiating coefficients of the signature (i.e. the coefficients of the signature that are materially different between the two samples) was significantly smaller than the one of non-differentiating coefficients so that the former were hidden by the latter when computing $MMD_{m,n}$. To address this issue, we applied a rescaling to all coefficients of the signature to make sure they are all of the same order of magnitude. Concretely, given two samples $\mathcal{S}_1 =\{X_1,\dots,X_m\}$ and $\mathcal{S}_2 = \{Y_1,\dots,Y_n\}$ of $d$-dimensional paths, the rescaling is performed as follows.
\begin{enumerate}
\item For all $i\in \{1,\dots, m\}$ and for all $j\in\{1,\dots,n\}$, compute $S(X_i)$ and $S(Y_j)$.
\item For all $\ell \in \{1,\dots, R\}$ and for all $I=(i_1,\dots,i_{\ell})\in \{1,\dots, d\}^{\ell}$, compute
\begin{equation}
M_I^{\ell} = \max \left( \max_{i=1,\dots,m} \left|\mathbf{X}^{\ell}_{i,I} \right|,\max_{j=1,\dots,n} \left|\mathbf{Y}^{\ell}_{j,I} \right|  \right)
\end{equation}
where $\mathbf{X}^{\ell}_{i,I}$ (resp. $\mathbf{Y}^{\ell}_{j,I}$) is the coefficient at position $I$ of the $\ell$-th term of the signature of $X_i$ (resp. $Y_j$).  
\item For all $i\in \{1,\dots, m\}$, for all $j\in \{1,\dots, n\}$, for all $\ell \in \{1,\dots, R\}$ and for all $I=(i_1,\dots,i_{\ell})\in \{1,\dots, d\}^{\ell}$, compute the rescaled signature as:
\begin{equation}
\mathbf{\hat{X}}^{\ell}_{i,I} = \frac{\mathbf{X}^{\ell}_{i,I}}{M_I^{\ell}} \quad \text{and} \quad \mathbf{\hat{Y}}^{\ell}_{j,I} = \frac{\mathbf{Y}^{\ell}_{j,I}}{M_I^{\ell}}.
\end{equation}
\end{enumerate}
This procedure guarantees that all coefficients of the signature lie within $[-1,1]$. Using this normalization for the time lead-lag rescaling, the power of the test is significantly better than with the plain lead-lag transformation, as shown in Figure \ref{fig:gbm_timeleadlag_sig}. In Figure \ref{fig:gbm_timeleadlag_logsig}, we show that the power can be further improved by considering the log-signature instead of the signature. Note that the increase of the power starts at order 3 which makes sense since order 2 only allows to capture the quadratic variation over $[0,T]$ (which is the same in the two parametrizations) while order 3 allows to capture the evolution of the quadratic variation over time. Alternatively, one can consider, instead of the time lead-lag transformation, the cumulative lead-lag transformation (see Equation (\ref{eq:cum_lead_lag})) on the log-paths (i.e. on $\log S_t$) which provides even better statistical powers as shown in Figure \ref{fig:gbm_cumsumleadlag}. \\
\begin{figure}[ht]
\centering
\begin{subfigure}{.5\textwidth}
  \centering
  \includegraphics[width=0.9\linewidth]{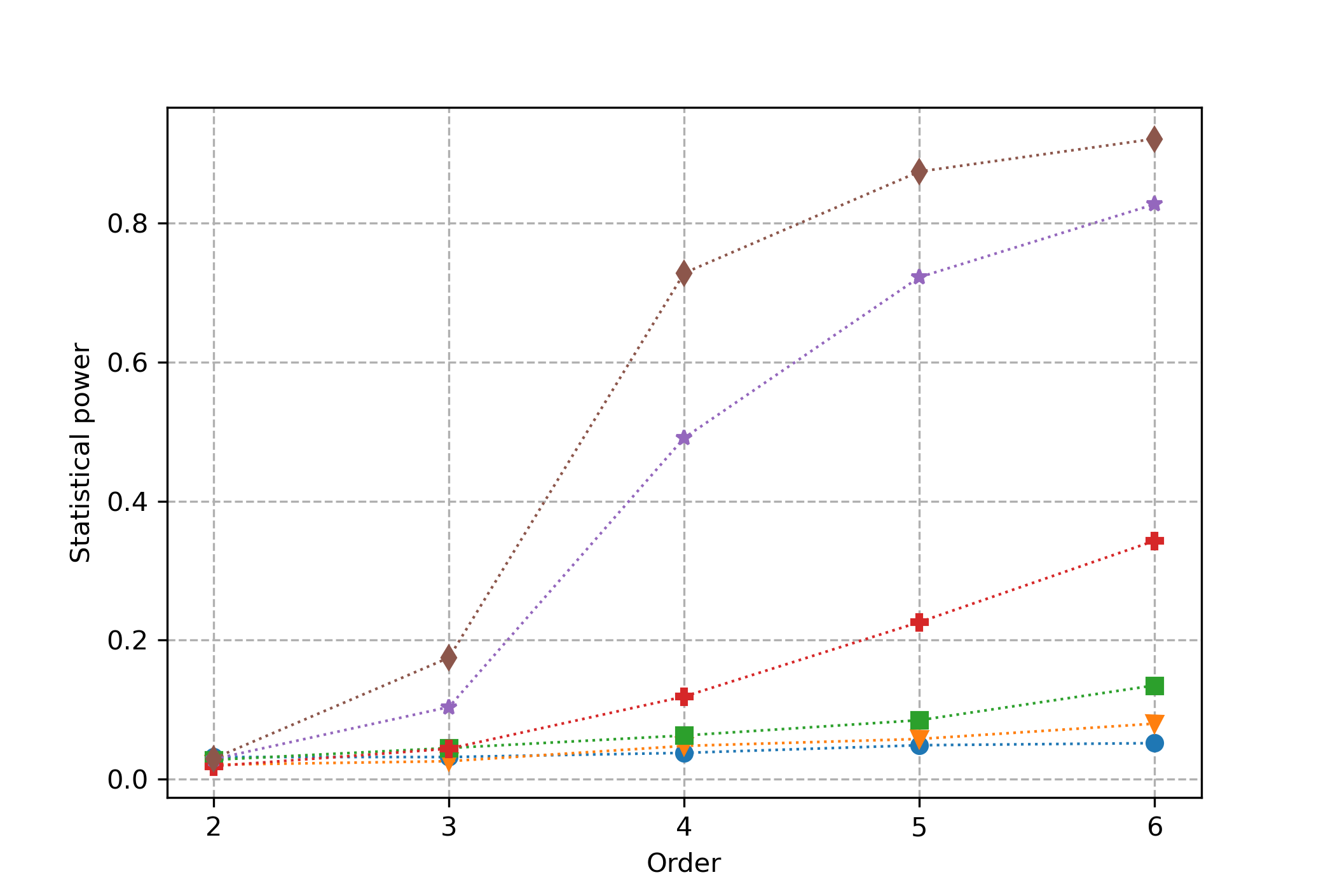}
  \caption{With the signature}
  \label{fig:gbm_timeleadlag_sig}
\end{subfigure}%
\begin{subfigure}{.5\textwidth}
  \centering
  \includegraphics[width=10cm]{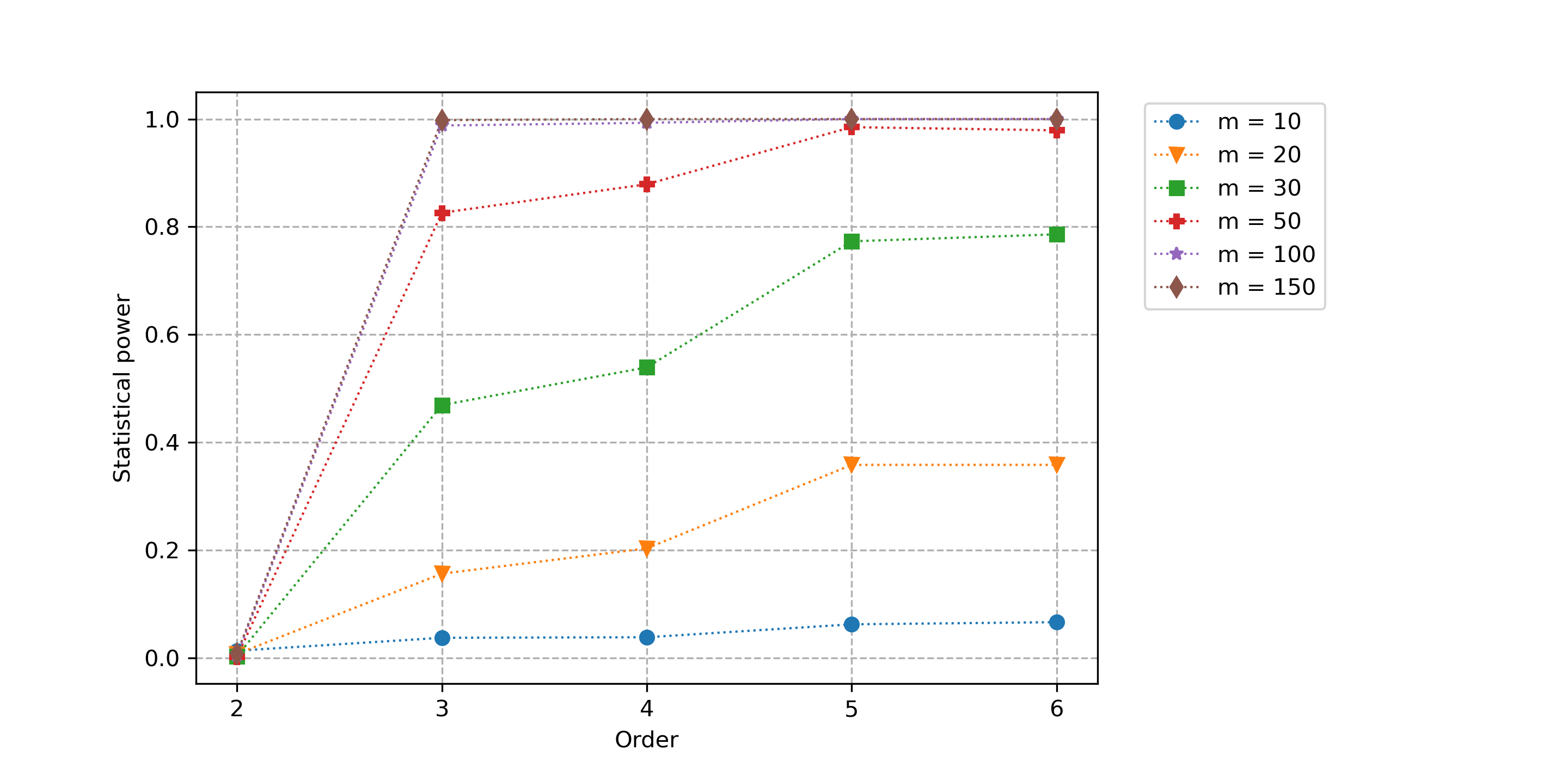}
  \caption{With the log-signature}
  \label{fig:gbm_timeleadlag_logsig}
\end{subfigure}
\caption{Statistical power (as a function of the truncation order and the size of the first sample) of the validation test when comparing constant volatility BSd paths against piecewise constant BSd paths using the time lead-lag transformation with the rescaling procedure.}
\end{figure}
\begin{figure}[ht]
\centering
\includegraphics[scale=0.6]{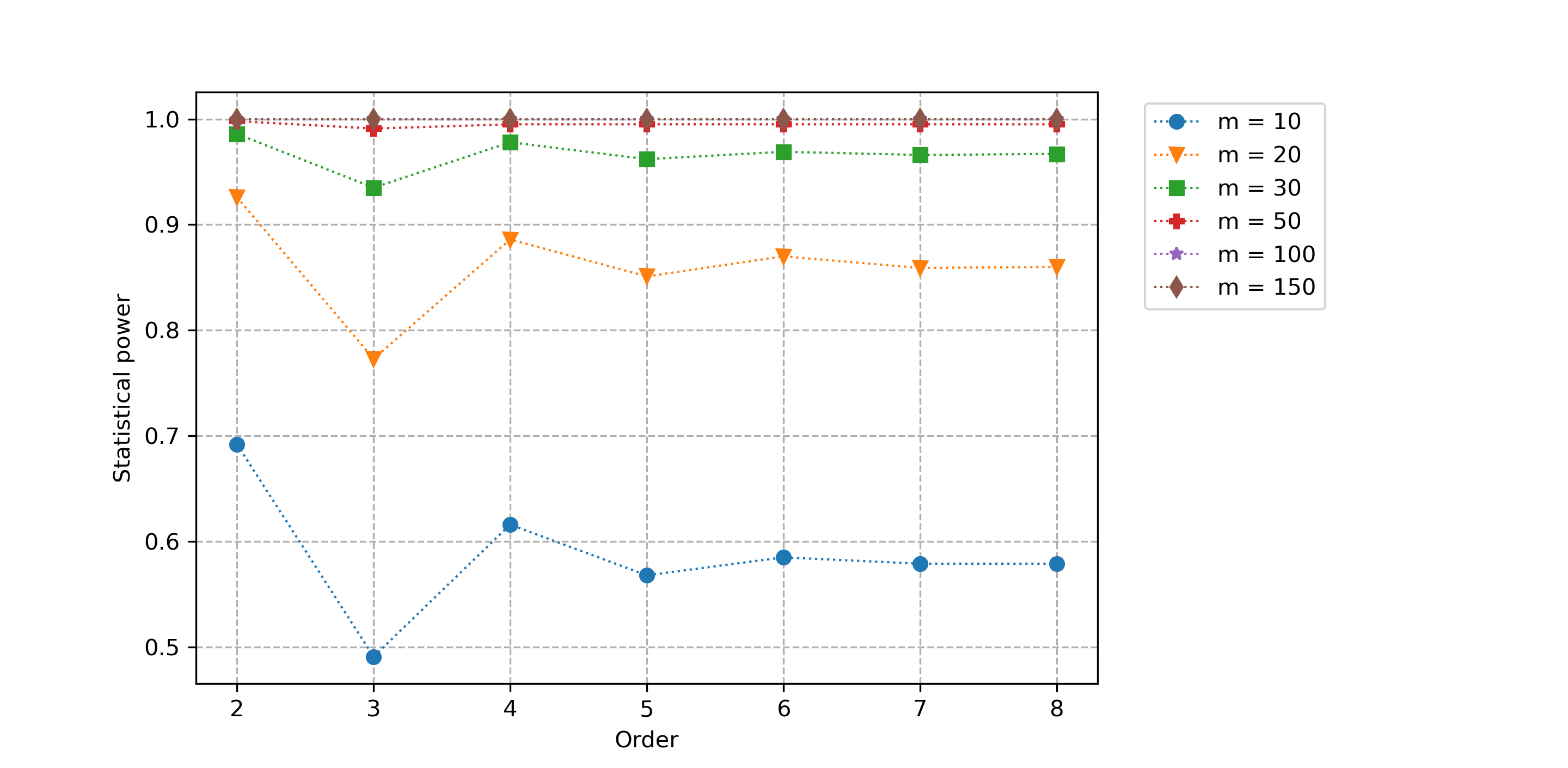}
\caption{Statistical power (as a function of the truncation order and the size of the first sample) of the signature-based validation test when comparing constant volatility BSd paths against piecewise constant BSd paths using the cumulative lead-lag transformation on the log-paths.}
\label{fig:gbm_cumsumleadlag}
\end{figure}

We also considered a slight variation of the BSd that has autocorrelation. Let $(t_i)_{0\le i \le N}$ be an equally-spaced partition of $[0,1]$ with $\Delta t = t_{i+1}-t_i = 1/N$. The autocorrelated discretized BSd $(S^C_{t_i})_{0\le i \le N}$ is defined as follows:
\begin{equation}
\left\{
\begin{array}{rcll}
S^C_{t_{i+1}} &=& S^C_{t_i}\exp\left((\mu^C-\frac{\gamma_{t_i}^2}{2})\Delta t + \gamma_{t_i}\sqrt{\Delta t}G_{i+1} \right),& i \in \{0,\dots, N-1\} \\
S^C_{t_0}&=&s_0&
\end{array}
\right.
\end{equation}
where $(G_i)_{1\le i \le N}$ is a sequence of standard normal random variables satisfying:
\begin{equation}
Cov(G_i,G_j) = 
\left\{
\begin{array}{ll}
1 & \text{if } i = j \\
\rho & \text{if } |i- j|=1  \\
0 & \text{otherwise.}
\end{array}
\right.
\end{equation}
The covariance matrix of $(G_1,\dots,G_N)$ is positive definite when $\rho \in \left(-\frac{1}{2\cos\left(\frac{\pi}{N+1}\right)}, \frac{1}{2\cos\left(\frac{\pi}{N+1}\right)} \right)$. Indeed, the covariance matrix is a tridiagonal Toepliz matrix so its eigenvalues are given by (page 59 in \citeauthor{smith1985numerical}, \citeyear{smith1985numerical}):
\begin{equation}
\lambda_k = 1+2\rho \cos \left(\frac{k\pi}{N+1} \right),\quad k=1,\dots,N.
\end{equation}
In our framework, we have $N=12$ and we can check that $[-0.5,0.5]\subset \left(-\frac{1}{2\cos\left(\frac{\pi}{N+1}\right)}, \frac{1}{2\cos\left(\frac{\pi}{N+1}\right)} \right)$. In Figure \ref{fig:autocorrelated_gBm}, we compare BSd paths with $\mu = 0.05$ and $\sigma = 0.2$ and autocorrelated BSd paths with correlation $\rho \in \{-0.5,-0.4,\dots, 0.5\}$ and a piecewise constant volatility like in the previous setting but with $\gamma_1 = \sigma/\sqrt{2}$ and $(\gamma_2,\mu^C)$ chosen such that $S^C_1$ has the same distribution as $s_0e^{\mu-\sigma^2/2+\sigma W_1}$. Here, the first sample size is fixed to $m=30$. While it was not possible to distinguish BSd paths with different volatility functions using the lead-lag transformation, we observe that the introduction of autocorrelation makes the distinction again possible even with a small sample size. More precisely, except if $\rho \in \{ -0.1, 0, 0.1,0.2\}$, we obtain a power greater than 90\% at order 2. We note however a decrease of the power with the truncation order as it appears that apart from the term of order 2, all the other terms of the signature are very close between the two samples. 
\begin{figure}[ht]
  \centering
  \includegraphics[scale=0.6]{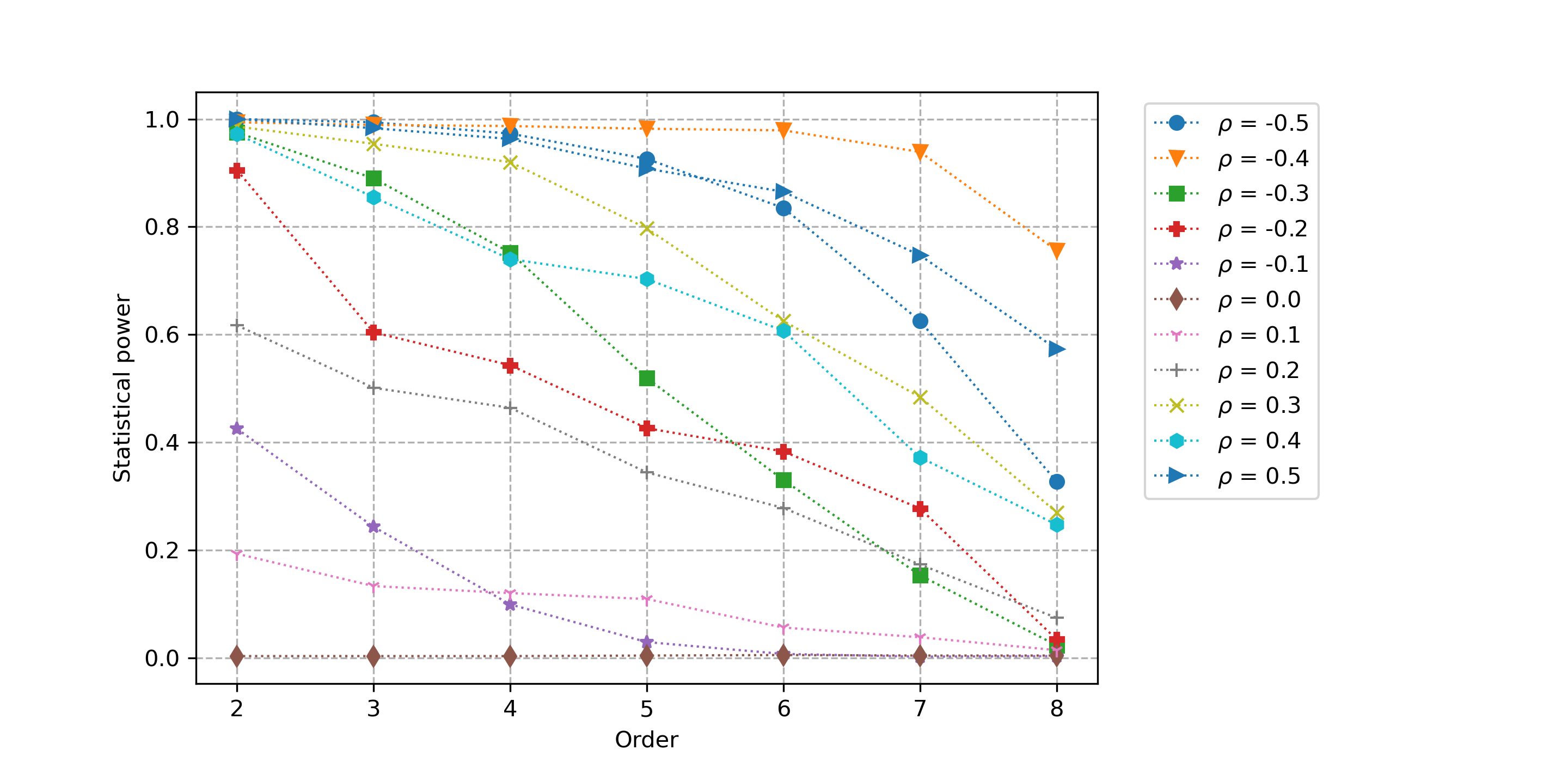}
  \caption{Statistical power (as a function of the truncation order and the correlation of the autocorrelated BSd paths) of the log-signature-based validation test when comparing constant volatility BSd paths against autocorrelated BSd paths using the lead-lag transformation with the rescaling procedure. The size of the first sample is fixed to $m=30$. }
  \label{fig:autocorrelated_gBm}
  \end{figure}

\subsubsection{The rough Heston model}\label{sec:rh}
First introduced as the limit of a microscopic model of prices on financial markets by \cite{eleuch2018}, the rough Heston model can be seen as a fractional extension of the classic Heston model. Indeed, its dynamics can be written as follows:
\begin{equation}\label{eq:rough_heston}
  \begin{array}{rcl}
    dS_t &=& S_t \sqrt{V_t} dW_t \\
    V_t &=& V_0 + \displaystyle\int_0^t \frac{(t-s)^{H-1/2}}{\Gamma(H+1/2)}\left((\theta -\lambda V_s)ds + \sigma\sqrt{V_s}dBs \right)
  \end{array}
\end{equation}
where $\Gamma$ is the gamma function, $S_0,V_0,\theta,\lambda,\sigma >0$, $H\in (0,1/2]$ and $\rho \in [-1,1]$ is the correlation between the Brownian motions $W$ and $B$. Note that for $H=1/2$, $V$ is a Cox-Ingersoll-Ross (CIR) process so that we recover the dynamics of the classic Heston model. Although mainly studied for applications under the risk-neutral probability (pricing and hedging options), the rough Heston model has also been used under the real-world probability in the context of portfolio selection (\citeauthor{han2021}, \citeyear{han2021} and \citeauthor{abijaber2021}, \citeyear{abijaber2021}) and optimal reinsurance (\citeauthor{ma2023}, \citeyear{ma2023}). \\

The purpose of this section is to compare price paths (i.e. paths of $S$ with the above notation) from the rough Heston model to price paths from the classic Heston model. In particular, let us insist on the fact that we do not include the simulations of the volatility process $V$ in the two samples on which the signature-based validation test is applied because this process is not observable in practice. The spirit of this comparison is quite similar to the one described in the previous section: we evaluate the ability of the test to distinguish two stochastic processes having different volatilities. However, the present comparison is more challenging than the one involving the Black-Scholes dynamics since the difference between the two volatilities is more subtle: they differ essentially in their Hölder regularity but the overall evolution is close. We illustrate this point in Figure \ref{fig:rh_illustrations}. Note in particular that the two price paths are very close. These simulations as well as the following ones rely on the approximate discretization scheme of \cite{alfonsi2024} for the rough Heston model and the explicit scheme $E(0)$ of \cite{alfonsi2005} for the CIR in the classic Heston model.\\
\begin{figure}[ht]
  \begin{subfigure}{0.5\linewidth}
    \centering
    \includegraphics[width=\linewidth]{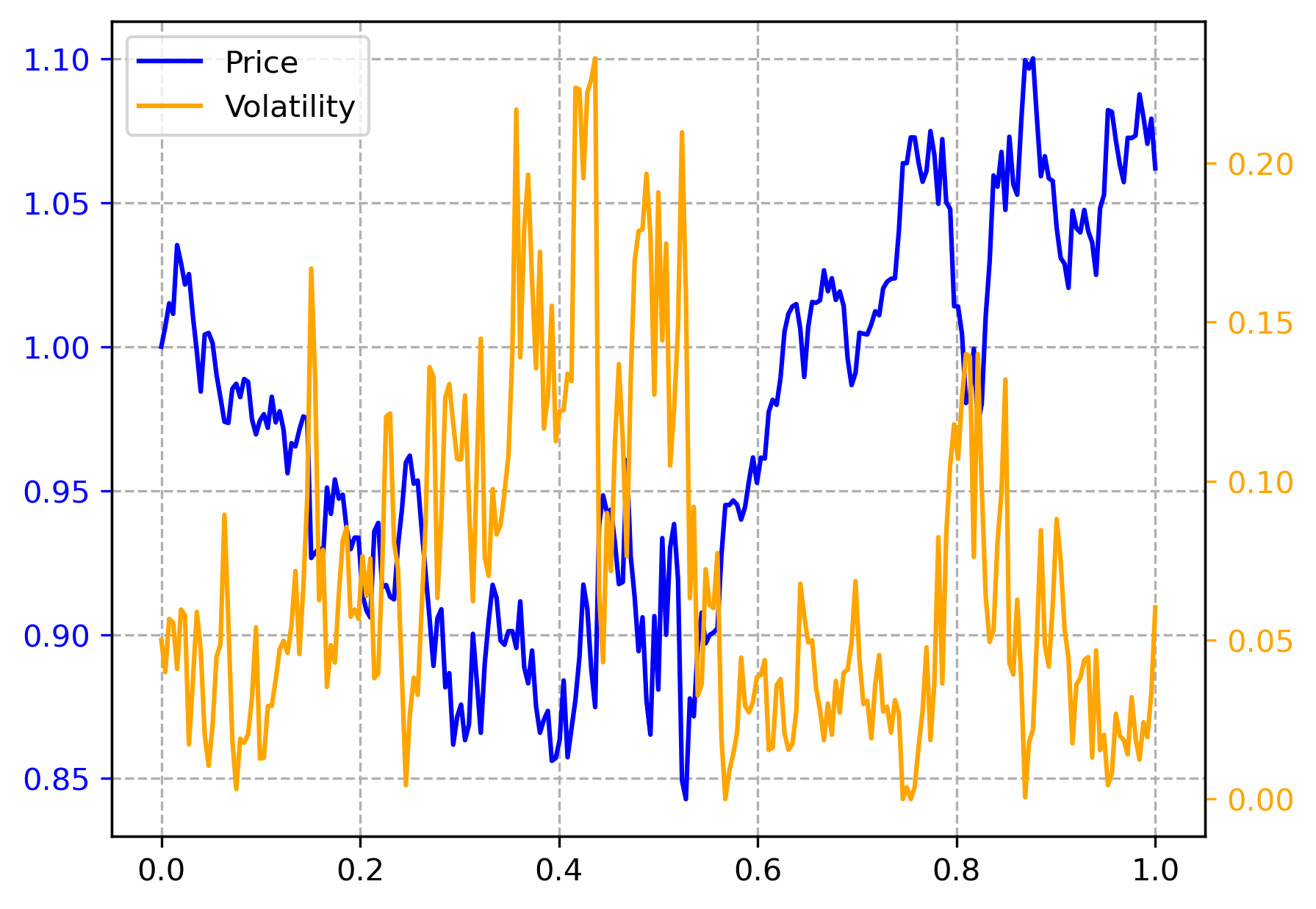}
    \caption{Rough Heston}
  \end{subfigure}
  \begin{subfigure}{0.5\textwidth}
    \centering
    \includegraphics[width=\linewidth]{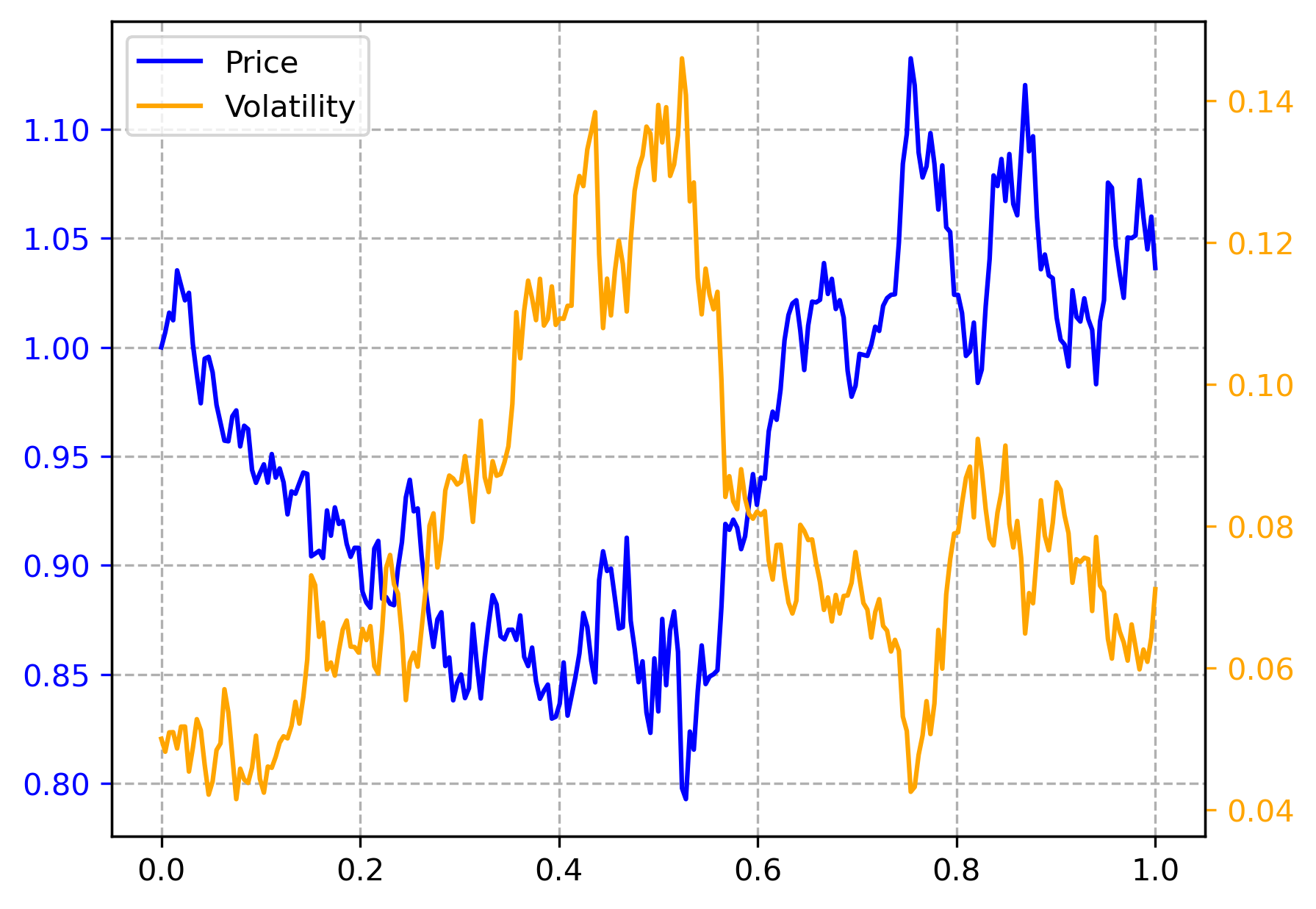}
    \caption{Classic Heston}
  \end{subfigure}
  \caption{Comparison between a sample path of $(S_t,V_t)$ from the rough Heston model with Hurst parameter $H=0.1$ and a sample path from the classic Heston model. We took $S_0=1, V_0=0.05, \theta=0.05, \lambda=0.3, \sigma=0.3$ and $\rho=-0.7$ in both cases. Moreover, both paths are sampled using the same Gaussian random variables so that they look similar. }
  \label{fig:rh_illustrations}
\end{figure}

We start with a comparison between rough Heston paths having a Hurst parameter $H=0.1$ and classic Heston paths. The parameters $S_0, V_0, \theta, \lambda, \sigma$ and $\rho$ are identical between the two models and the retained values are those presented in the caption of Figure \ref{fig:rh_illustrations}. Although this configuration does not imply that both models have exactly the same marginal one-year distributions, the two distributions are very close as shown in Figure \ref{fig:1Y_distr_comparison_rh}. Besides, a two-sample Kolmogorov-Smirnov test (applied on 1000 annual log-returns from each model) yields a $p$-value of 0.31, hence the hypothesis that the two distributions are the same cannot be rejected at standard levels. As a first step, we simulate paths with a time step $\Delta = 1/120$ over a one-year horizon and then we extract every tenth values to obtain monthly trajectories. The use of this finer discretization for the simulation aims at limiting the discretization error. In this setting, we obtain very low powers (below 2\%) be it with the lead-lag transformation, the time lead-lag transformation or the cumulative lead-lag transformations. Besides, neither the use of the log-signature instead of the signature, nor the rescaling procedure described in the previous section, lead to any improvement. Given that the price evolution is extremely close in both models (as shown in Figure \ref{fig:rh_illustrations}), the obtained powers are no surprise. In order to be able to distinguish between the two models, one has to extract somehow the volatility from the price path. However, this is illusory with monthly observations of the price over a period of one year. Therefore, as a second step, we propose to consider daily simulations of the price over one year and to compute a monthly realized volatility as 
\begin{equation}
  RV_{k\Delta} = \sqrt{\sum_{\ell = 0}^{p-1} \left(\log \frac{S_{(k-1)\Delta+(\ell+1)\Delta'}}{S_{(k-1)\Delta+\ell\Delta'}}\right)^2}
\end{equation}
where $k\in \{1,\dots,12\}$, $\Delta = 1/12$, $\Delta'=1/252$ and $p = 21$ ($252=21\times 12$ is the average number of business days in a year). Now, if we apply the lead-lag transformation to the obtained monthly realized volatility paths, we obtain statistical powers that are close to 1 at all truncation orders and for all sample sizes $m\in \{10,20,30,50,100,150\}$ with the signature-based validation test. Note that these results are consistent with those obtained when comparing fBm paths with Hurst parameter $H=0.1$ and standard Brownian motion paths. To complete this numerical study, we also test rough Heston paths with $H\in \{0.2, 0.3, 0.4\}$ against classic Heston paths. The resulting statistical powers are reported in Figure \ref{fig:rh_powers} for $m=30$ and using the log-signature (which slightly improves the powers compared to the signature as for the fBm). As expected, the powers decrease when $H$ gets closer to 0.5 (corresponding to the classic Heston) but we still have very good powers for $H=0.2$. For $H=0.3$ and $H=0.4$, one needs to work with samples of greater size to achieve a statistical power above 95\% ($m\ge 100$ for $H=0.3$ and $m\ge 600$ for $H=0.4$).

\begin{figure}[ht]
  \centering
  \includegraphics[scale=0.6]{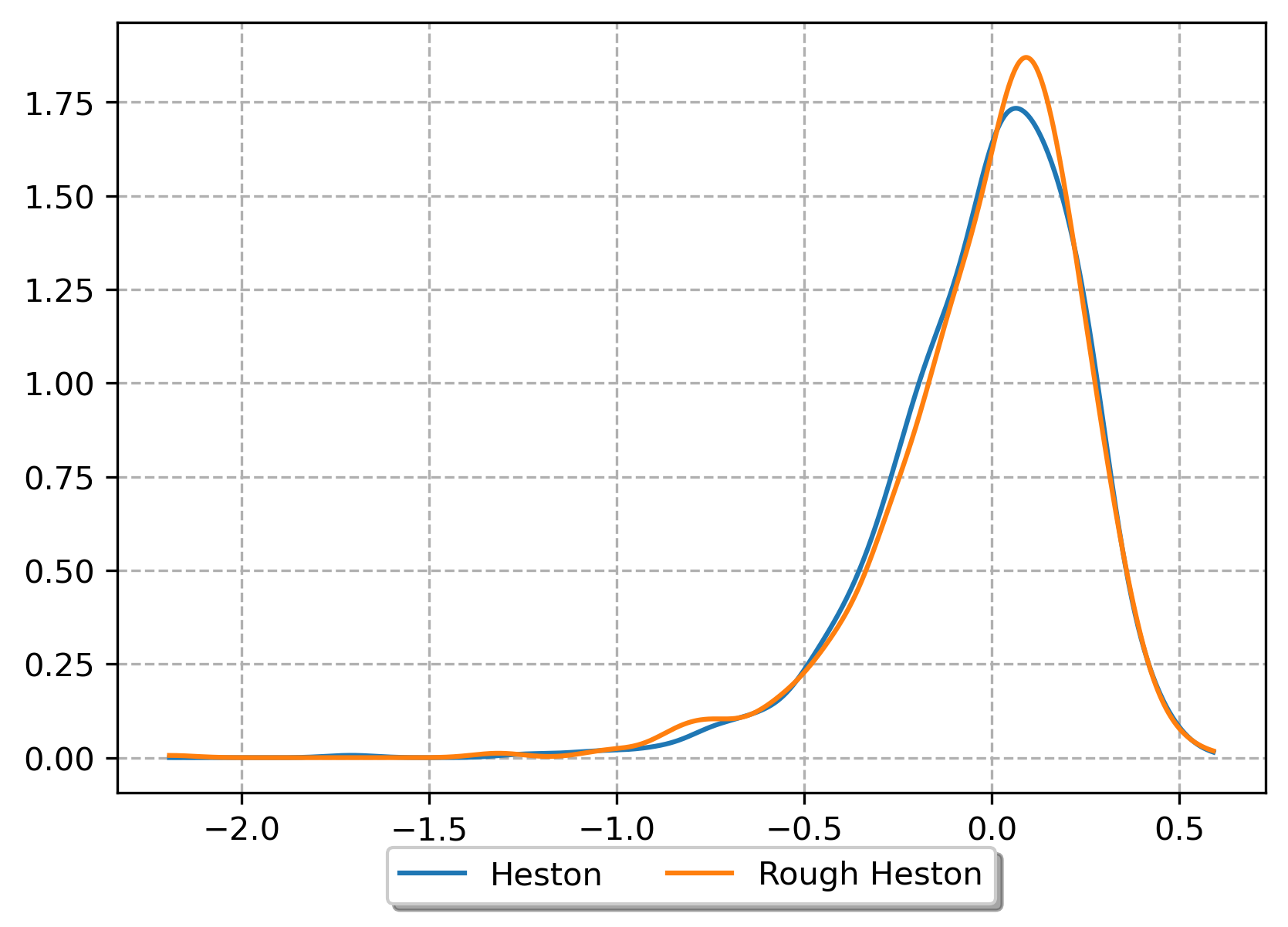}
  \caption{Kernel-density estimates of annual log-returns using Gaussian kernels. The annual log-returns are obtained by simulating 1000 one-year paths with a monthly frequency of each model. }
  \label{fig:1Y_distr_comparison_rh}
\end{figure}
\begin{figure}[ht]
  \centering
  \includegraphics[scale=0.6]{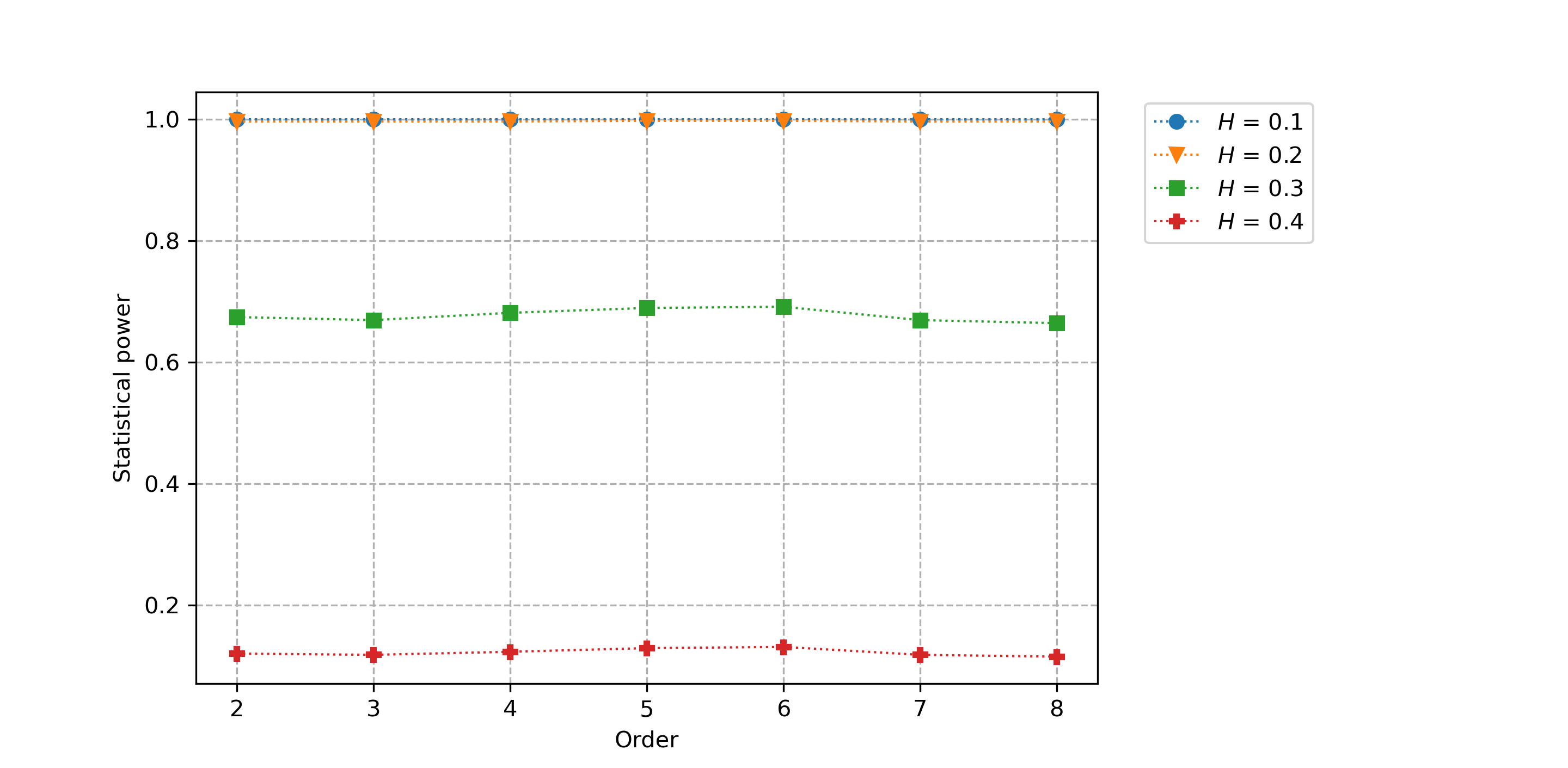}
  \caption{Statistical power of the log-signature-based validation test when comparing rough Heston paths to classic Heston paths with the lead-lag transformation on the monthly realized volatility paths extracted from the daily prices paths. The size of the first sample is fixed to $m=30$.}
  \label{fig:rh_powers}
\end{figure}

\subsubsection{Time series models}\label{sec:time_series}
Let us denote by $I_t$ the value of an index (e.g. the Consumer Price Index) at time $t$ and by $X_t = \log \frac{I_t}{I_{t-\Delta}}$ the log-return of this index over a time interval of size $\Delta$. The first model under study assumes that the log-returns evolve as a regime-switching $AR(1)$ process:
\begin{equation}\label{eq:regime_switching}
X_{t+\Delta} = \mu_{S_{t+\Delta}}+\phi_{S_{t+\Delta}}(X_t-\mu_{S_{t+\Delta}})+\sigma_{S_{t+\Delta}}\epsilon_{t+\Delta}
\end{equation}
where $X_0$ is fixed and $S$ is a time-homogeneous discrete Markov chain with $K\ge 2$ states and whose transition matrix is denoted by $P$. The noises $(\epsilon_t)_{t\ge0}$ are assumed to be i.i.d. standard normal variables. The second model under study assumes that the log-returns are i.i.d. non-centred Gamma noises whose shift, shape and scale parameters are respectively denoted by $\gamma$, $\alpha$ and $\beta$. In the following, we refer to this model as the Gamma random walk. The choice of these two models will be further motivated in Section \ref{sec:test_inflation}. \\

Note that for the regime-switching $AR(1)$ process, the annual log-return $\log \frac{I_1}{I_0}$ is distributed according to a Gaussian mixture while for the Gamma random walk, the annual log-return is Gamma distributed. Therefore, in order to still have close distributions for the annual log-return between the two models, we choose the parameters of the models so that the three first moments of the annual log-return are the same in both models. This is made possible by the fact that we can compute all moments in both cases and the fact that the parameters $\gamma$, $\alpha$ and $\beta$ of the Gamma noises $\epsilon_t$ can be explicitly written as a function of the three first moments of $\epsilon_t$ (see section 17.7 of \citeauthor{johnson1994}, \citeyear{johnson1994}). Therefore, given $X_0$, $(\mu_i)_{1\le i \le K}$,$(\phi_i)_{1\le i \le K}$, $(\sigma_i)_{1\le i \le K}$ and $P$, we can find $\gamma$, $\alpha$ and $\beta$ such that the three first moments of the annual log-return are matched. In Figure \ref{fig:1Y_distr_comparison_bis}, the distributions of the annual log-return are compared for the parameters reported in Table \ref{tab:rsar1_params}. We observe that they are close which is confirmed by a two-sample Kolmogorov-Smirnov test (applied on 1000 simulated one-year monthly paths of each model) yielding a $p$-value of 0.50. \\

\begin{table}[ht]
  \centering
  \caption{Parameters of the regime-switching $AR(1)$ process and the Gamma random walk. The parameters of the former are inspired from parameters calibrated on real inflation data that we present later while the parameters of the latter are obtained by moment-matching as described above.}
  \label{tab:rsar1_params}
  \resizebox{\textwidth}{!}{%
  \begin{tabular}{@{}ccccccccc@{}}
  \midrule
                           & \textbf{$X_0$} & $\mu_1$  & $\mu_2$ & $ \phi_1$ & $\phi_2$ & $\sigma_1$ & $\sigma_2$ & $P$                                               \\ \midrule
  Regime-switching $AR(1)$ & 0              & 0.002    & 0.006   & 0.45      & 0.6      & 0.0025     & 0.004      & $\begin{pmatrix} 0.95   &   0.05 \\0.1   &   0.9 \end{pmatrix}$ \\ \midrule
                           &                &          &         &           &          &            &            &                                                   \\ \cmidrule(r){1-4}
                           & $\gamma$       & $\alpha$ & $\beta$ &           &          &            &            &                                                   \\ \cmidrule(r){1-4}
  Gamma random walk        & -0.6880        & 0.4734   & 1.4534  &           &          &            &            &                                                   \\ \cmidrule(r){1-4}
  \end{tabular}%
  }
  \end{table}
\begin{figure}[ht]
\centering
\includegraphics[scale=0.6]{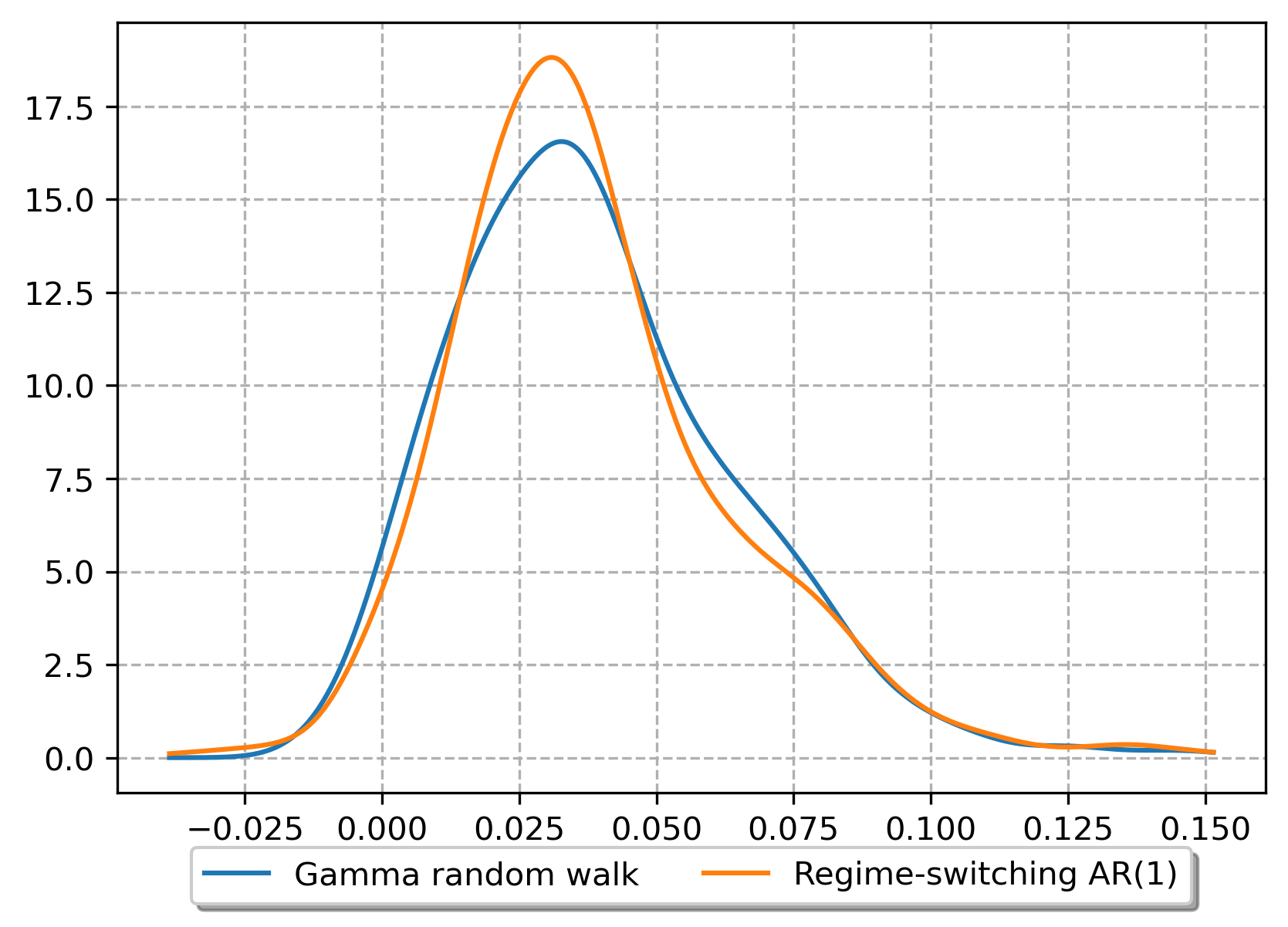}
\caption{Kernel-density estimates of annual log-returns using Gaussian kernels. The annual log-returns are obtained by simulating 1000 one-year monthly paths of each model. }
\label{fig:1Y_distr_comparison_bis}
\end{figure}

In Figure \ref{fig:rsar1_leadlag}, we plot the statistical power of the two-sample test using the lead-lag transformation and the log-signature for an $AR(1)$ process with two regimes and for a Gamma random walk. The parameters of both models are given in Table \ref{tab:rsar1_params}. Note that the regime at $t=0$ is sampled from the stationary distribution of the Markov chain and that we work again on the log-paths as this representation improves slightly the results. We obtain statistical powers that are very close to 1 at any order for a sample size greater than $m=30$. This third case shows that the signature-based validation test is still powerful when working with time series models and allows to distinguish between paths exhibiting changes of regimes over time and first-order autocorrelation from paths with i.i.d. log-returns. Note that we have performed the same experiment with a regime-switching random walk (i.e. $\phi_1=\phi_2=0$) instead of a regime-switching $AR(1)$ process and we obtained statistical powers above 88\% for a sample size greater than $m=50$ which shows that the first-order autocorrelation component is not necessary to distinguish the two models. 

\begin{figure}[ht]
\centering
\includegraphics[scale=0.6]{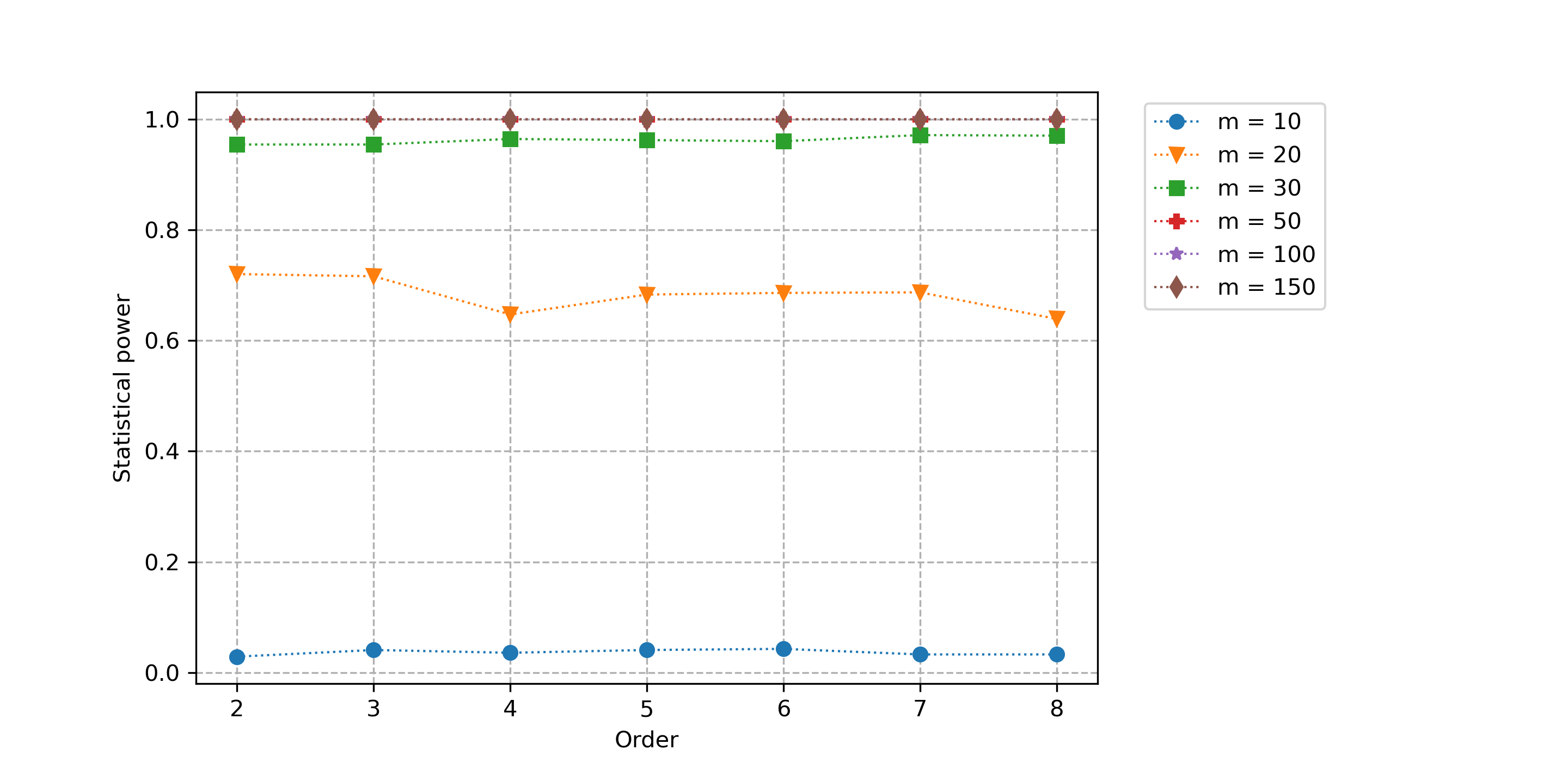}
\caption{Statistical power (as a function of the truncation order and the size of the first sample) of the log-signature-based validation test when comparing regime-switching $AR(1)$ log-paths against Gamma random walk log-paths using the lead-lag transformation.}
\label{fig:rsar1_leadlag}
\end{figure}

\subsubsection{A two-dimensional process}\label{sec:multi}
In the four preceding examples, we compare the paths of one-dimensional processes but in practice economic scenarios consist in joint simulations of several correlated risk drivers. In this subsection, we propose to study whether the signature-based validation test also allows to capture the dependence structure between two models. To this end, we consider the following two samples.
\begin{enumerate}
  \item The first sample consists of paths of the two-dimensional process $Z=(S,I)$ where $(S,V)$ is solution of the rough Heston SDE (\ref{eq:rough_heston}) and the log-returns of $I$ are modelled by a regime-switching $AR(1)$ process. The Brownian motions $W$ and $B$ in Equation (\ref{eq:rough_heston}) are assumed to be independent from the noises $\epsilon$ in Equation (\ref{eq:regime_switching}). Besides, the parameters of each model are those presented in the caption of Figure \ref{fig:rh_illustrations} and in Table \ref{tab:rsar1_params} respectively. 
  \item The second sample consists of paths of the exact same two-dimensional process except that $W$ is correlated with $\epsilon$ (Gaussian dependence). In this case, the process is denoted by $\tilde{Z}$. 
\end{enumerate}
Note that the marginal distributions of the vectors $Z$ and $\tilde{Z}$ are exactly the same making again the two samples very similar.\\

In this multidimensional setting, it is \textit{a priori} no longer necessary to perform a transformation on the simulated paths to extract information beyond the increments over $[0,T]$. However, starting with a correlation of 0.5, the measured statistical power of the plain signature-based validation test is below 2\%. The rescaling procedure helps to improve materially these results as we obtain powers above 60\% for $m=150$. In order to achieve higher powers, especially for lower sample sizes, we considered the signature of the sequence of the log-returns of $S$ and $I$ in the test. This is motivated by the fact that the log-returns are the quantities that are uncorrelated for $Z$ and correlatd for $\tilde{Z}$. This modification allowed to considerably increase the powers for all sample sizes. We also considered the lead-lag transformation of the log-returns of $S$ and $I$ and we further improved the numerical results. In particular, we obtain a statistical power above 98\% for a sample size of $m=30$. In Figure \ref{fig:rh_rsar1_powers}, we show the obtained powers in this configuration for values of the correlation between $W$ and $\epsilon$ in the range $\{-0.5,-0.4,-0.3,-0.2,-0.1,0.1,0.2,0.3,0.4,0.5\}$ when the size of the first sample is fixed to $m=30$. When the absolute value of the correlation is greater or equal to 0.4, the statistical power is above 88\%. Thus, we conclude that the signature-based validation test can also be powerful in a multidimensional setting. Note that the powers are decreasing with the truncation order. This seems to be mainly explained by the fact that the terms of odd order of the signature are essentially the same between the two samples (which can be seen as a consequence of the fact that correlation is captured through quadratic terms) so that the differentiating terms of even order are hidden.

\begin{figure}
  \centering
  \includegraphics[scale=0.6]{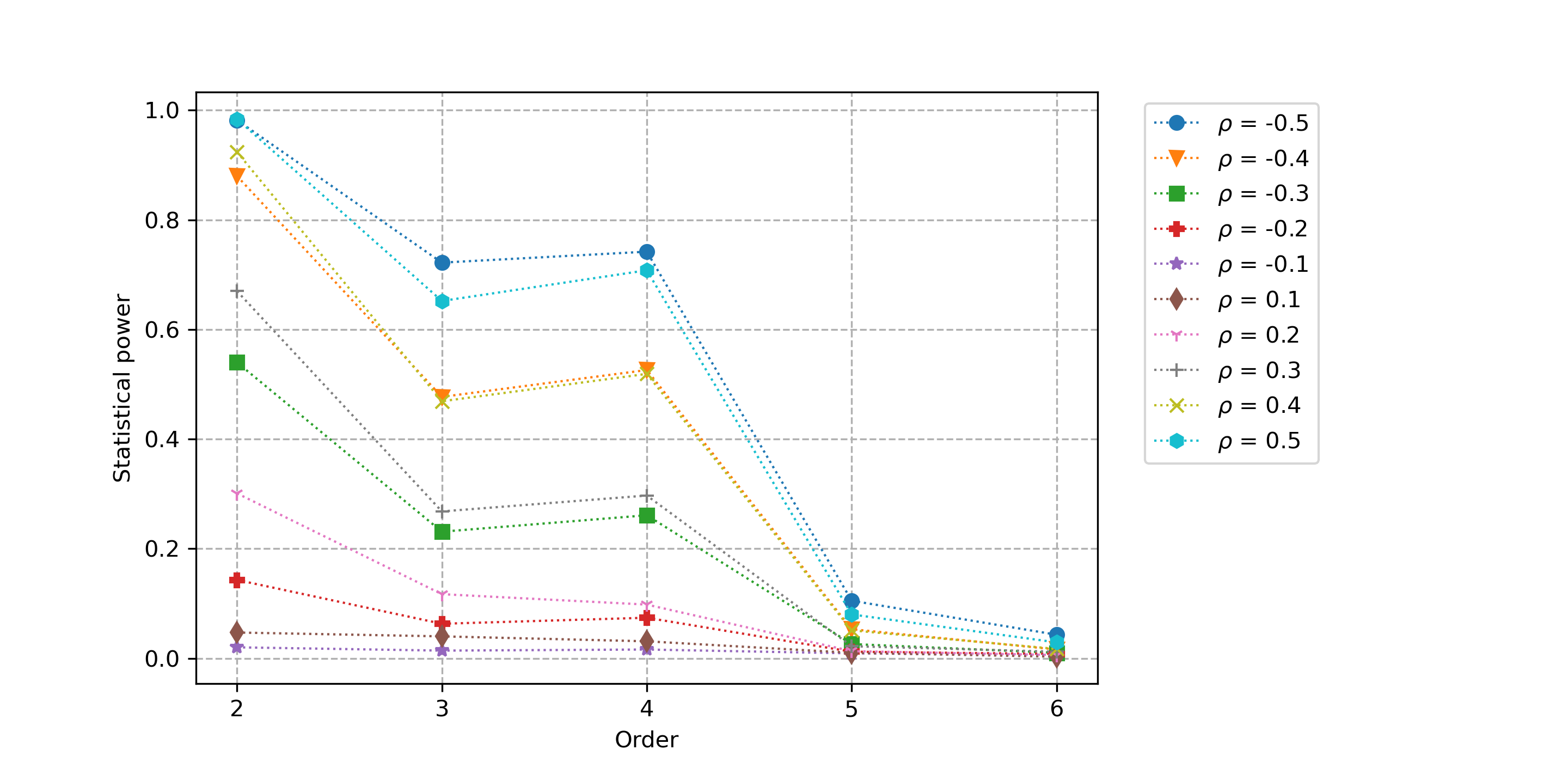}
  \caption{Statistical power (as a function of the truncation order and the correlation between the two components of $\tilde{Z}$) of the signature-based validation test when comparing paths of $Z$ against paths of $\tilde{Z}$. The lead-lag transformation with the rescaling procedure is applied on the log-returns of each component of $Z$ and $\tilde{Z}$. Note that the size of the first sample is fixed to $m=30$.}
  \label{fig:rh_rsar1_powers}
\end{figure}

\subsection{Application to historical data}\label{sec:hist_data}
The purpose of this subsection is to show that the signature-based validation test is able to discriminate between stochastic models calibrated on historical data. This is of practical interest for the validation task of real-world economic scenarios but more generally it is of interest for academics or practitioners that would like to compare a new model to existing ones based on a criteria of consistency with historical data. The following two subsections present numerical results on a realized volatility data set and on an inflation data set respectively. 

\subsubsection{Realized volatility data}
We consider the daily realized estimates of the S\&P 500 from January 2000 to January 2022 obtained from \citet{Realized_Vol_Data}. On the log-volatilities derived from this data set (illustrated in Figure \ref{fig:vol_hist_data}), we calibrate an ordinary Ornstein-Uhlenbeck (OU) process and a fractional Ornstein-Uhlenbeck (FOU) process whose dynamics is recalled below:
\begin{equation}\label{eq:fOu}
dY_t = \alpha(\theta-Y_t)dt + \sigma dW_t^H
\end{equation}
where $(W_t^H)_{t\ge 0}$ is a fractional Brownian motion with Hurst parameter $H$. The OU process (corresponding to $H=1/2$ in Equation \ref{eq:fOu}) is calibrated using the maximum likelihood estimator of $(\alpha,\theta,\sigma)$ while the FOU process is calibrated using a two-step method. First, the Hurst parameter is estimated using the approach of \citet{gatheral2018volatility}. This approach relies on the following scaling property of the fractional Brownian motion:
\begin{equation}
  \mathbb{E}\left[\left|W^H_{t+\Delta}-W^H_t\right|^q \right] =  \Delta^{qH} \mathbb{E}[|G|^q]
\end{equation}
where $G$ is a standard normal random variable. Thus, using the daily log-volatilities $(y_{k\Delta})_{k=0,\dots,N}$, we estimate the parameters $\xi_q$ for several values of $q$ (we took $q\in \{0.5,1,1.5,2,3\}$ as \citeauthor{gatheral2018volatility}, \citeyear{gatheral2018volatility}) in the following linear regression:
\begin{equation}
  \log \left( \frac{1}{\lfloor \frac{N\Delta}{\Delta'} \rfloor}\sum_{k=1}^{\lfloor \frac{N\Delta}{\Delta'} \rfloor} \left|y_{k\Delta'}-y_{(k-1)\Delta'}\right|^q \right) = \beta_q + \xi_q \log(\Delta') + \epsilon_{\Delta'}
\end{equation}
where $\Delta'\in \{\Delta,2\Delta,\dots,p\Delta\}$ (we took $\Delta=1/252$ and $p=252$), $\beta_q$ is the intercept and $\epsilon_{\Delta'}$ is the noise term. The slopes $\xi_q$ of these regressions are then themselves regressed against $q$. The slope of this regression is our estimate $\hat{H}$ of the Hurst parameter. Alternatively, one could consider the procedures presented in \cite{cont2022rough} and \cite{cont2023quadratic}. Second, the mean-reversion speed $\alpha$, the mean-reversion level $\theta$ and the volatility of volatility $\sigma$ are estimated using the method-of-moments estimators from \citet{wang2023modeling}:
\begin{equation}
  \begin{array}{rcl}
    \hat{\theta} &=& \frac{1}{N}\displaystyle\sum_{k=1}^N y_{k\Delta},\\
    \hat{\alpha} &=& \displaystyle\left(\frac{N^2 \hat{\sigma}^2\hat{H}\Gamma(2\hat{H})}{N\sum_{k=1}^N y_{k\Delta}^2 - \left(\sum_{k=1}^N y_{k\Delta}\right)^2} \right)^{1/(2\hat{H})},\\
    \hat{\sigma} &=& \displaystyle\sqrt{\frac{\sum_{k=1}^{N-2}(y_{(k+2)\Delta}-2y_{(k+1)\Delta}+y_{k\Delta})^2}{N(4-2^{2\hat{H}})\Delta^{2\hat{H}}}}.
  \end{array}
\end{equation}
The calibrated parameters are reported in Table \ref{tab:vol_calib_params}. \\
\begin{figure}[t]
  \begin{subfigure}[t]{.5\textwidth}
    \centering
    \includegraphics[width=\linewidth]{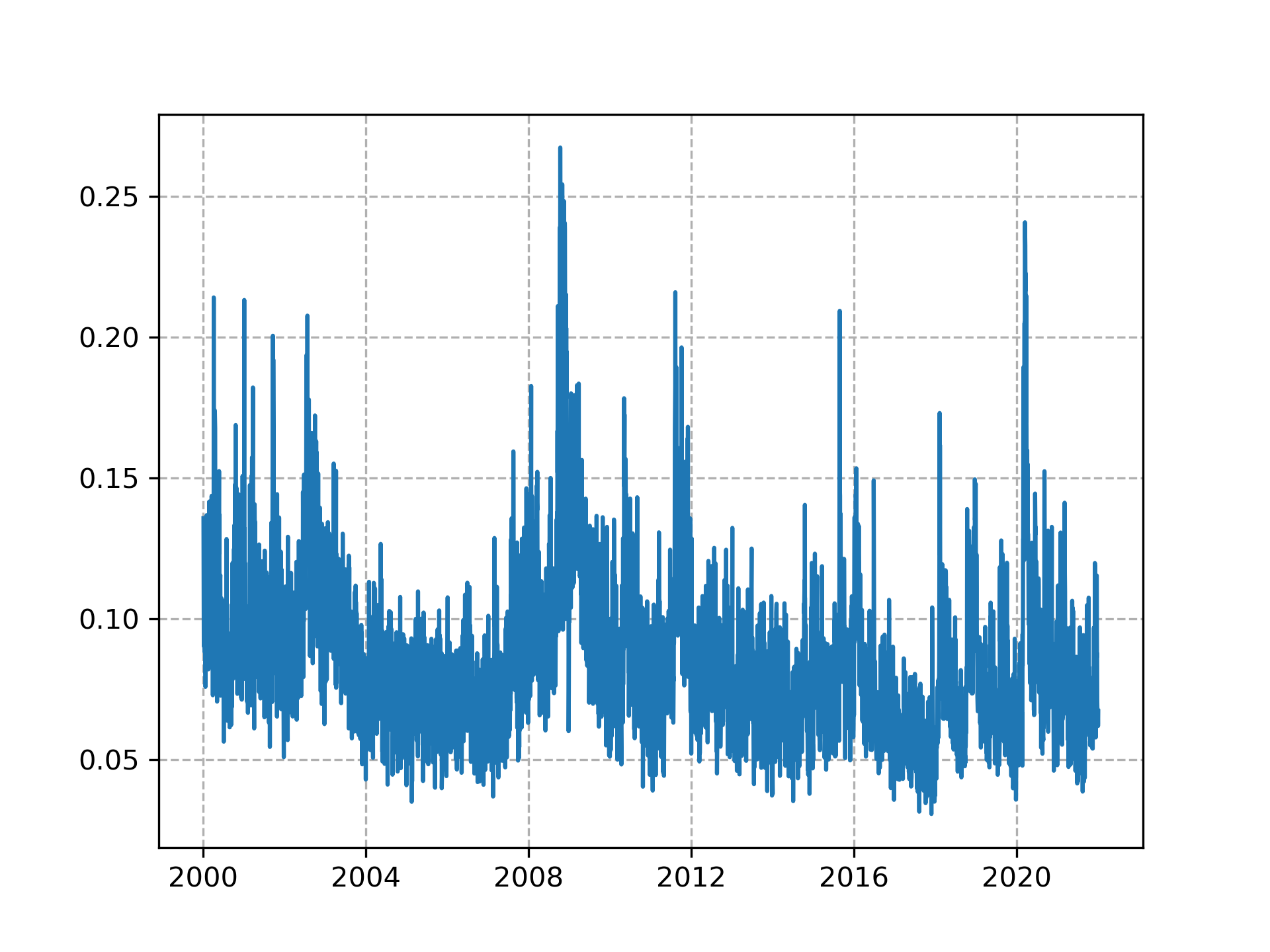}
    \caption{Daily log-volatility of the the S\&P 500 from January 2000 to January 2022.}
    \label{fig:vol_hist_data}
  \end{subfigure}%
  \begin{subfigure}[t]{.5\textwidth}
    \centering
    \includegraphics[width=\linewidth]{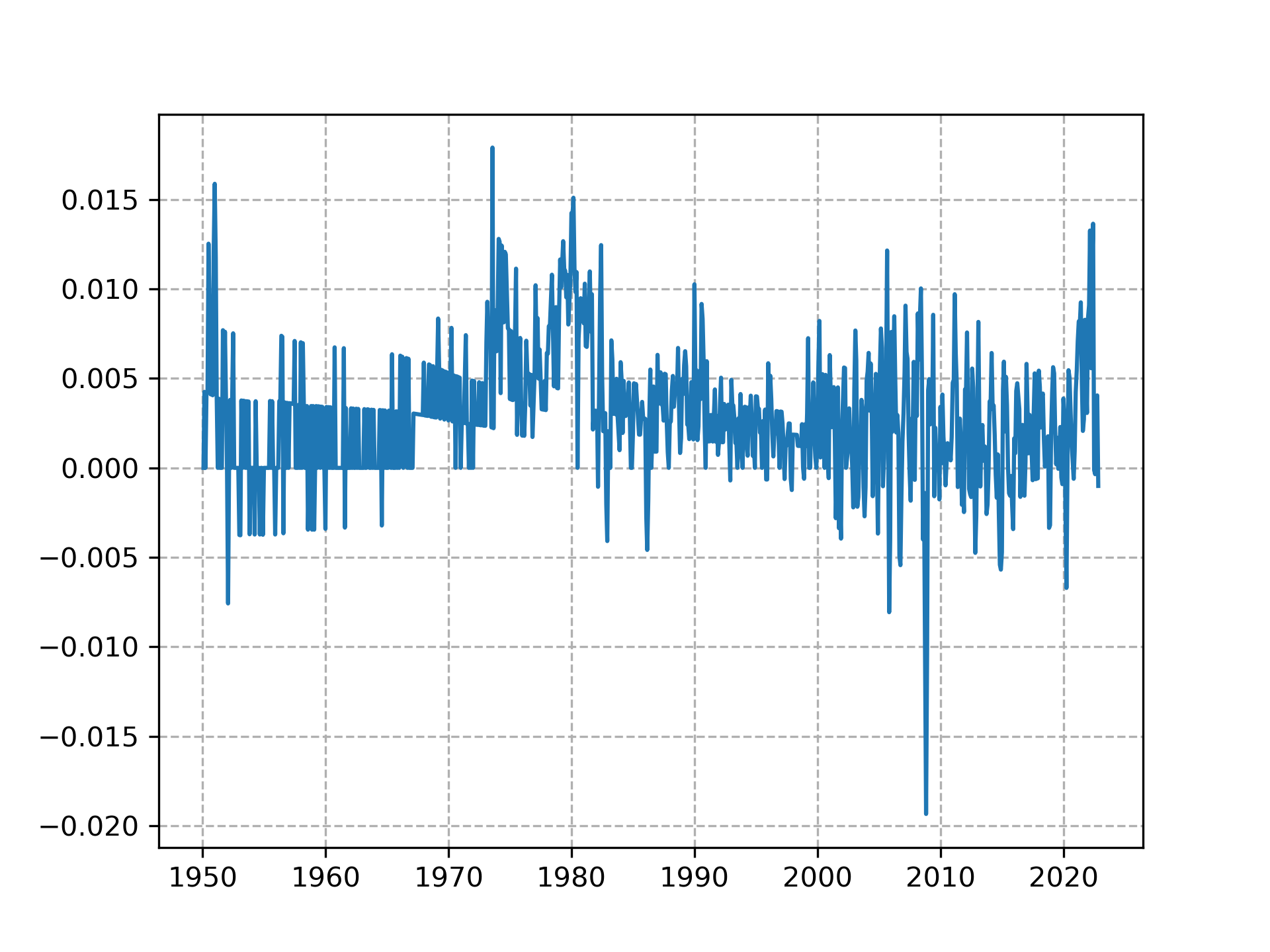}
    \caption{Monthly log-returns of the Consumer Price Index for All Urban Consumers (CPI-U) in the United States from January 1950 to November 2022.}
    \label{fig:inflation_hist_data}
  \end{subfigure}
  \caption{Data sets illustrations. Note on the right figure that between 1950 and the beginning of 70s, the smaller precision in the measurement of the inflation index leads to monthly log-returns that seem to oscillate between some fixed values.}
  \label{fig:hist_data}
  \end{figure}

  \begin{table}[]
    \caption{Parameters calibrated on daily realized log-volatilities of the S\&P 500. Note that the decrease in the $\alpha$ value between the ordinary and the fractional Ornstein-Uhlenbeck models results from the rough noise in latter model ($H<1/2$) that already captures a part of the negative autocorrelation. \\}
    \begin{subtable}{0.45\textwidth}
      \centering
      \caption{Ornstein-Uhlenbeck parameters}
      \begin{tabular}{@{}ccc@{}}
        \toprule
        $\theta$ & $\alpha$ & $\sigma$ \\ \midrule
        -5.0132  & 90.7993  & 8.2209   \\ \bottomrule
        \end{tabular}
    \end{subtable}
    \begin{subtable}{0.45\textwidth}
      \centering
      \caption{Fractional Ornstein-Uhlenbeck parameters}
      \begin{tabular}{@{}cccc@{}}
        \toprule
        $H$    & $\theta$ & $\alpha$ & $\sigma$ \\ \midrule
        0.0916 & -5.0131  & 0.2383   & 0.7876   \\ \bottomrule
        \end{tabular}
    \end{subtable}
    \label{tab:vol_calib_params}
\end{table}

Using the calibrated parameters and considering $Y_0=\theta$, we simulate 10000 one-year paths with a monthly frequency for both models and we check that the simulated annual log-returns are close to the historical ones both graphically and using the two-sample Kolmogorov-Smirnov test. This preliminary verification allows to determine whether it would be possible to reject one of the models based on the marginal distributions. The plotted densities in Figure \ref{fig:1Y_distr_comparison_vol} appear reasonably close and the hypothesis of same distribution is not rejected by the Kolmogorov-Smirnov test (see Table \ref{tab:vol_ks_test}) at standard levels (note however that the $p$-value of the OU paths is below 10\%). \\
\begin{figure}[ht]
  \centering
  \begin{subfigure}{.5\textwidth}
    \centering
    \includegraphics[width=\linewidth]{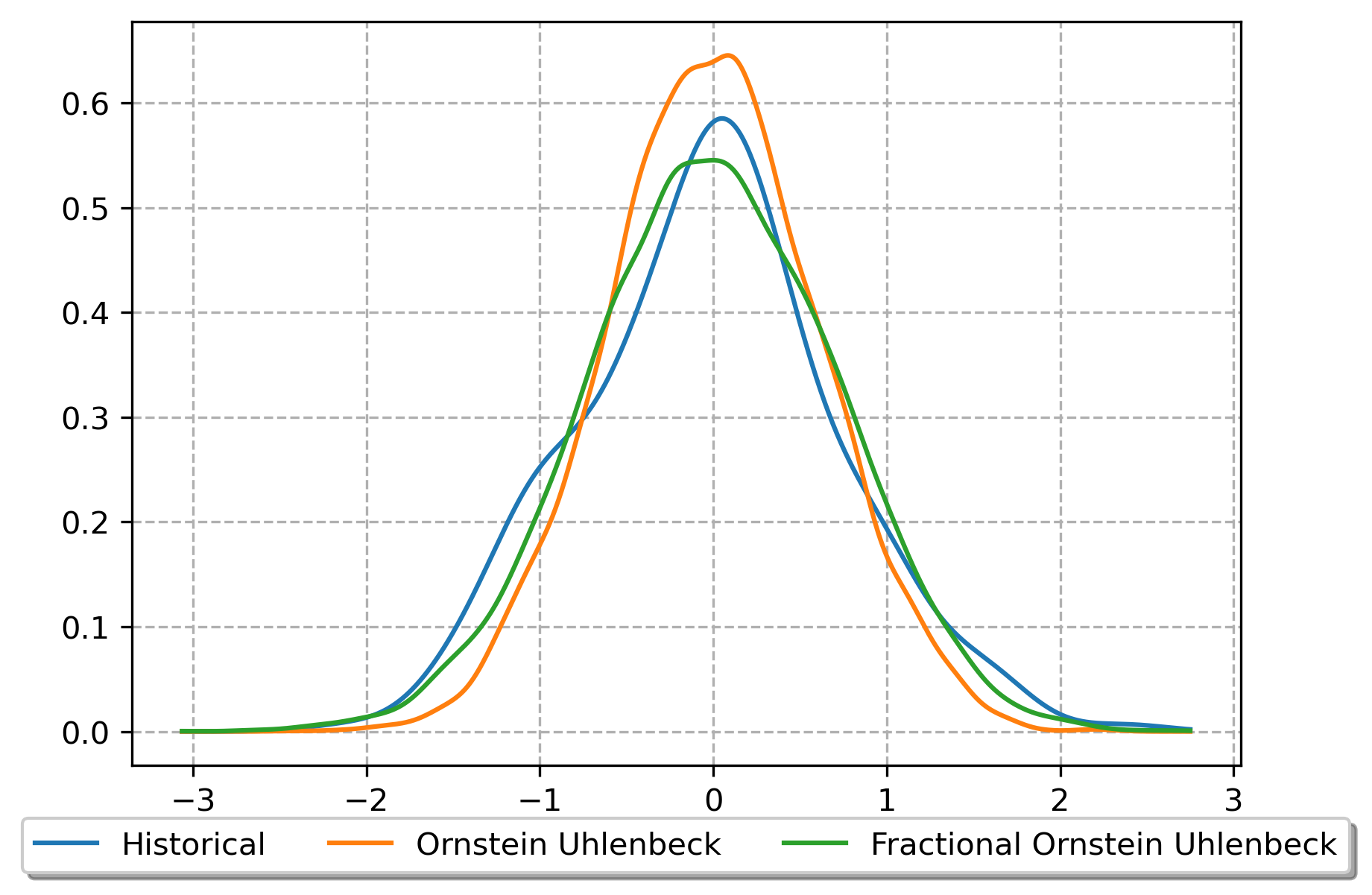}
    \caption{Realized volatility data set.}
    \label{fig:1Y_distr_comparison_vol}
  \end{subfigure}%
  \begin{subfigure}{.5\textwidth}
    \centering
    \includegraphics[width=\linewidth]{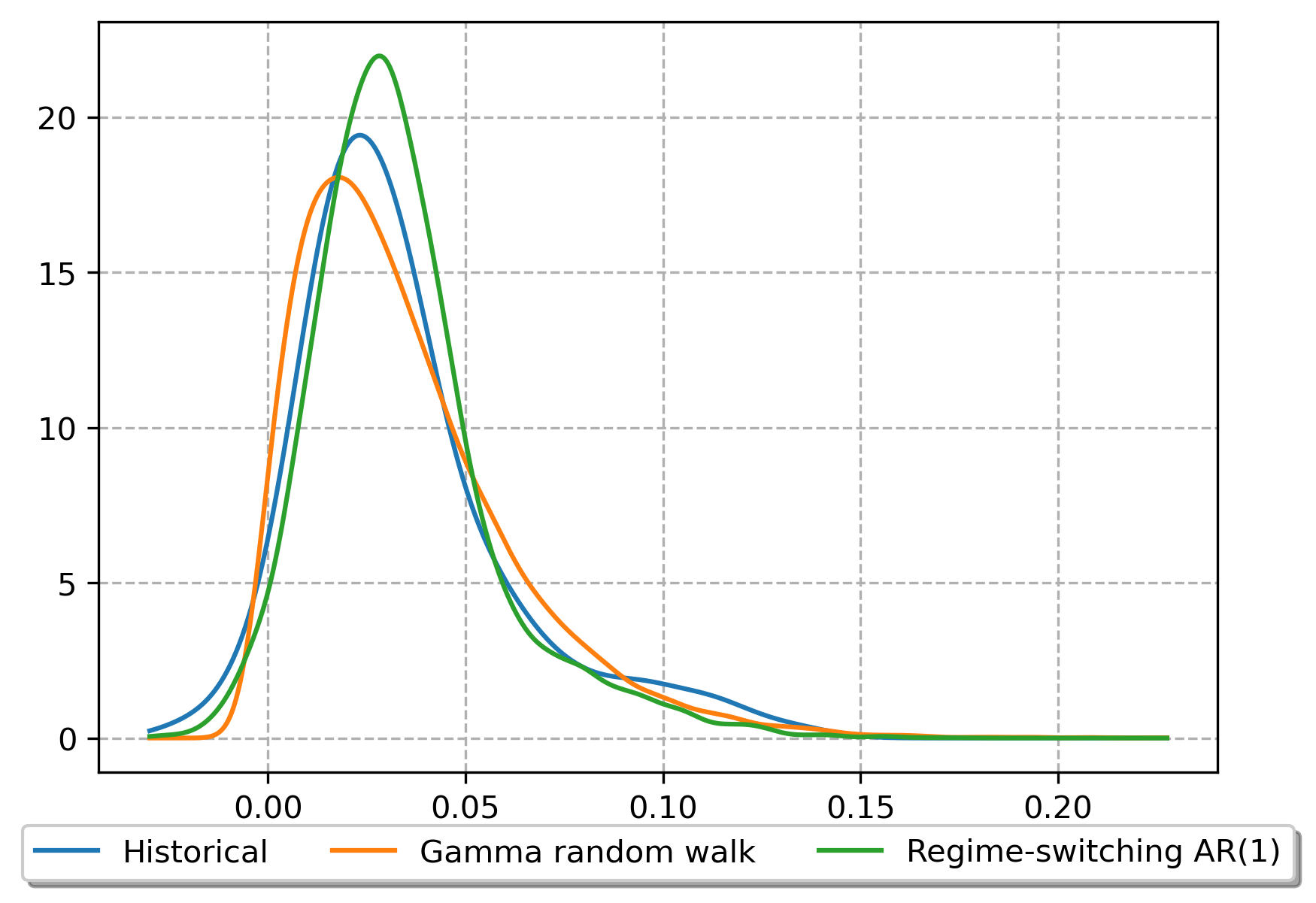}
    \caption{Inflation data set.}
    \label{fig:1Y_distr_comparison_inflation}
  \end{subfigure}
  \caption{Kernel-density estimates of annual log-returns using Gaussian kernels. The historical annual log-returns are computed using a one-year window which is moved with a monthly step for the realized volatility data set and with a quarterly step for the inflation data set.}
  \end{figure}
  \begin{table}[ht]
    \caption{$p$-value of the two-sample Kolmogorov-Smirnov test applied to historical annual log-returns and simulated annual log-returns for each model. The historical annual log-returns are computed using a one-year window which is moved with a monthly step for the realized volatility data set and with a quarterly step for the inflation data set.}
    \begin{subtable}{0.45\textwidth}
      \centering
      \caption{Realized volatility data set}
      \begin{tabular}{@{}cc@{}}
        \toprule
        OU & FOU\\ \midrule
        0.0852             & 0.6042                        \\ \bottomrule
        \end{tabular}%
        \label{tab:vol_ks_test}
      \end{subtable}
      \hfill
      \begin{subtable}{0.45\textwidth}
        \centering
        \caption{Inflation data set}
        \begin{tabular}{@{}cc@{}}
          \toprule
           GRW    & RSAR(1) \\ \midrule
          0.1227 & 0.1302  \\ \bottomrule
          \end{tabular}%
          \label{tab:inflation_ks_test}
      \end{subtable}
  \end{table}

Now, we want to compare the paths simulated by each model to historical data using the signature-based validation test. To this end, we implement the following steps.
\begin{enumerate}
  \item In order to be able to apply the signature-based validation test, we construct one-year "historical paths" with a monthly frequency from the monthly observations $(y_{k\tilde{\Delta}})_{k=0,\dots,\tilde{N}}$ of the realized log-volatility (we keep the last value of each month) as follows: we split the observations $(y_{k\tilde{\Delta}})_{k=0,\dots,\tilde{N}}$ of the realized log-volatility into $m=\lfloor \tilde{N}/12 \rfloor$ one-year historical paths. The $i$-th path consists of $y_{12 i\tilde{\Delta}}$, $y_{(12i+1)\tilde{\Delta}}$, $\dots$, $y_{12(i+1)\tilde{\Delta}}$. Here, $m=22$ which lies in the range studied in the previous subsection. 
  \item The lead-lag transformation is applied to these historical paths and to 1000 simulated sample paths of each calibrated model.
  \item For both models, we test the transformed historical paths against 1000 transformed simulated sample paths using the signature-based validation test with the signature truncated at order 4. Note that the term of order 1 of the signature is removed because it is not very informative given that the increments in both models have close distributions.
  \item The $p$-value of each test, obtained by computing $1-\hat{F}_{H_0}(\hat{MMD}^2_{m,n})$ where $\hat{F}_{H_0}$ is the empirical cumulative distribution function of $MMD_{m,n}^2$ under $H_0$ and $\hat{MMD}^2_{m,n} $ is the test statistic value, is reported in Table \ref{tab:vol_p_values}. 
\end{enumerate}

 \begin{table}[ht]
  \caption{$p$-value of the signature-based validation test applied to historical paths and simulated paths for each model.}
  \begin{subtable}{0.45\textwidth}
    \centering
    \caption{Realized volatility data set}
    \begin{tabular}{@{}cc@{}}
      \toprule
      OU     & FOU    \\ \midrule
      0.0008 & 0.2214 \\ \bottomrule
    \end{tabular}%
    \label{tab:vol_p_values}
  \end{subtable}
  \begin{subtable}{0.45\textwidth}
    \centering
    \caption{Inflation data set}
    \begin{tabular}{@{}cc@{}}
      \toprule
      GRW    & RSAR(1) \\ \midrule
      0.0000 & 0.6919  \\ \bottomrule
      \end{tabular}%
      \label{tab:inflation_p_values}
  \end{subtable}
\end{table}
We observe that the OU model is rejected at any standard level while the FOU model is not. This result shows using a different method than \citet{gatheral2018volatility} that a rough volatility model is consistent with historical data. 

\subsubsection{Inflation data}\label{sec:test_inflation}
The second data set that we consider contains monthly observations of the Consumer Price Index for All Urban Consumers (CPI-U) in the United States from January 1950 to November 2022. It is obtained from the \citet{BLS_Inflation_Data}. As pointed out in Section 3.3.2. of the survey article of \cite{petropoulos2022forecasting}, the state-of-the-art inflation forecasting models are the Dynamic Stochastic General Equilibrium (DSGE) models which rely on the description of the behavior of the main economic agents (households, firms, governments and central banks) and their interactions through macro- and micro-economic concepts derived from economic theory. However, this class of models requires also GDP data, interest rates data, expert judgments, etc. which is beyond the scope of this study. The literature on inflation models requiring only inflation data is rather limited and we are not aware of a comprehensive survey paper comparing the performance of several models. We decide to calibrate a Gamma random walk (GRM) and a regime-switching $AR(1)$ (RSAR(1)) process (see specifications in Section \ref{sec:time_series}) on the CPI-U log-returns since:
\begin{enumerate}
  \item simple models such as random walks models are widely used in the insurance industry in particular for the computation of the Solvency Capital Requirements. The choice of a Gamma distribution allows to capture the positive asymmetry in the distribution of the inflation rate (see Figure \ref{fig:1Y_distr_comparison_inflation}).
  \item A regime-switching model is very natural given that the history of inflation is a succession of periods of low and high inflation of varying lengths. The idea to use such models is not new as it was first proposed for US inflation by \cite{evans1993inflation}. We follow \cite{amisano2013money} in the use of an regime-switching $AR(1)$ process except that we do not make the transition probabilities depend on money growth. 
\end{enumerate}

The GRM is calibrated by matching the three first moments of the historical annual log-returns while the RSAR(1) process is calibrated by log-likelihood maximization. The calibrated parameters are reported in Table \ref{tab:inflation_calib_params}.\\
\begin{table}
 \caption{Parameters calibrated on the log-returns of the CPI-U. We refer to Section \ref{sec:time_series} for the definitions of these parameters.}
  \begin{subtable}{0.4\textwidth}
    \centering
    \caption{Gamma random walk parameters}
    \begin{tabular}{@{}ccc@{}}
      \toprule
      $\gamma$ & $\alpha$ & $\beta$ \\ \midrule
      -0.00047 & 0.18208  & 0.01832 \\ \bottomrule
      \end{tabular}
  \end{subtable}
  \begin{subtable}{0.6\textwidth}
    \centering
    \caption{Regime-switching $AR(1)$ parameters}
    \resizebox{\textwidth}{!}{
    \begin{tabular}{@{}ccccccc@{}}
      \toprule
      $\mu_1$ & $\mu_2$ & $\phi_1$ & $\phi_2$ & $\sigma_1$ & $\sigma_2$ & $P$                                                    \\ \midrule
      0.0022  & 0.0062  & 0.4622   & 0.5774   & 0.0025     & 0.0042     & $\begin{pmatrix} 0.9870    &     0.0130\\0.0637    &    0.9363 \end{pmatrix}$ \\ \bottomrule
      \end{tabular}
    }%
  \end{subtable}
  \label{tab:inflation_calib_params}
\end{table}

Similarly to the previous section, we simulate 10000 one-year paths with a monthly frequency for both calibrated models and we check that the distributions of the annual log-returns (i.e. the annual inflation rates) are close to the historical ones. The comparison of the empirical densities (see \ref{fig:1Y_distr_comparison_inflation}) does not reveal significant deviations from the historical density and the hypothesis of same distribution is not rejected by the Kolmogorov-Smirnov test (see Table \ref{tab:inflation_ks_test}) for both models. Note that the simulated paths all start from 0 and the initial regime for the RSAR(1) process is sampled from the stationary distribution of the Markov chain.\\

Again, we complete this preliminary verification with the signature-based validation test. We implement the same steps described in the previous section (here, we obtain $m=72$ historical paths) except that we work with the log-signature which lead to higher statistical powers on synthetic data. The obtained $p$-values (see Table \ref{tab:inflation_p_values}) show that  the GRM model is rejected at any standard level while the RSAR(1) process is not. This is particularly interesting given that the $p$-values of the two-sample Kolmogorov-Smirnov test reported in Table \ref{tab:inflation_ks_test} are very close. Moreover, this result is in line with the empirical observation that the inflation dynamics exhibits roughly two regimes in the inflation data set (see Figure \ref{fig:hist_data}): a regime of low inflation (e.g. between 1982 and 2021) and a regime of high inflation (e.g. between 1972 and 1982). 

\subsection{Numerical results summary}
In this section, we have first presented the statistical power of the signature-based validation test in five different settings. In each of these settings, we have shown that high statistical powers can be achieved even in a small sample configuration and with a constraint on the closeness of the compared paths by using the following levers:
\begin{enumerate}
\item several representations of the original paths can be considered (log-paths, realized volatility, log-returns);
\item a transformation (lead-lag, time lead-lag, cumulative lead-lag) lifting the chosen representation of the paths to higher dimensional paths can be applied before taking the signature ;
\item several truncation orders can be tested;
\item the signature and the log-signature can be compared;
\item a rescaling of the terms of the signature can be applied.
\end{enumerate}
The combinations of these levers resulting in the highest statistical power over the tested sample sizes in the five settings are presented in Table \ref{tab:best_configs_summary}. Note that items 2, 4 and 5 are part of the \textit{generalized signature method} introduced by \citet{morrill2020generalised} that aims at providing a unifying framework for the use of the signature as a feature in machine learning. A natural question at this stage is how to choose the path representation, transformation, truncation order, whether to use the signature or log-signature and whether to use a rescaling or not in a new setting that is not studied here. Unfortunately, we have not been able to find a general rule especially because it does not seem possible to relate these choices to the properties of the models under study except in some cases that we have exhibited above (e.g. even truncation order should be used for centred Gaussian processes). Despite this, our numerous experiments suggest some heuristics:
\begin{itemize}
  \item the path representation should be chosen to "isolate" a property of the paths allowing to discriminate between models (e.g. extracting the realized volatility is key to distinguish rough Heston paths from classic Heston paths).
  \item Truncating the signature at order 2 or 3 is generally the best choice to maximize the statistical power. 
  \item The rescaling procedure is recommended when the dimension of the paths is greater or equal to 3 since the risk of a discriminating coefficient of the signature being "overwhelmed" by other coefficients is higher. 
\end{itemize}

When no prior knowledge on the data is available, a possible strategy for the practitioners is the following: use the test for every combination of the levers above and reject the null hypothesis if a majority of the tests is rejected. Applying this strategy for the two historical data sets leads to the same conclusions about the four studied models: the ordinary Ornstein-Uhlenbeck and the Gamma random walk are rejected while the fractional Ornstein-Uhlenbeck and the regime-switching $AR(1)$ process are not. \\
 
\begin{table}[ht]
  \centering
  \caption{Configurations of the two-sample test leading to the best statistical power.}
  \label{tab:best_configs_summary}
  \resizebox{\textwidth}{!}{%
  \begin{tabular}{@{}cccccc@{}}
  \toprule
  Setting                                             &  Representation &Transformation      & Order                 & Signature type & Rescaling \\ \midrule
  Fractional Brownian motion (Section \ref{sec:fBm}) &  Original paths & Lead-lag            & 2                     & Log-signature  & No        \\
  Black-Scholes dynamics (Section \ref{sec:bs_dyn}) &  Log-paths    & Cumulative lead-lag & 2                     & Signature      & No        \\
 Rough Heston (Section \ref{sec:rh})  &  Realized volatility &  Lead-lag &  No material influence & Log-signature &  No    \\
  Time series models  (Section \ref{sec:time_series}) &  Log-paths & Lead-lag            & No material influence & Log-signature  & No        \\ 
  Two-dimensional process (Section \ref{sec:multi})&  Log-returns &  Lead-lag & 2  & Signature &Yes \\\bottomrule
  \end{tabular}%
  }
\end{table}

Second, we have shown that the signature-based validation test allows to reject some models while others are not rejected despite the fact that all produce annual log-returns that are reasonably close to the historical ones. In particular, a two-sample Kolmogorov-Smirnov test is not able to distinguish the distribution of the annual log-returns in each model from the historical distribution. As such, the current point-in-time validation approach would not have rejected the ordinary Ornstein-Uhlenbeck model and the Gamma random walk. This demonstrates that the signature-based validation test is a promising tool to validate real-world models.

\section{Concluding remarks}
We propose a new approach for the validation of real-world economic scenarios motivated by insurance applications. This approach relies on the formulation of the problem of validating real-world economic scenarios as a two-sample hypothesis testing problem where the first sample consists of historical paths, the second sample consists of simulated paths of a given real-world stochastic model and the null hypothesis is that the two samples come from the same distribution. For this purpose, we use the statistical test developed by \citet{chevyrev2018signature} which precisely allows to check whether two samples of stochastic processes paths come from the same distribution. It relies on the notions of signature and maximum mean distance which are presented in this article. Our contribution is to study this test from a numerical point of view in settings that are relevant for applications. More specifically, we start by measuring the statistical power of the test on synthetic data under two practical constraints: first, the marginal one-year distributions of the compared samples are equal or very close so that point-in-time validation methods are unable to distinguish the two samples and second, one sample is assumed to be of small size (below 50) while the other is of larger size (1000). To this end, we apply the test to three stochastic processes in continuous time, namely the fractional Brownian motion (fBm), the Black-Scholes dynamics (BSd) and the rough Heston model, and two time series models, namely a regime-switching $AR(1)$ process and a random walk with i.i.d. Gamma increments. The test is also applied to a two-dimensional process combining the price in the rough Heston model and a regime-switching $AR(1)$ process. The numerical experiments have highlighted the need to configure the test specifically for each stochastic process to achieve a good statistical power. In particular, the path representation (original paths, log-paths, realized volatility or log-returns), the path transformation (lead-lag, time lead-lag or cumulative lead-lag), the truncation order, the signature type (signature or log-signature) and the rescaling are key ingredients to be adjusted for each model. For example, the test achieves statistical powers that are close to one in the following settings which illustrate three different risk factors (stock volatility, stock price and inflation respectively): 
\begin{itemize}
\item  fBm paths with Hurst parameter $H=0.1$ against fBm paths with Hurst parameter $H'=0.2$ using the lead-lag transformation and the log-signature;
\item BSd paths with constant volatility against BSd paths with piecewise constant volatility using the time lead-lag transformation and the log-signature with a proper rescaling or using the cumulative lead-lag transformation on the log-paths along with the signature;
\item paths of the rough Heston model with Hurst parameter $H\le0.2$ against paths of the classic Heston model using the lead-lag transformation on the realized volatility extracted from the price paths along with the signature;
\item paths of a regime-switching $AR(1)$ process against paths of a random walk with i.i.d. Gamma increments using the lead-lag transformation on the log-paths along with the log-signature;
\item paths of a two-dimensional process where the two coordinates are independent against paths of the same two-dimensional process where the two coordinates are correlated with correlation (in absolute value) above 40\% using the lead-lag transformation on the log-returns along with the signature and a proper rescaling.
\end{itemize}
In addition to these numerical experiments on synthetic data, we show that the test also performs well on historical data since it rejects some models whereas others are not rejected even if the distributions of the annual log-increments are very close in all the models. For example, we show that the fractional Ornstein-Uhlenbeck model with Hurst parameter around 0.1 is consistent with historical log-volatility of the S\&P 500 while the ordinary Ornstein-Uhlenbeck model is not, which is another piece of evidence that volatility is rough \citep{gatheral2018volatility}. These results indicate that this test represents a promising validation tool for real-world scenarios in a practical framework motivated by insurance applications. More broadly, the test appears as a universal tool for academics and practitioners that would like to challenge a new model against historical data.

\paragraph{Acknowledgments.} The authors are grateful to the anonymous referees for their valuable comments. 
\paragraph{Competing interests.} The authors declare none.

\appendix
\makeatletter  
\renewcommand{\@seccntformat}[1]{Appendix \csname the#1\endcsname .\quad}
\makeatother
\renewcommand{\thesection}{\Alph{section}} 
\makeatother
\section{Some examples}\label{sec:sig_examples}
The signature being already defined in the paper, the purpose of this Appendix is to provide more insights on the signature thanks to examples and a brief overview of its main properties. First, we present several examples that allow to better understand the signature and the log-signature. 

\begin{ex}\label{ex:sig_linear_path}
If $X:[0,T]\rightarrow E$ is a linear path , i.e. $X_t = X_0 + (X_T-X_0)\frac{t}{T}$, then for any $n\ge 0$:
\begin{equation}
\mathbf{X}^n = \frac{1}{n!}(X_T-X_0)^{\otimes n}.
\end{equation}
\end{ex}
\begin{ex}
If $E$ is a vector space of dimension 2, the second order term of the signature is given by:
\begin{equation}
\mathbf{X}^2 = \int_0^T\int_0^t dX_s\otimes dX_t = \begin{pmatrix}
\int_0^T\int_0^t dX_s^{(1)} dX_t^{(1)} & \int_0^T\int_0^t dX_s^{(1)} dX_t^{(2)}\\
\int_0^T\int_0^t dX_s^{(2)} dX_t^{(1)} & \int_0^T\int_0^t dX_s^{(2)} dX_t^{(2)}
\end{pmatrix}.
\end{equation}
Note that the difference of the anti-diagonal coefficients of $\mathbf{X}^2$ corresponds, up to a factor $1/2$, to the Lévy area of the curve $t\mapsto (X_t^1,X_t^2)$ which is defined as:
\begin{equation}
\mathcal{A}^{Levy} = \frac{1}{2}\left( \int_0^T (X_t^1-X_0^1)dX_t^2 - \int_0^T (X_t^2-X_0^2)dX_t^1 \right).
\end{equation}
It is the signed area between the curve and the chord connecting the two endpoints (see Figure \ref{fig:levy_area}). 
\end{ex}
\begin{figure}[h]
\centering
\includegraphics[scale=0.6]{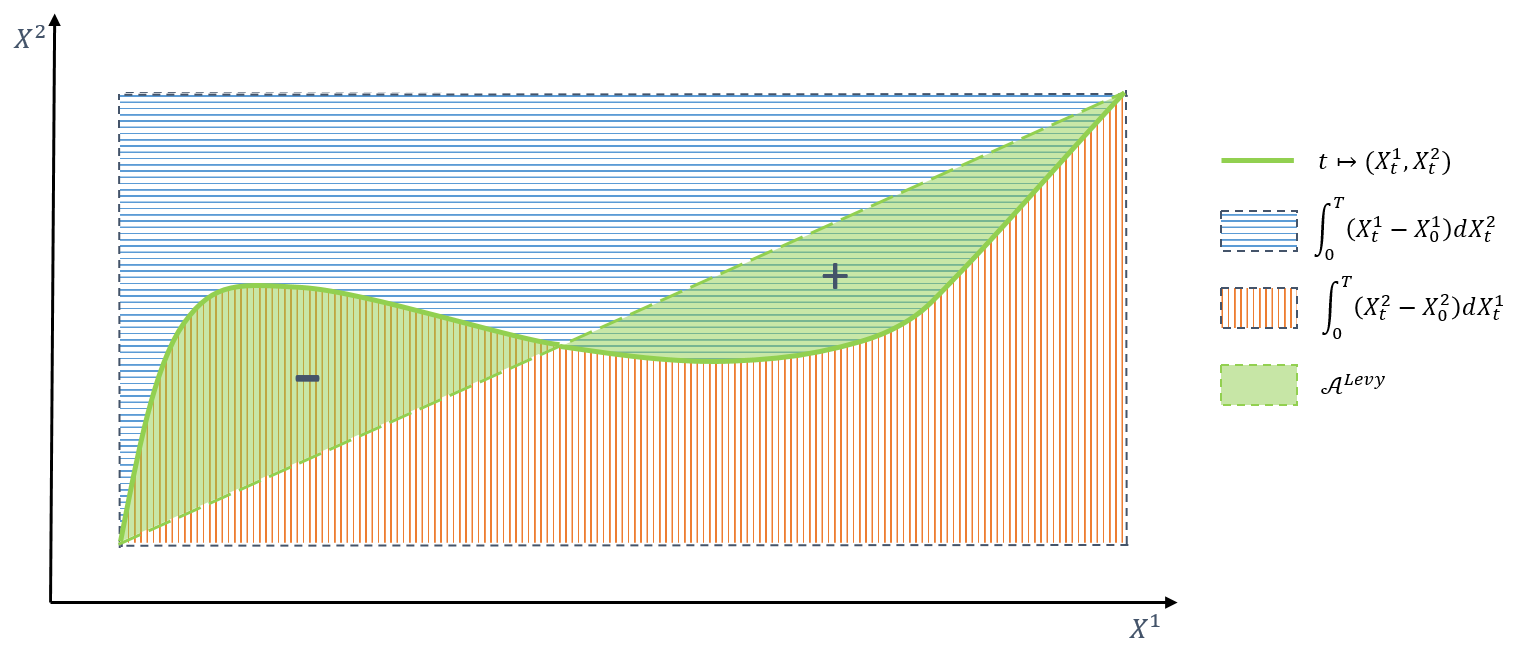}
\caption{Illustration of the Lévy area. The blue dashed area corresponds to the integral $\int_0^T(X_t^1-X_0^1)dX_t^2$ while the red dashed area corresponds to the integral $\int_0^T(X_t^2-X_0^2)dX_t^1$. Taking the difference between these two areas yields the Levy area (represented in green transparent) up to a factor 2 because of a double counting. The '$+$' (resp. '$-$') sign indicates that the surrounding area is counted positively (resp. negatively).  }
\label{fig:levy_area}
\end{figure}

In Section 3.1.2., we mentioned that the lead-lag transformation allows to capture the quadratic variation of a path in the signature. More precisely, the Levy area of the lead-lag transformation is the quadratic variation up to a factor $1/2$ as stated by the following proposition which is a direct consequence of the definition of the lead-lag transformation.  
\begin{prop}\label{prop:levy_area_ll}
Let $t_0=0< t_1 < \dots < t_N=T$ be a partition of $[0,T]$ and $(X_{t_i})_{i=0,\dots,N}$ be the vector of observations of a real-valued process $X$ on this partition. The Levy area of the lead-lag transformation of $(X_{t_i})_{i=0,\dots,N}$ is equal to the quadratic variation of $X$ on the partition $(t_i)_{i=0,\dots,N}$ up to a factor $1/2$, i.e.
\begin{equation}
\frac{1}{2}\left(\int_0^T (X_t^{lead}-X_0^{lead})dX_t^{lag} - \int_0^T (X_t^{lag}-X_0^{lag})dX_t^{lead}\right)  =\frac{1}{2} \sum_{i=0}^{N-1}(X_{t_{i+1}}-X_{t_i})^2.
\end{equation}
\end{prop}

\begin{rk}\label{rk:cum_lead_lag}
We also mentioned in Section 3.1.2. that the cumulative lead-lag transformation $\tilde{X}$ of a sequence of observations $(X_{t_i})_{i=0,\dots,N}$ on $[0,T]$ can be related to the statistical moments of $X$. Indeed, the term of order 1 of the signature of $\tilde{X}$ is given by:
\begin{equation}
\tilde{\mathbf{X}}^1 = \begin{pmatrix}
\tilde{X}_T - \tilde{X}_0 \\
\tilde{X}_T - \tilde{X}_0 \\
\end{pmatrix}
= \begin{pmatrix}
\sum_{i=0}^{N} X_{t_i}  \\
\sum_{i=0}^{N} X_{t_i}  \\
\end{pmatrix}
\end{equation}
which is the empirical mean of $X$ up to a factor $1/(N+1)$. From Proposition \ref{prop:levy_area_ll}, we also deduce that the Levy area of the cumulative lead-lag transformation is given by $\frac{1}{2}\sum_{i=0}^{N} (\tilde{X}_{t_{i+1}}-\tilde{X}_{t_i})^2 =\frac{1}{2}\sum_{i=0}^{N} X_{t_i}^2 $ which is the empirical second order (non-central) moment of $X$ up to a factor $1/(N+1)$. More generally, the $n$-th (non-central) moment of $X$ can be obtained from the term of order $n$ of the signature of the cumulative lead-lag transformation. 
\end{rk}

We have seen in our numerical experiments in Section 3.2.2. that the lead-lag transformation is not always sufficient to distinguish models that are too close from a statistical perspective. In the following example, we show that, as the time step converges to 0, the first two terms of the signature of the lead-lag transformation of a driftless Black-Scholes dynamics with constant volatility have the same distributions as the first two terms of the signature of the lead-lag transformation of a driftless Black-Scholes dynamics with a time-dependent deterministic volatility if the total variances at time $T$ of both models are the same. 
\begin{ex}\label{ex:time_reparametrization}
Consider $X$ and $Y$ the solutions of the following SDE's
\begin{equation}
    dX_t = \sigma X_t dW_t \text{ and } dY_t = \gamma(t) Y_t dW_t
\end{equation}
with $X_0=Y_0=1$ and where $(W_t)_{t\ge 0}$ is a Brownian motion and $\gamma$ is a deterministic function satisfying $\int_0^T\gamma(t)^2dt = \sigma^2T$. The explicit formulas of $X$ and $Y$ write:
\begin{equation}
X_t = \exp\left(\sigma W_t -\frac{1}{2}\sigma^2t \right),\quad Y_t = \exp\left(\int_0^t \gamma(s)dW_s -\frac{1}{2}\int_0^t \gamma(s)^2ds \right).
\end{equation}
Let us denote by $\hat{X}_N$ (resp. $\hat{Y}_N$) the lead-lag transformation of $X$ (resp. $Y$) on a partition $(t_i)_{i=0,\dots,N}$ of $[0,T]$ such that $t_i=iT/N$. The constraint $\int_0^T\gamma(t)^2dt = \sigma^2T$ implies that $X_T \overset{d}{=} Y_T$ so the first order terms of the signatures of $\hat{X}_N$ and $\hat{Y}_N$ (which reduce to the increments of $X$ and $Y$ over $[0,T]$) have the same distribution for all $N\ge 1$. The second order term of the signature of $\hat{Y}_N$ is given by (see the proof of the above proposition):
\begin{equation}
\hat{\mathbf{Y}}_N^2 = \begin{pmatrix}
\frac{1}{2}(Y_T-Y_0)^2 & \sum_{i=0}^{N-1} \left[(Y_{t_{i+1}}-Y_{t_i})^2 + (Y_{t_i}-Y_0)(Y_{t_{i+1}}-Y_{t_i}) \right] \\
\sum_{i=0}^{N-1} (Y_{t_i}-Y_0)(Y_{t_{i+1}}-Y_{t_i}) & \frac{1}{2}(Y_T-Y_0)^2
\end{pmatrix}
\end{equation}
Now, given that $Y$ is a square-integrable continuous martingale, the coefficient at position $(1,2)$ of $\hat{\mathbf{Y}}_N$ converges in probability as $N\to +\infty$ to $\langle Y \rangle_T + \int_0^T(Y_t-Y_0) dY_t =  \frac{1}{2}\left[(Y_T-Y_0)^2 + \langle Y\rangle_T\right]$ where $\langle Y \rangle$ denotes the quadratic variation process of $Y$ and the equality is obtained using the integration by parts formula. Similarly, the coefficient at position $(2,1)$ of $\hat{\mathbf{Y}}_N$ converges in probability as $N\to +\infty$ to $\int_0^T(Y_t-Y_0) dY_t = \frac{1}{2}\left[(Y_T-Y_0)^2 - \langle Y\rangle_T\right]$. The same convergences hold for $\hat{\mathbf{X}}_N$. Now remark that the processes $\left(\int_0^t \gamma(s)dW_s \right)_{t\ge 0}$ and $\left(W_{\int_0^t \gamma(s)^2ds}\right)_{t\ge 0}$ are both Gaussian processes with the same mean and the same covariance function, we deduce that they have the same distribution. Analogously, $(\sigma W_t)_{t\ge 0}$ has the same distribution as $(W_{\sigma^2t})$. We deduce that:
\begin{equation}
(X_t)_{t\ge 0} \overset{d}{=} \left(\exp\left( W_{\sigma^2t} -\frac{1}{2}\sigma^2t \right)\right)_{t\ge 0} \text{ and }
(Y_t)_{t\ge 0} \overset{d}{=} \left(\exp \left(W_{\int_0^t \gamma(s)^2ds} -\frac{1}{2}\int_0^t \gamma(s)^2ds \right)\right)_{t\ge 0}
\end{equation}
Setting $\varphi(t)=\frac{1}{\sigma^2}\int_0^t \gamma(s)^2ds$, we deduce that $(Y_t)_{0\le t \le T} \overset{d}{=} (X_{\varphi(t)})_{0\le t \le T}$. As a consequence, $(Y_t,\langle Y \rangle_t)_{0\le t \le T}$ has the same distribution as $(X_{\varphi(t)},\langle X \rangle_{\varphi(t)} - \langle X \rangle_{\varphi(0)})_{0\le t \le T}$. Since $\varphi(0)=0$ and $\varphi(T)=T$, we conclude that the limit of $\hat{\mathbf{Y}}_N$ has the same distribution as the limit of $\hat{\mathbf{X}}_N$.
\end{ex}

\paragraph{The log-signature \\}
We now introduce more formally the log-signature. We recall that the space of formal series of tensors is defined as:
\begin{equation}
T(E) = \left\{(\mathbf{t}^n)_{n\ge 0} \mid \forall n\ge 0, \mathbf{t}^n \in E^{\otimes n} \right\}
\end{equation}
with the convention $E^{\otimes 0} = \mathbb{R}$. This space can be equipped with the following operations: for $\mathbf{t}$, $\mathbf{u}\in T(E)$, $\lambda \in \mathbb{R}$,
\begin{equation}\label{eq:tensor_operations}
\begin{array}{rcl}
\mathbf{t}+\mathbf{u}&=&(\mathbf{t}^n+\mathbf{u}^n)_{n\ge 0} \\
\lambda \mathbf{t} &=& (\lambda \mathbf{t}^n)_{n\ge 0} \\
\mathbf{t}\otimes \mathbf{u} &=& \left(\mathbf{v}^n=\sum_{k=0}^n \mathbf{t}^k\otimes \mathbf{u}^{n-k} \right)_{n \ge 0}.
\end{array}
\end{equation}
Since by convention the term of order 0 of the signature is set to 1, the signature takes its values in the following affine subspace of $T(E)$:
\begin{equation}
T_1(E) = \left\{\mathbf{t} \in T(E) \mid \mathbf{t}^0=1 \right\}.
\end{equation}
A closely related subspace of $T(E)$ is the following:
\begin{equation}
T_0(E) = \left\{\mathbf{t} \in T(E) \mid \mathbf{t}^0=0 \right\}.
\end{equation}
In fact, there is a bijection between $T_1(E)$ and $T_0(E)$ (Lemma 2.21 in \citeauthor{lyons2007differential}, \citeyear{lyons2007differential}):
\begin{prop}
Let us define respectively the exponential and logarithm mappings as:
\begin{equation}
\begin{array}{rrcl}
\exp: &T_0(E) & \rightarrow& T_1(E) \\
&\mathbf{t} & \mapsto & \exp(\mathbf{t}):=\displaystyle\sum_{n\ge 0} \frac{\mathbf{t}^{\otimes n}}{n!}
\end{array}
\text{ and }
\begin{array}{rrcl}
 \log: &   T_1(E) & \rightarrow& T_0(E) \\
    &\mathbf{t} & \mapsto & \log(\mathbf{t}):=\displaystyle\sum_{n\ge 1} \frac{(-1)^{n-1}}{n}(\mathbf{t}-\mathbf{1})^{\otimes n} 
    \end{array}
\label{eq:log}
\end{equation}
with the convention $t^{\otimes 0}=1$ and where $\mathbf{1}=(1,0,\dots,0,\dots)\in T_1(E)$. The exponential mapping is bijective from $T_0(E)$ to $T_1(E)$ and its inverse is the logarithm mapping. 
\end{prop}

\addtocounter{ex}{-3}
\begin{ex}[continued]\label{ex:logsig_linear_path}
Using the exponential and the logarithm mappings, we can rewrite the signature in Example \ref{ex:sig_linear_path} in the following way:
\begin{equation}
S(X) = \exp(X_T-X_0) 
\end{equation}
where $X_T-X_0$ should be interpreted as the element $(0,X_T-X_0,0,\dots,0,\dots)$ of $T_0(E)$. 
Moreover, 
\begin{equation}
\log(S(X)) = X_T-X_0. 
\end{equation}
\end{ex}
\addtocounter{ex}{2}

Using the logarithm, it is therefore possible to define the log-signature of a path $X$ as $\log(S(X))$. Although there is a one-to-one correspondence between the signature and the log-signature, the log-signature is a more parsimonious representation of the path than the signature in the sense that it removes the redundancies. This can be seen in Example \ref{ex:logsig_linear_path}: the only non-zero term of the log-signature of a linear path is the term of order 1 which contains the increments of the path. In comparison to the signature, all the powers of the increments have disappeared. However, no information is lost. More generally, it can be shown (see for example \citeauthor{liao2019learning}, \citeyear{liao2019learning}) that the log-signature has more zeros than the signature. As such, it represents a useful object for applications as it allows to avoid the exponential increase of the size of the truncated signature with the order. Indeed, if $E$ is a vector space of dimension $d$, the term of order $n$ of the signature has $d^n$ elements.

\begin{ex}\label{ex:log_sig_levy_area}
Let us consider $X\in \mathcal{C}^1([0,T],\mathbb{R}^2)$. The second order term of the log-signature writes:
\begin{equation}
\mathbf{lX}^2 = \mathbf{X}^2 - \frac{1}{2}\mathbf{X}^1\otimes \mathbf{X}^1
\end{equation} 
where $\mathbf{X}^2$ comes from the first term ($n=1$) of the log series in Equation (\ref{eq:log}) and $\mathbf{X}^1\otimes \mathbf{X}^1$ comes from the second term ($n=2$). We have:
\begin{equation}
\mathbf{X}^2 = \begin{pmatrix}
\frac{(X_T^1-X_0^1)^2}{2} & \int_0^T(X_t^1-X_0^1)dX_t^2 \\
\int_0^T(X_t^2-X_0^2)dX_t^1 & \frac{(X_T^2-X_0^2)^2}{2} 
\end{pmatrix}
\end{equation} 
and
\begin{equation}
\mathbf{X}^1\otimes \mathbf{X}^1 = \begin{pmatrix}
(X_T^1-X_0^1)^2 & (X_T^1-X_0^1)(X_T^2-X_0^2) \\
(X_T^1-X_0^1)(X_T^2-X_0^2) & (X_T^2-X_0^2)^2 
\end{pmatrix}.
\end{equation}
Using the integration by part formula, we obtain:
\begin{equation}
\mathbf{lX}^2 = \underbrace{\frac{1}{2}\left(\int_0^T(X_t^1-X_0^1)dX_t^2-\int_0^T(X_t^2-X_0^2)dX_t^1 \right)}_{\text{Lévy area of $X$}}\begin{pmatrix}
0& 1\\
-1 & 0 
\end{pmatrix}
\end{equation}
Hence, the second order term of the log-signature reduces to the Lévy area. 
\end{ex}
\begin{rk}
Note that only the $N$ first terms of the logarithm series (\ref{eq:log}) contribute to the $N$-th term of the log-signature. Indeed, for $n > N$, the contributions to $E^{\otimes N}$ of $(\mathbf{t}-\mathbf{1})^{\otimes n}$ always involve some product by $(\mathbf{t}-\mathbf{1})^0=0$.  
\end{rk}

\section{Properties}\label{sec:sig_properties}
We have seen in the first subsection that the signature allows to capture some information about the path. A natural question at this stage is how much information about $X$ does the signature of $X$ contain. This subsection aims at answering this question. 

\begin{prop}[Invariance under time reparametrization]\label{prop:time_reparam}
Let $X\in \mathcal{C}^1([0,T],E) $ and consider $\varphi:[0,T]\rightarrow [0,T]$ a non-decreasing surjection. If we set $\tilde{X}_t=X_{\varphi(t)}$, then:
\begin{equation}
S(\tilde{X})=S(X).
\end{equation}
\end{prop}
This first property (see Proposition 7.10 in \citeauthor{friz2010multidimensional}, \citeyear{friz2010multidimensional} for a proof) means that the speed at which the path is traversed is not captured by the signature. The signature is also invariant by translation. Indeed, if we define $\bar{X}_t = x + X_t$, then $d\bar{X}_t = dX_t$ and by definition of the signature we have $S(\bar{X})=S(X)$. The next property we will outline is Chen's identity. Before introducing it, we need the following definition.

\begin{defi}[Concatenation]
Let $X\in \mathcal{C}^1([0,t],E)$ and $Y \in \mathcal{C}^1([t,T],E)$. The concatenation of $X$ and $Y$ is the path in $\mathcal{C}^1([0,T],E)$ defined as:
\begin{equation}
(X*Y)_s = \left\{
\begin{array}{ll}
X_s & \text{if } s\in[0,t]\\
X_{t}+Y_s-Y_{t}& \text{if } s\in[t,T].
\end{array}
\right.
\end{equation}
\end{defi}

\begin{thm}[Chen's identity]\label{thm:chen}
Let $X\in \mathcal{C}^1([0,t],E)$ and $Y\in \mathcal{C}^1([t,T],E)$. Then,
\begin{equation}
S_{[0,T]}(X*Y) = S_{[0,t]}(X)\otimes S_{[t,T]}(Y).
\end{equation}
\end{thm}

A proof can be found in Theorem 2.9 of \citet{lyons2007differential}. A useful application of Chen's identity is the computation of the signature of a piecewise linear path. Let $(t_i)_{0\le i\le n}$ be a subdivision of $[0,T]$ and $X:[0,T]\rightarrow E$ be a path such that for $t\in [t_i,t_{i+1}]$ with $0 \le i \le n-1$,
\begin{equation}
X_t = X_{t_i} + \frac{X_{t_{i+1}}-X_{t_i}}{t_{i+1}-t_i}(t-t_i).
\end{equation}
Then by Chen's identity and by using that $S_{[t_i,t_{i+1}]}(X) = \exp(X_{t_{i+1}}-X_{t_i})$ (since $X$ is linear on each $[t_i,t_{i+1}]$),
\begin{equation}
S_{[0,T]}(X) = \bigotimes_{i=0}^{n-1} S_{[t_i,t_{i+1}]}(X) =  \bigotimes_{i=0}^{n-1} \exp(X_{t_{i+1}}-X_{t_i}).
\end{equation}
In general, the right hand side cannot be simplified to $\exp(X_{T}-X_{0})$ because the tensor product $\otimes$ is not commutative. Another consequence of Chen's identity is the following proposition (Proposition 2.14 in \citeauthor{lyons2007differential}, \citeyear{lyons2007differential}).
\begin{prop}[Time-reversal]
Let $X\in \mathcal{C}^1([0,T],E)$. Define $\overleftarrow{X}$ as $\overleftarrow{X}_t = X_{2T-t}$ for $t\in[T,2T]$. Then,
\begin{equation}
S_{[0,2T]}(X*\overleftarrow{X}) = S_{[0,T]}(X)\otimes S_{[T,2T]}(\overleftarrow{X}) = \mathbf{1}
\end{equation}
where we recall that $\mathbf{1}=(1,0,\dots,0,\dots)\in T_1(E)$. 
\end{prop}  
Because constant paths also have $\mathbf{1}$ as signature, the above proposition implies that $X*\overleftarrow{X}$ has the same signature as constant paths.\\

Due to the invariance by reparametrisation and by translation and the time-reversal property, it is clear that if two paths have the same signature, then they are not necessarily equal. In other words, the signature mapping is not injective. Fortunately, the presented invariances and the time-reversal property are essentially the only cases when paths can differ but have the same signature. To make this precise, we need the notion of tree-like paths.

\begin{defi}[Tree-like path]\label{def:tree_like_path}
A path $X:[0,T]\rightarrow E$ is tree-like if there exists a continuous function $h:[0,T]\rightarrow [0,+\infty[$ such that $h(0)=h(T)=0$ and for all $s,t\in[0,T]$ with $s\le t$:
\begin{equation}\label{eq:tree_like}
\|X_t-X_s\|_E \le h(s) + h(t) -2\inf_{u \in [s,t]} h(u).
\end{equation}
This function is called a height function for the path $X$. 
\end{defi}
\begin{rk}
Note that a tree-like path necessarily satisfies $X_0 = X_T$. Indeed, by Definition \ref{def:tree_like_path}:
\begin{equation}
\|X_T-X_0\|_E \le h(0)+h(T)-2\inf_{u\in[0,T]} h(u) = 0
\end{equation}
because $h(0)=h(T)=0$ and $h$ is non-negative. Therefore, one way to turn a tree-like path into a path that is not tree-like is to consider the path $t\mapsto (t,X_t)$ obtained as the time transformation of $X$.
\end{rk}
As suggested by their name, tree-like paths are paths whose graph looks like a tree (see Figure \ref{fig:tree_like_path}), i.e. an acyclic and connected graph in graph theory and the height function $h$ corresponds to the depth of each node of the tree in a depth-first search. Another equivalent way to see tree-like paths is to see them as paths that can be reduced to a constant path by removing pieces of the form $W*\overleftarrow{W}$. For example, if $X$ and $Y$ are non-constant paths, $X*Y*\overleftarrow{Y}*\overleftarrow{X}$ is an example of tree-like path. 
\begin{figure}
\centering
\includegraphics[scale=0.8]{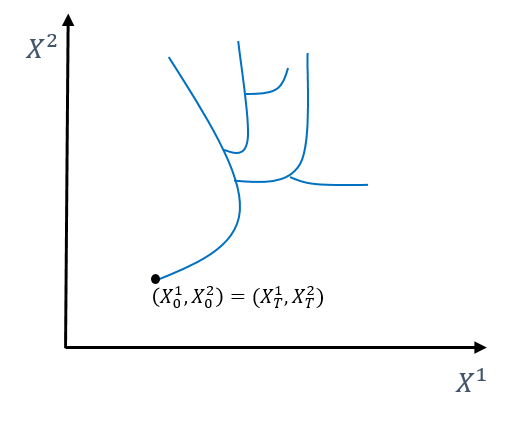}
\caption{Example of a tree like path.}
\label{fig:tree_like_path}
\end{figure}
This notion of tree-like paths is crucial to understand the information that is not captured by the signature as \citet{hambly2010uniqueness} showed that the signature determines the path up to tree-like equivalence, which we will now define.
\begin{defi}[Tree-like equivalence]
For $X$ and $Y$ two paths, we say that $X$ and $Y$ are tree-like equivalent if $X*\overleftarrow{Y}$ is a tree-like path. This relation is denoted by $X\sim_t Y$.
\end{defi}

We can now state Hambly and Lyons's theorem. 
\begin{thm}\label{thm:sig_injectivity}
Let $X \in \mathcal{C}^1([0,T],E)$. Then $S(X)=\mathbf{1}$ if and only if $X$ is a tree-like path. Moreover, if $Y\in\mathcal{C}^1([0,T],E) $ is another bounded variation path, then $S(X)=S(Y)$ if and only if $X\sim_t Y$. 
\end{thm}
This theorem can be understood as follows: two paths will have the same signature if and only if one can be obtained from the second by using translations, by changing the traversal speeds and by removing parts of the form $W*\overleftarrow{W}$. This uniqueness result has then been extended to a more general class of paths (namely weakly geometric rough paths) by \citet{boedihardjo2016signature}. 

\begin{rk}
The conclusion of Theorem \ref{thm:sig_injectivity} still holds if the signature is replaced by the log-signature since the log mapping is a bijection. Note however that the first statement of the theorem should be modified as follows: $\log(S(X)) = \mathbf{0}$ if and only if $X$ is a tree-like path where $\mathbf{0}=(0,\dots,0,\dots)\in T_0(E)$. 
\end{rk}

We have seen that in dimension 1, the signature only captures the path increment between 0 and $T$ (see Example 2.1. in Section 2.2.) so that the signature will only allow to distinguish paths $X$ and $Y$ such that $X_T-X_0\neq Y_T-Y_0$. This result is actually a consequence of the following proposition and of Theorem \ref{thm:sig_injectivity}. 
\begin{prop}
If $E$ is a one-dimensional real vector space and $X$, $Y$ are $E$-valued paths such that $X_T-X_0=Y_T-Y_0$, then $X$ and $Y$ are tree-like equivalent.
\end{prop}
\begin{proof}
Since any one-dimensional real vector space is isometrically isomorph to $\mathbb{R}$, we can assume that $E=\mathbb{R}$. Let $X$ and $Y$ be two paths from $[0,T]$ to $\mathbb{R}$ such that $X_T-X_0=Y_T-Y_0$. Let us set $Z=X*\overleftarrow{Y}$ and $h(t) = |Z_t-Z_0|$ for $t\in [0,2T]$. Using the definition of concatenation operator and the fact that $X_T-X_0=Y_T-Y_0$, we have $Z_0=X_0$ and $Z_{2T}=X_T+Y_0-Y_T=X_0$ so that $h(0)=h(2T)=0$. The non-negativity of $h$ results from the non-negativity of the absolute value. Moreover, the continuity of $X$ and $Y$ imply the continuity of $Z$ by definition of the concatenation operator, so $h$ is continuous as well. The only remaining property to show is inequality (\ref{eq:tree_like}). Let $s,t \in[0,2T]$ with $s\le t$. Let us assume that $Z_s \le Z_t$ (the proof in the case $Z_t \le Z_s$ is similar) so that $|Z_t-Z_s| = Z_t-Z_s = Z_t-Z_0-(Z_s-Z_0)$. We distinguish three cases:
\begin{itemize}
\item If $Z_0 \le Z_s \le Z_t$, then $h(t) = Z_t-Z_0$ and $h(s) = Z_s-Z_0$. Thus,
\begin{equation}
|Z_t-Z_s| = h(t)-h(s)\le h(t) -\inf_{u\in[s,t]}h(u)\le h(t)+h(s) -2 \inf_{u\in[s,t]}h(u).
\end{equation}
\item If $Z_s \le Z_0 \le Z_t$, then $h(t)=Z_t-Z_0$ and $h(s) = Z_0-Z_s$. Thus,
\begin{equation}
|Z_t-Z_s| = h(t)+h(s) =h(t)+h(s) -2 \inf_{u\in[s,t]}h(u)
\end{equation}
because by the intermediate value theorem, there exists $v\in[s,t]$ such that $Z_v = Z_0$ which implies $\inf_{u\in[s,t]}h(u) = 0$.  
\item If $Z_s \le Z_t \le Z_0$, then $h(t)=Z_0-Z_t$ and $h(s)=Z_0-Z_s$. Thus,
\begin{equation}
|Z_t-Z_s| = h(s)-h(t) \le h(s)-\inf_{u\in[s,t]}h(u)\le h(t)+h(s) -2\inf_{u\in[s,t]}h(u).
\end{equation}
\end{itemize}
Hence, $h$ is a height function of $Z$ and $Z$ is tree-like. 
\end{proof}

\section{Signature and stochastic processes}
In the last two subsections, the signature has been presented in a deterministic setting. However, it is clear that the stated results in the previous subsection remain true for stochastic processes by defining the signature as a random variable. In view of the uniqueness theorem from Hambly and Lyons, a natural question at this stage is whether the signature allows to characterize the law of stochastic processes. A first positive answer has been provided by \citet{chevyrev2016characteristic}. They succeeded to construct a characteristic function for the signature of stochastic processes and they proved that it characterizes the law of stochastic processes in the same way as the traditional characteristic function does for random variables. However, this construction is quite abstract and as such is not suitable for applications so far. They also gave some technical conditions under which the expected signature (defined as $\mathbb{E}[S(X)]$ where $X$ is a stochastic process) characterizes the law. \\

These results have then been extended by \citet{chevyrev2018signature}. They showed that by considering a normalization of the signature,   
the expected normalized signature characterizes the law of stochastic processes under mild regularity assumptions. This result is stronger than the one from Chevyrev and Lyons as it requires less assumptions. We now provide a brief description of their main result.

Let us denote by $T_1^*(E)$ the subset of $T^*(E)$ (see Equation (2.10) in Section 2.2.) defined by:
\begin{equation}
T^*_1(E):= \left\{\mathbf{t} \in T^*(E) \mid \mathbf{t}^0=1 \right\}.
\end{equation}
We define a tensor normalization as follows:
\begin{defi}[Tensor normalization]
\label{def:tensor_normalisation}
A tensor normalization is a continuous injective map of the form
\begin{equation}
\begin{array}{rcl}
\Lambda: T_1^*(E)& \rightarrow & \{\mathbf{t}\in T^*_1(E) \mid \|\mathbf{t}\|\le K\} \\
\mathbf{t} &\mapsto &(\mathbf{t}^0,\lambda(\mathbf{t}) \mathbf{t}^1, \lambda(\mathbf{t})^2 \mathbf{t}^2,\dots, \lambda(\mathbf{t})^n \mathbf{t}^n, \dots).
\end{array}
\end{equation}
where $K>0$ is a constant and $\lambda: T^*_1(E) \rightarrow (0,+\infty)$ is a positive function. 
\end{defi} 
The existence of such object is discussed in Proposition 14 of \citet{chevyrev2018signature}. We can now state a simplified version of Chevyrev and Oberhauser's main theorem:
\begin{thm}\label{thm:sig_law_sto}
Let $X=(X_t)_{t\in[0,T]}$ and $Y=(Y_t)_{t\in[0,T]}$ be two stochastic processes defined on a probability space $(\Omega,\mathcal{A},\mathbb{P})$ such that $X$ and $Y$ are in $\mathcal{P}^1([0,T],E)$ almost surely where $\mathcal{P}^1([0,T],E)=\mathcal{C}^1([0,T],E)/\sim_t$ is the space of bounded variation paths quotiented by the tree-like equivalence relation. Let $\Lambda$ be a tensor normalization and define the normalized signature as $\Phi = \Lambda \circ S$. Then,
\begin{equation}
\mathbb{E}[\Phi(X)] = \mathbb{E}[\Phi(Y)] \text{ iff } X \overset{d}{=} Y.
\end{equation}
\end{thm}

\begin{rk}
    This theorem can be extended to a more general space of processes, namely the space of geometric $p$-rough paths quotiented by the tree-like equivalence. This extension corresponds to Theorem 26 in \cite{chevyrev2018signature}.  
\end{rk}

\begin{rk}\label{rk:log-signature_caveat}
The proof of this theorem does not work anymore if we replace the signature by the log-signature. Indeed, one of the key ingredients of the proof is the shuffle product identity (stated and proved in Theorem 2.15 of \citeauthor{lyons2007differential}, \citeyear{lyons2007differential}) which holds for the signature but not for the log-signature.
\end{rk}


\bibliographystyle{abbrvnat}
\bibliography{bibli}

\begin{thebibliography}{64}
\providecommand{\natexlab}[1]{#1}
\providecommand{\url}[1]{\texttt{#1}}
\expandafter\ifx\csname urlstyle\endcsname\relax
  \providecommand{\doi}[1]{doi: #1}\else
  \providecommand{\doi}{doi: \begingroup \urlstyle{rm}\Url}\fi

\bibitem[Abi~Jaber et~al.(2021)Abi~Jaber, Miller, and Pham]{abijaber2021}
E.~Abi~Jaber, E.~Miller, and H.~Pham.
\newblock Markowitz portfolio selection for multivariate affine and quadratic
  {V}olterra models.
\newblock \emph{SIAM J. Financial Math.}, 12\penalty0 (1):\penalty0 369--409,
  2021.
\newblock ISSN 1945-497X.

\bibitem[Akyildirim et~al.(2022)Akyildirim, Gambara, Teichmann, and
  Zhou]{akyildirim2022applications}
E.~Akyildirim, M.~Gambara, J.~Teichmann, and S.~Zhou.
\newblock Applications of signature methods to market anomaly detection.
\newblock \emph{arXiv preprint arXiv:2201.02441}, 2022.

\bibitem[Alfonsi(2005)]{alfonsi2005}
A.~Alfonsi.
\newblock On the discretization schemes for the {CIR} (and {B}essel squared)
  processes.
\newblock \emph{Monte Carlo Methods Appl.}, 11\penalty0 (4):\penalty0 355--384,
  2005.
\newblock ISSN 0929-9629,1569-3961.

\bibitem[Alfonsi and Kebaier(2024)]{alfonsi2024}
A.~Alfonsi and A.~Kebaier.
\newblock Approximation of {S}tochastic {V}olterra {E}quations with kernels of
  completely monotone type.
\newblock \emph{Math. Comp.}, 93\penalty0 (346):\penalty0 643--677, 2024.
\newblock ISSN 0025-5718,1088-6842.

\bibitem[Amisano and Fagan(2013)]{amisano2013money}
G.~Amisano and G.~Fagan.
\newblock Money growth and inflation: A regime switching approach.
\newblock \emph{Journal of International Money and Finance}, 33:\penalty0
  118--145, 2013.

\bibitem[Asadi and Al~Janabi(2020)]{asadi2020}
S.~Asadi and M.~A.~M. Al~Janabi.
\newblock Measuring market and credit risk under {S}olvency {II}: evaluation of
  the standard technique versus internal models for stock and bond markets.
\newblock \emph{Eur. Actuar. J.}, 10\penalty0 (2):\penalty0 425--456, 2020.
\newblock ISSN 2190-9733.

\bibitem[Azran and Ghahramani(2006)]{azran2006new}
A.~Azran and Z.~Ghahramani.
\newblock A new approach to data driven clustering.
\newblock In \emph{Proceedings of the 23rd international conference on Machine
  learning}, pages 57--64, 2006.

\bibitem[Bayer et~al.(2023)Bayer, Hager, Riedel, and Schoenmakers]{bayer2023}
C.~Bayer, P.~P. Hager, S.~Riedel, and J.~Schoenmakers.
\newblock Optimal stopping with signatures.
\newblock \emph{Ann. Appl. Probab.}, 33\penalty0 (1):\penalty0 238--273, 2023.
\newblock ISSN 1050-5164,2168-8737.

\bibitem[Bilokon et~al.(2021)Bilokon, Jacquier, and McIndoe]{bilokon2021market}
P.~Bilokon, A.~Jacquier, and C.~McIndoe.
\newblock Market regime classification with signatures.
\newblock \emph{arXiv preprint arXiv:2107.00066}, 2021.

\bibitem[Black and Scholes(1973)]{black1973pricing}
F.~Black and M.~Scholes.
\newblock The pricing of options and corporate liabilities.
\newblock \emph{J. Polit. Econ.}, 81\penalty0 (3):\penalty0 637--654, 1973.
\newblock ISSN 0022-3808.

\bibitem[Boedihardjo et~al.(2016)Boedihardjo, Geng, Lyons, and
  Yang]{boedihardjo2016signature}
H.~Boedihardjo, X.~Geng, T.~Lyons, and D.~Yang.
\newblock The signature of a rough path: uniqueness.
\newblock \emph{Adv. Math.}, 293:\penalty0 720--737, 2016.
\newblock ISSN 0001-8708.

\bibitem[Boudreault and Panneton(2009)]{boudreault2009}
M.~Boudreault and C.-M. Panneton.
\newblock Multivariate models of equity returns for investment guarantees
  valuation.
\newblock \emph{N. Am. Actuar. J.}, 13\penalty0 (1):\penalty0 36--53, 2009.
\newblock ISSN 1092-0277.

\bibitem[Buehler et~al.(2020)Buehler, Horvath, Lyons, Perez~Arribas, and
  Wood]{buehler2020generating}
H.~Buehler, B.~Horvath, T.~Lyons, I.~Perez~Arribas, and B.~Wood.
\newblock Generating financial markets with signatures.
\newblock \emph{Available at SSRN 3657366}, 2020.

\bibitem[Cartea et~al.(2022)Cartea, P\'{e}rez~Arribas, and
  S\'{a}nchez-Betancourt]{cartea2022}
A.~Cartea, I.~P\'{e}rez~Arribas, and L.~S\'{a}nchez-Betancourt.
\newblock Double-execution strategies using path signatures.
\newblock \emph{SIAM J. Financial Math.}, 13\penalty0 (4):\penalty0 1379--1417,
  2022.
\newblock ISSN 1945-497X.

\bibitem[Chen(1957)]{chen1957}
K.-T. Chen.
\newblock Integration of paths, geometric invariants and a generalized
  {B}aker-{H}ausdorff formula.
\newblock \emph{Ann. of Math. (2)}, 65:\penalty0 163--178, 1957.
\newblock ISSN 0003-486X.

\bibitem[Chevyrev and Kormilitzin(2016)]{chevyrev2016primer}
I.~Chevyrev and A.~Kormilitzin.
\newblock A primer on the signature method in machine learning.
\newblock \emph{arXiv preprint arXiv:1603.03788}, 2016.

\bibitem[Chevyrev and Lyons(2016)]{chevyrev2016characteristic}
I.~Chevyrev and T.~Lyons.
\newblock Characteristic functions of measures on geometric rough paths.
\newblock \emph{Ann. Probab.}, 44\penalty0 (6):\penalty0 4049--4082, 2016.
\newblock ISSN 0091-1798.

\bibitem[Chevyrev and Oberhauser(2022)]{chevyrev2018signature}
I.~Chevyrev and H.~Oberhauser.
\newblock Signature moments to characterize laws of stochastic processes.
\newblock \emph{Journal of Machine Learning Research}, 23\penalty0
  (176):\penalty0 1--42, 2022.

\bibitem[Cohen et~al.(2023)Cohen, Lui, Malpass, Mantoan, Nesheim, de~Paula,
  Reeves, Scott, Small, and Yang]{cohen2023nowcasting}
S.~N. Cohen, S.~Lui, W.~Malpass, G.~Mantoan, L.~Nesheim, A.~de~Paula,
  A.~Reeves, C.~Scott, E.~Small, and L.~Yang.
\newblock Nowcasting with signature methods.
\newblock \emph{arXiv preprint arXiv:2305.10256}, 2023.

\bibitem[Cont and Das(2022)]{cont2022rough}
R.~Cont and P.~Das.
\newblock Rough volatility: fact or artefact?
\newblock \emph{arXiv preprint arXiv:2203.13820}, 2022.

\bibitem[Cont and Das(2023)]{cont2023quadratic}
R.~Cont and P.~Das.
\newblock Quadratic variation and quadratic roughness.
\newblock \emph{Bernoulli}, 29\penalty0 (1):\penalty0 496--522, 2023.

\bibitem[Cuchiero and M{\"o}ller(2023)]{cuchiero2023signature}
C.~Cuchiero and J.~M{\"o}ller.
\newblock Signature methods in stochastic portfolio theory.
\newblock \emph{arXiv preprint arXiv:2310.02322}, 2023.

\bibitem[Cuchiero et~al.(2023{\natexlab{a}})Cuchiero, Gazzani, M{\"o}ller, and
  Svaluto-Ferro]{cuchiero2023joint}
C.~Cuchiero, G.~Gazzani, J.~M{\"o}ller, and S.~Svaluto-Ferro.
\newblock Joint calibration to spx and vix options with signature-based models.
\newblock \emph{arXiv preprint arXiv:2301.13235}, 2023{\natexlab{a}}.

\bibitem[Cuchiero et~al.(2023{\natexlab{b}})Cuchiero, Gazzani, and
  Svaluto-Ferro]{cuchiero2022signature}
C.~Cuchiero, G.~Gazzani, and S.~Svaluto-Ferro.
\newblock Signature-based models: theory and calibration.
\newblock \emph{SIAM J. Financial Math.}, 14\penalty0 (3):\penalty0 910--957,
  2023{\natexlab{b}}.
\newblock ISSN 1945-497X.

\bibitem[Cucker and Smale(2002)]{cucker2002}
F.~Cucker and S.~Smale.
\newblock On the mathematical foundations of learning.
\newblock \emph{Bull. Amer. Math. Soc. (N.S.)}, 39\penalty0 (1):\penalty0
  1--49, 2002.
\newblock ISSN 0273-0979.

\bibitem[Dudley(2002)]{dudley2002real}
R.~M. Dudley.
\newblock \emph{Real analysis and probability}, volume~74 of \emph{Cambridge
  Studies in Advanced Mathematics}.
\newblock Cambridge University Press, Cambridge, 2002.
\newblock ISBN 0-521-00754-2.
\newblock Revised reprint of the 1989 original.

\bibitem[El~Euch et~al.(2018)El~Euch, Fukasawa, and Rosenbaum]{eleuch2018}
O.~El~Euch, M.~Fukasawa, and M.~Rosenbaum.
\newblock The microstructural foundations of leverage effect and rough
  volatility.
\newblock \emph{Finance Stoch.}, 22\penalty0 (2):\penalty0 241--280, 2018.
\newblock ISSN 0949-2984,1432-1122.

\bibitem[Evans and Wachtel(1993)]{evans1993inflation}
M.~Evans and P.~Wachtel.
\newblock Inflation regimes and the sources of inflation uncertainty.
\newblock \emph{Journal of Money, Credit and Banking}, 25\penalty0
  (3):\penalty0 475--511, 1993.

\bibitem[Fawcett(2002)]{fawcett2002problems}
T.~Fawcett.
\newblock \emph{Problems in stochastic analysis: Connections between rough
  paths and non-commutative harmonic analysis}.
\newblock PhD thesis, University of Oxford, 2002.

\bibitem[Fermanian(2021)]{fermanian2021embedding}
A.~Fermanian.
\newblock Embedding and learning with signatures.
\newblock \emph{Comput. Statist. Data Anal.}, 157:\penalty0 Paper No. 107148,
  23, 2021.
\newblock ISSN 0167-9473.

\bibitem[Floryszczak et~al.(2019)Floryszczak, L\'{e}vy~V\'{e}hel, and
  Majri]{floryszczak2019}
A.~Floryszczak, J.~L\'{e}vy~V\'{e}hel, and M.~Majri.
\newblock A conditional equity risk model for regulatory assessment.
\newblock \emph{Astin Bull.}, 49\penalty0 (1):\penalty0 217--242, 2019.
\newblock ISSN 0515-0361.

\bibitem[Friz and Victoir(2010)]{friz2010multidimensional}
P.~K. Friz and N.~B. Victoir.
\newblock \emph{Multidimensional stochastic processes as rough paths}, volume
  120 of \emph{Cambridge Studies in Advanced Mathematics}.
\newblock Cambridge University Press, Cambridge, 2010.
\newblock ISBN 978-0-521-87607-0.
\newblock Theory and applications.

\bibitem[Gatheral et~al.(2018)Gatheral, Jaisson, and
  Rosenbaum]{gatheral2018volatility}
J.~Gatheral, T.~Jaisson, and M.~Rosenbaum.
\newblock Volatility is rough.
\newblock \emph{Quant. Finance}, 18\penalty0 (6):\penalty0 933--949, 2018.
\newblock ISSN 1469-7688.

\bibitem[Graf et~al.(2014)Graf, Haertel, Kling, and Ru\ss]{graf2014}
S.~Graf, L.~Haertel, A.~Kling, and J.~Ru\ss.
\newblock The impact of inflation risk on financial planning and risk-return
  profiles.
\newblock \emph{Astin Bull.}, 44\penalty0 (2):\penalty0 335--365, 2014.
\newblock ISSN 0515-0361.

\bibitem[Gretton et~al.(2009)Gretton, Fukumizu, Harchaoui, and
  Sriperumbudur]{gretton2009fast}
A.~Gretton, K.~Fukumizu, Z.~Harchaoui, and B.~K. Sriperumbudur.
\newblock A fast, consistent kernel two-sample test.
\newblock \emph{Advances in neural information processing systems}, 22, 2009.

\bibitem[Gretton et~al.(2012)Gretton, Borgwardt, Rasch, Sch\"{o}lkopf, and
  Smola]{gretton2012kernel}
A.~Gretton, K.~M. Borgwardt, M.~J. Rasch, B.~Sch\"{o}lkopf, and A.~Smola.
\newblock A kernel two-sample test.
\newblock \emph{J. Mach. Learn. Res.}, 13:\penalty0 723--773, 2012.
\newblock ISSN 1532-4435.

\bibitem[Gyurk\'{o} and Lyons(2010)]{gyurko2010}
L.~G. Gyurk\'{o} and T.~Lyons.
\newblock Rough paths based numerical algorithms in computational finance.
\newblock In \emph{Mathematics in finance}, volume 515 of \emph{Contemp.
  Math.}, pages 17--46. Amer. Math. Soc., Providence, RI, 2010.

\bibitem[Gyurk{\'o} et~al.(2013)Gyurk{\'o}, Lyons, Kontkowski, and
  Field]{gyurko2013extracting}
L.~G. Gyurk{\'o}, T.~Lyons, M.~Kontkowski, and J.~Field.
\newblock Extracting information from the signature of a financial data stream.
\newblock \emph{arXiv preprint arXiv:1307.7244}, 2013.

\bibitem[Hambly and Lyons(2010)]{hambly2010uniqueness}
B.~Hambly and T.~Lyons.
\newblock Uniqueness for the signature of a path of bounded variation and the
  reduced path group.
\newblock \emph{Ann. of Math. (2)}, 171\penalty0 (1):\penalty0 109--167, 2010.
\newblock ISSN 0003-486X.

\bibitem[Han and Wong(2021)]{han2021}
B.~Han and H.~Y. Wong.
\newblock Mean-variance portfolio selection under {V}olterra {H}eston model.
\newblock \emph{Appl. Math. Optim.}, 84\penalty0 (1):\penalty0 683--710, 2021.
\newblock ISSN 0095-4616,1432-0606.

\bibitem[Hardy et~al.(2006)Hardy, Freeland, and Till]{hardy2006}
M.~R. Hardy, R.~K. Freeland, and M.~C. Till.
\newblock Validation of long-term equity return models for equity-linked
  guarantees.
\newblock \emph{N. Am. Actuar. J.}, 10\penalty0 (4):\penalty0 28--47, 2006.
\newblock ISSN 1092-0277.

\bibitem[Heber et~al.(2009)Heber, Lunde, Shephard, and
  Sheppard]{Realized_Vol_Data}
G.~Heber, A.~Lunde, N.~Shephard, and K.~Sheppard.
\newblock Oxford-man institute’s realized library.
\newblock \emph{Version 0.1, Oxford\&Man Institute, University of Oxford},
  2009.

\bibitem[Johnson et~al.(1994)Johnson, Kotz, and Balakrishnan]{johnson1994}
N.~L. Johnson, S.~Kotz, and N.~Balakrishnan.
\newblock \emph{Continuous univariate distributions. {V}ol. 1}.
\newblock Wiley Series in Probability and Mathematical Statistics: Applied
  Probability and Statistics. John Wiley \& Sons, Inc., New York, second
  edition, 1994.
\newblock ISBN 0-471-58495-9.
\newblock A Wiley-Interscience Publication.

\bibitem[Kalsi et~al.(2020)Kalsi, Lyons, and Arribas]{kalsi2020}
J.~Kalsi, T.~Lyons, and I.~P. Arribas.
\newblock Optimal execution with rough path signatures.
\newblock \emph{SIAM J. Financial Math.}, 11\penalty0 (2):\penalty0 470--493,
  2020.
\newblock ISSN 1945-497X.

\bibitem[Kidger and Lyons(2020)]{kidger2020}
P.~Kidger and T.~Lyons.
\newblock Signatory: differentiable computations of the signature and
  logsignature transforms, on both {CPU} and {GPU}.
\newblock \emph{International Conference on Learning Representations}, 2020.

\bibitem[Levin et~al.(2013)Levin, Lyons, and Ni]{levin2013learning}
D.~Levin, T.~Lyons, and H.~Ni.
\newblock Learning from the past, predicting the statistics for the future,
  learning an evolving system.
\newblock \emph{arXiv preprint arXiv:1309.0260}, 2013.

\bibitem[Liao et~al.(2019)Liao, Lyons, Yang, and Ni]{liao2019learning}
S.~Liao, T.~Lyons, W.~Yang, and H.~Ni.
\newblock Learning stochastic differential equations using rnn with log
  signature features.
\newblock \emph{arXiv preprint arXiv:1908.08286}, 2019.

\bibitem[Lin and Yang(2020)]{lin2020}
X.~S. Lin and S.~Yang.
\newblock Efficient dynamic hedging for large variable annuity portfolios with
  multiple underlying assets.
\newblock \emph{Astin Bull.}, 50\penalty0 (3):\penalty0 913--957, 2020.
\newblock ISSN 0515-0361.

\bibitem[Lyons(1998)]{lyons1998}
T.~Lyons.
\newblock Differential equations driven by rough signals.
\newblock \emph{Rev. Mat. Iberoamericana}, 14\penalty0 (2):\penalty0 215--310,
  1998.
\newblock ISSN 0213-2230.

\bibitem[Lyons et~al.(2007)Lyons, Caruana, and L\'{e}vy]{lyons2007differential}
T.~Lyons, M.~Caruana, and T.~L\'{e}vy.
\newblock \emph{Differential equations driven by rough paths}, volume 1908 of
  \emph{Lecture Notes in Mathematics}.
\newblock Springer, Berlin, 2007.
\newblock ISBN 978-3-540-71284-8; 3-540-71284-4.
\newblock Lectures from the 34th Summer School on Probability Theory held in
  Saint-Flour, July 6--24, 2004, With an introduction concerning the Summer
  School by Jean Picard.

\bibitem[Lyons et~al.(2020)Lyons, Nejad, and Perez~Arribas]{lyons2020}
T.~Lyons, S.~Nejad, and I.~Perez~Arribas.
\newblock Non-parametric pricing and hedging of exotic derivatives.
\newblock \emph{Appl. Math. Finance}, 27\penalty0 (6):\penalty0 457--494, 2020.
\newblock ISSN 1350-486X,1466-4313.

\bibitem[Ma et~al.(2023)Ma, Lu, and Chen]{ma2023}
J.~Ma, Z.~Lu, and D.~Chen.
\newblock Optimal reinsurance-investment with loss aversion under rough
  {H}eston model.
\newblock \emph{Quant. Finance}, 23\penalty0 (1):\penalty0 95--109, 2023.
\newblock ISSN 1469-7688,1469-7696.

\bibitem[Morrill et~al.(2020)Morrill, Fermanian, Kidger, and
  Lyons]{morrill2020generalised}
J.~Morrill, A.~Fermanian, P.~Kidger, and T.~Lyons.
\newblock A generalised signature method for multivariate time series feature
  extraction.
\newblock \emph{arXiv preprint arXiv:2006.00873}, 2020.

\bibitem[Ni et~al.(2021)Ni, Szpruch, Sabate-Vidales, Xiao, Wiese, and
  Liao]{ni2021sig}
H.~Ni, L.~Szpruch, M.~Sabate-Vidales, B.~Xiao, M.~Wiese, and S.~Liao.
\newblock Sig-wasserstein gans for time series generation.
\newblock In \emph{Proceedings of the Second ACM International Conference on AI
  in Finance}, pages 1--8, 2021.

\bibitem[Otero et~al.(2012)Otero, Dur{\'a}n, Fern{\'a}ndez, and
  Vivel]{otero2012estimating}
L.~Otero, P.~Dur{\'a}n, S.~Fern{\'a}ndez, and M.~Vivel.
\newblock Estimating insurers capital requirements through markov switching
  models in the solvency ii framework.
\newblock \emph{International Research Journal of Finance and Economics},
  86:\penalty0 20--38, 2012.

\bibitem[Perez~Arribas(2020)]{perez2020signatures}
I.~Perez~Arribas.
\newblock \emph{Signatures in machine learning and finance}.
\newblock PhD thesis, University of Oxford, 2020.

\bibitem[Petropoulos et~al.(2022)Petropoulos, Apiletti, Assimakopoulos, Babai,
  Barrow, Taieb, Bergmeir, Bessa, Bijak, Boylan,
  et~al.]{petropoulos2022forecasting}
F.~Petropoulos, D.~Apiletti, V.~Assimakopoulos, M.~Z. Babai, D.~K. Barrow,
  S.~B. Taieb, C.~Bergmeir, R.~J. Bessa, J.~Bijak, J.~E. Boylan, et~al.
\newblock Forecasting: theory and practice.
\newblock \emph{International Journal of Forecasting}, 38\penalty0
  (3):\penalty0 705--871, 2022.

\bibitem[Reizenstein and Graham(2020)]{reizenstein2020}
J.~F. Reizenstein and B.~Graham.
\newblock Algorithm 1004: the iisignature library: efficient calculation of
  iterated-integral signatures and log signatures.
\newblock \emph{ACM Trans. Math. Software}, 46\penalty0 (1):\penalty0 Art. 8,
  21, 2020.
\newblock ISSN 0098-3500,1557-7295.

\bibitem[Sen(1977)]{kumar1977}
P.~K. Sen.
\newblock Almost sure convergence of generalized {$U$}-statistics.
\newblock \emph{Ann. Probability}, 5\penalty0 (2):\penalty0 287--290, 1977.
\newblock ISSN 0091-1798.

\bibitem[Smith(1985)]{smith1985numerical}
G.~D. Smith.
\newblock \emph{Numerical solution of partial differential equations}.
\newblock Oxford Applied Mathematics and Computing Science Series. The
  Clarendon Press, Oxford University Press, New York, third edition, 1985.
\newblock ISBN 0-19-859641-3; 0-19-859650-2.
\newblock Finite difference methods.

\bibitem[{U.S. Bureau of Labor Statistics}(2023)]{BLS_Inflation_Data}
{U.S. Bureau of Labor Statistics}.
\newblock Consumer price index for all urban consumers, u.s. city average, all
  items.
\newblock \url{https://www.bls.gov/cpi/data.htm}, 2023.
\newblock Accessed on: 2023-01-27.

\bibitem[Wang et~al.(2023)Wang, Xiao, and Yu]{wang2023modeling}
X.~Wang, W.~Xiao, and J.~Yu.
\newblock Modeling and forecasting realized volatility with the fractional
  ornstein--uhlenbeck process.
\newblock \emph{Journal of Econometrics}, 232\penalty0 (2):\penalty0 389--415,
  2023.

\bibitem[Young(1936)]{young1936}
L.~C. Young.
\newblock An inequality of the {H}\"{o}lder type, connected with {S}tieltjes
  integration.
\newblock \emph{Acta Math.}, 67\penalty0 (1):\penalty0 251--282, 1936.
\newblock ISSN 0001-5962.

\bibitem[Zhu et~al.(2018)Zhu, Hardy, and Saunders]{zhu2018}
X.~Zhu, M.~R. Hardy, and D.~Saunders.
\newblock Dynamic hedging strategies for cash balance pension plans.
\newblock \emph{Astin Bull.}, 48\penalty0 (3):\penalty0 1245--1275, 2018.
\newblock ISSN 0515-0361.

\end{thebibliography}

\end{document}